\documentclass[runningheads]{llncs}

\usepackage{algorithm}
\usepackage[noend]{algpseudocode}
\usepackage{graphicx}
\usepackage{booktabs}
\usepackage{url}
\usepackage{topcapt}
\usepackage{array}
\usepackage{epsfig}
\usepackage{xcolor}
\usepackage{comment}
\usepackage{mathrsfs}
\usepackage{amsmath,amssymb} 
\usepackage{enumitem}
\usepackage{xifthen}
\setlist{nolistsep}

\newcommand{\mpara}[1]{\medskip\noindent{\bf{#1}}}

\newcommand{\dstname}[1]{\texttt{#1}}
\newcommand{\dstTZapX}{\dstname{3zap}}
\newcommand{\dstBugzX}{\dstname{bugzilla}}
\newcommand{\dstTZap}[1]{\dstname{3zap-{#1}}}
\newcommand{\dstBugz}[1]{\dstname{bugzilla-{#1}}}
\newcommand{\dstSacha}{\dstname{sacha}}
\newcommand{\dstSachaAR}[1]{\dstname{sacha-{#1}}}
\newcommand{\dstSachaG}[1]{\dstname{sacha-{\iAbs}-G{#1}}}
\newcommand{\dstSamba}{\dstname{samba}}
\newcommand{\dstUbi}{\dstname{ubiqLog}}
\newcommand{\dstUbiAR}[1]{\dstname{ubiqLog-{#1}}}
\newcommand{\dstUbiSAbs}[1]{\dstname{#1}}
\newcommand{\dstUbiSRel}[1]{\dstname{#1}}
\newcommand{\iRel}{rel}
\newcommand{\iAbs}{abs}

\newcommand{\algname}[1]{\textsc{#1}}
\newcommand{\algCoCo}{\algname{CoCo}}

\newcommand{\algCycles}{\algname{ExtractCycles}}
\newcommand{\algCyclesDP}{\algname{ExtractCyclesDP}}
\newcommand{\algCyclesFold}{\algname{ExtractCyclesTri}}
\newcommand{\algCombineV}{\algname{CombineVertically}}
\newcommand{\algCombineH}{\algname{CombineHorizontally}}
\newcommand{\algFilterTopKCov}{\algname{FilterCandidates}}
\newcommand{\algFilterFinal}{\algname{GreedyCover}}

\newcommand{\algGrowV}{\algname{GrowVertically}}
\newcommand{\algGrowH}{\algname{GrowHorizontally}}

\newcommand{\abs}[1]{\left\lvert#1\right\rvert}
\DeclareMathOperator{\omed}{med}

\newcommand{\optspan}{\Delta}
\newcommand{\ABC}{\Omega}

\newcommand{\len}[1]{\abs{#1}}
\newcommand{\tspan}[1]{\optspan(#1)}

\newcommand{\tSstart}{t_{\text{start}}}
\newcommand{\tSend}{t_{\text{end}}}

\newcommand{\ccycle}{\ensuremath{\mathcal{C}}}

\newcommand{\seq}[1][]{%
   {\ifthenelse{\isempty{#1}}%
     {\ensuremath{S}}
     {\ensuremath{S^{(#1)}}}
   }
}
\newcommand{\seqex}[2][]{%
   {\ifthenelse{\isempty{#1}}%
     {\ensuremath{S_#2}}
     {\ensuremath{S_#2^{(#1)}}}
   }
}

\newcommand{\cycle}{\ensuremath{C}}
\newcommand{\patt}{\ensuremath{P}}
\newcommand{\Ptree}{\ensuremath{T}}
\newcommand{\Pblock}{\ensuremath{B}}

\newcommand{\inDP}{\ensuremath{D}}

\newcommand{\Cev}{\ensuremath{\alpha}}
\newcommand{\Clen}{\ensuremath{r}}
\newcommand{\Cprd}{\ensuremath{p}}
\newcommand{\Cto}{\ensuremath{\tau}}
\newcommand{\Csc}{\ensuremath{E}}
\newcommand{\Ces}{\ensuremath{e}}

\newcommand{\cov}[1]{\mathit{cover}(#1)}
\newcommand{\residual}[1]{\mathit{residual}(#1)}

\newcommand{\nullC}{\textbf{0}}
\newcommand{\emptyC}{\llempty}

\newcommand{\lls}{\langle}
\newcommand{\lle}{\rangle}
\newcommand{\LL}[1]{\lls#1\lle}
\newcommand{\llempty}{\lls\lle}
\newcommand{\opconcat}{\oplus}
\newcommand{\fnconcat}{\bigoplus}

\newcommand{\sumel}[1]{\sigma(#1)}

\newcommand{\opchild}{\Gamma}
\newcommand{\opshift}{\mathit{shift}}
\newcommand{\opoccs}{\mathit{occs}}
\newcommand{\opmapOids}{\Theta}

\newcommand{\child}[1]{\opchild(#1)}
\newcommand{\shift}[2]{\opshift(#1, #2)}
\newcommand{\occs}[1]{\opoccs(#1)}
\newcommand{\occsStar}[1]{\opoccs^{*}(#1)}
\newcommand{\mapOids}[1]{\opmapOids(#1)}

\newcommand{\cume}[1]{\epsilon(#1)}

\newcommand{\Lchild}[1]{\gamma_\mathbf{L}(#1)}

\newcommand{\optspanRep}{\delta}
\newcommand{\tspanRep}[1]{\optspanRep(#1)}

\newcommand{\maxtspanStar}[1]{\optspan^{*}_{\max}(#1)}
\newcommand{\maxtspanRepStar}[1]{\optspanRep^{*}_{\max}(#1)}
\newcommand{\tspanStar}[1]{\optspan^{*}(#1)}
\newcommand{\tspanRepStar}[1]{\optspanRep^{*}(#1)}

\newcommand{\interd}[1]{d_{#1}}

\newcommand{\Bid}{X}

\newcommand{\evtseqfun}[1]{\zeta(#1)}
\newcommand{\evtseq}{A}
\newcommand{\lensfun}[1]{\rho(#1)}
\newcommand{\lens}{R}

\newcommand{\cl}{\mathit{L}}
\newcommand{\fr}{\mathit{fr}}
\newcommand{\unic}{l_{\mathbb{N}}}
\newcommand{\univ}[1]{l_{#1}}

\newcommand{\prcCl}{\%\cl}

\newcommand{\collF}{\ccycle_{F}}
\newcommand{\collV}{\ccycle_{V}}
\newcommand{\collH}{\ccycle_{H}}
\newcommand{\collVH}{\ccycle_{V\!+H}}
\newcommand{\collS}{\ccycle_{S}}

\newcommand{\resSet}{\mathcal{R}}

\newcommand{\ratioClR}{\cl\!:\!\resSet}
\newcommand{\clEmpty}{\cl(\emptyset, \seq)}
\newcommand{\clCC}{\cl(\ccycle, \seq)}

\newcommand{\nbTD}{m}
\newcommand{\nbV}{v}
\newcommand{\nbH}{h}
\newcommand{\nbS}{s}

\newcommand{\nbOTC}{c_{>3}}
\newcommand{\nbOmed}{c^{\text{M}}}
\newcommand{\nbOmax}{c^{+}}

\newcommand{\sepDG}[1]{\textbf{#1}}
\newcommand{\signDG}[1]{#1}
\newcommand{\valDG}[1]{\textit{#1}}

\newcommand{\BinfoRP}[2]{\textcolor{darkgray}{\{\Clen\!=#1, \Cprd\!=#2\}}}
\newcommand{\BinfoD}[1]{\textcolor{darkgray}{\,\text{--}\,#1\,\text{--}\,}}
\newcommand{\BinfoRPT}[2]{\textcolor{darkgray}{\{$\Clen\!=#1$, $\Cprd\!=#2$\}}}
\newcommand{\BinfoDT}[1]{\textcolor{darkgray}{\,--\,$#1$\,--\,}}
\newcommand{\Bstart}{\textcolor{darkgray}{\big(}}
\newcommand{\Bend}{\textcolor{darkgray}{\big)}}

\newcommand{\activity}[1]{#1}
\newcommand{\activityStart}[1]{[#1}
\newcommand{\activityEnd}[1]{#1]}
\newcommand{\activityIns}[1]{#1}

\usepackage{tikz}
\usetikzlibrary{decorations.pathreplacing}

\tikzset{
node distance=.8cm,
ghost node/.style={text width=0cm, inner sep=0pt}, 
root node/.style={circle, text width=0cm, inner sep=1.5pt, draw, align=center}, 
main node/.style={circle, text width=0cm, inner sep=1.5pt, draw, fill=black, align=center}, 
leaf node/.style={rectangle, text width=0cm, inner sep=1.5pt, draw, black, fill=black, align=center}, 
label node/.style={text width=0.5cm, inner sep=2pt, align=center, node distance=.5cm, xshift=.3cm},
lterm node/.style={text width=0.5cm, inner sep=2pt, align=center, node distance=.3cm},
prop node/.style={text width=0.5cm, inner sep=2pt, align=center, node distance=.3cm, xshift=.2cm , color=darkgray},
pterm node/.style={text width=0.5cm, inner sep=2pt, align=center, node distance=.3cm, xshift=.2cm, color=darkgray},
cprop node/.style={text width=0.5cm, inner sep=2pt, align=center, node distance=.3cm, color=darkgray},
cpterm node/.style={text width=0.5cm, inner sep=2pt, align=center, node distance=.3cm, color=darkgray},
rep edge/.style={black, very thick},
lghost node/.style={text width=0cm, inner sep=0pt, color=white, opacity=0}, 
lmain node/.style={inner sep=1pt, align=center, node distance=.35cm, fill=white, anchor=center, font=\small},
li node/.style={inner sep=0pt, anchor=west, xshift=.08cm, yshift=.15cm, node distance=0cm, fill=white, font=\small},
le node/.style={text width=1.9cm, inner sep=1pt, align=center, node distance=0cm, rotate=60, anchor=east, xshift=-1.2cm, font=\small},
lleaf node/.style={text width=1.2cm, inner sep=1pt, align=center, node distance=0cm, rotate=60, anchor=east, xshift=-.8cm, font=\small},
ptstart node/.style={inner sep=2pt, node distance=.12cm, align=left, anchor=west, font=\small\bf},
ptend node/.style={inner sep=2pt, node distance=.12cm, align=left, anchor=west, font=\small\bf},
dst node/.style={inner sep=1pt, align=left, text=darkgray, node distance=.05cm, anchor=west, font=\small},
info node/.style={inner sep=1pt, align=left, text=darkgray, anchor=west, font=\small},
info rec/.style={very thin, lightgray, fill=vlgray, opacity=0.25, rounded corners}
}

\newcommand{\BlockMark}[1]{$\Pblock_{#1}$}
\newcommand{\OccMark}[1]{#1}

\newcommand{\TreeOneEvtOneCyclVar}[4]{%
  \node[main node] (#1b0) #2 {};
  \node[leaf node] (#1b1) [below of=#1b0] {};
  \node[label node] (#1l0) [above of=#1b0, xshift=-0.2cm] {\BlockMark{0}};
  \node[lterm node] (#1l1) [below of=#1b1] {\BlockMark{1}};
  \node[prop node] (#1p0) [below of=#1l0] {#3};
  \node[pterm node] (#1p1) [below of=#1l1] {#4};
  \path
    (#1b0) edge (#1b1);
}

\newcommand{\TreeOneEvtTwoCyclVar}[5]{%
  \node[main node] (#1b0) #2 {};
  \node[main node] (#1b1) [below of=#1b0] {};
  \node[leaf node] (#1b11) [below of=#1b1] {};
  \node[label node] (#1l0) [above of=#1b0, xshift=-0.2cm] {\BlockMark{0}};
  \node[label node] (#1l1) [above of=#1b1] {\BlockMark{1}};
  \node[lterm node] (#1l11) [below of=#1b11] {\BlockMark{11}};
  \node[prop node] (#1p0) [below of=#1l0] {#3};
  \node[prop node] (#1p1) [below of=#1l1] {#4};
  \node[pterm node] (#1p11) [below of=#1l11] {#5};
  \path
    (#1b0) edge (#1b1)
    (#1b1) edge (#1b11);
}

\newcommand{\TreeThreeEvtOneCyclVar}[8]{%
  \node[main node] (#1b0) #2 {};
  \node[leaf node] (#1b2) [below of=#1b0] {};
  \node[leaf node] (#1b1) [left of=#1b2, node distance=1.2cm] {};
  \node[leaf node] (#1b3) [right of=#1b2, node distance=1.2cm] {};

  \node[label node] (#1l0) [above of=#1b0, xshift=-0.2cm] {\BlockMark{0}};
  \node[lterm node] (#1l1) [below of=#1b1] {\BlockMark{1}};
  \node[lterm node] (#1l2) [below of=#1b2] {\BlockMark{2}};
  \node[lterm node] (#1l3) [below of=#1b3] {\BlockMark{3}};

  \node[prop node] (#1p0) [below of=#1l0] {#3};
  \node[pterm node] (#1p1) [below of=#1l1] {#4};
  \node[pterm node] (#1p2) [below of=#1l2] {#6};
  \node[pterm node] (#1p3) [below of=#1l3] {#8};
  \path
    (#1b0) edge (#1b1)
    (#1b0) edge (#1b2)
    (#1b0) edge (#1b3);

\path[dotted, thin, ->, bend right=20, color=darkgray]
    (#1b1) edge node[above] {#5} (#1b2)
    (#1b2) edge node[above] {#7} (#1b3);
}

\newcommand{\TreeThreeEvtTwoCyclVar}[9]{%
  \node[main node] (#1b0) #2 {};
  \node[main node] (#1b2) [below of=#1b0] {};
  \node[leaf node] (#1b1) [left of=#1b2, node distance=1.2cm] {};
  \node[leaf node] (#1b3) [right of=#1b2, node distance=1.2cm] {};
  \node[leaf node] (#1b21) [below of=#1b2] {};

  \node[label node] (#1l0) [above of=#1b0, xshift=-0.2cm] {\BlockMark{0}};
  \node[lterm node] (#1l2) [below of=#1b2, xshift=.3cm] {\BlockMark{2}};
  \node[lterm node] (#1l1) [below of=#1b1] {\BlockMark{1}};
  \node[lterm node] (#1l3) [below of=#1b3] {\BlockMark{3}};
  \node[lterm node] (#1l21) [below of=#1b21] {\BlockMark{21}};

  \node[prop node] (#1p0) [below of=#1l0] {#3};
  \node[prop node] (#1p2) [below of=#1l2] {#9};
  \node[pterm node] (#1p1) [below of=#1l1] {#4};
  \node[pterm node] (#1p3) [below of=#1l3] {#8};
  \node[pterm node] (#1p21) [below of=#1l21] {#6};
  \path
    (#1b0) edge (#1b1)
    (#1b0) edge (#1b2)
    (#1b0) edge (#1b3)
    (#1b2) edge (#1b21);

\path[dotted, thin, ->, bend right=20, color=darkgray]
    (#1b1) edge node[above] {#5} (#1b2)
    (#1b2) edge node[above] {#7} (#1b3);
}

\newcommand{\TreeThreeEvtThreeCyclVar}[9]{%
  \node[main node] (#1b0) #2 {};
  \node[main node] (#1b-1) [above of=#1b0] {};
  \node[main node] (#1b2) [below of=#1b0] {};
  \node[leaf node] (#1b1) [left of=#1b2, node distance=1.2cm] {};
  \node[leaf node] (#1b3) [right of=#1b2, node distance=1.2cm] {};
  \node[leaf node] (#1b21) [below of=#1b2] {};

  \node[label node] (#1l-1) [above of=#1b-1, xshift=-0.2cm] {\BlockMark{0}};
  \node[label node] (#1l0) [above of=#1b0] {\BlockMark{1}};
  \node[lterm node] (#1l2) [below of=#1b2, xshift=.3cm] {\BlockMark{12}};
  \node[lterm node] (#1l1) [below of=#1b1] {\BlockMark{11}};
  \node[lterm node] (#1l3) [below of=#1b3] {\BlockMark{13}};
  \node[lterm node] (#1l21) [below of=#1b21] {\BlockMark{121}};

  \node[prop node] (#1p-1) [below of=#1l-1] {$2,33$};
  \node[prop node] (#1p0) [below of=#1l0] {#3};
  \node[prop node] (#1p2) [below of=#1l2] {#9};
  \node[pterm node] (#1p1) [below of=#1l1] {#4};
  \node[pterm node] (#1p3) [below of=#1l3] {#8};
  \node[pterm node] (#1p21) [below of=#1l21] {#6};
  \path
    (#1b-1) edge (#1b0)
    (#1b0) edge (#1b1)
    (#1b0) edge (#1b2)
    (#1b0) edge (#1b3)
    (#1b2) edge (#1b21);

\path[dotted, thin, ->, bend right=20, color=darkgray]
    (#1b1) edge node[above] {#5} (#1b2)
    (#1b2) edge node[above] {#7} (#1b3);
}

\newcommand{\rootf}{.5}

\newcommand{\LowestBranchA}[4]{%
  \node[main node] (#1) #2 {};
  \node[leaf node] (#1b1) [below of=#1] {};
  \node[cprop node] (#1l0) [above of=#1, xshift=-0.2cm] {#3};
  \node[cpterm node] (#1l1) [below of=#1b1] {a};
  \path
    (#1) edge (#1b1);
}

\newcommand{\LowestBranchB}[4]{%
  \node[main node] (#1) #2 {};
  \node[leaf node] (#1b2) [below of=#1] {};
  \node[leaf node] (#1b1) [left of=#1b2, node distance=.6cm] {};
  \node[leaf node] (#1b3) [right of=#1b2, node distance=.6cm] {};

  \node[cprop node] (#1l0) [above of=#1, xshift=-0.2cm] {#3};
  \node[cpterm node] (#1l1) [below of=#1b1] {b};
  \node[cpterm node] (#1l2) [below of=#1b2] {a};
  \node[cpterm node] (#1l3) [below of=#1b3] {c};

  \path
    (#1) edge (#1b1)
    (#1) edge (#1b2)
    (#1) edge (#1b3);
}

\newcommand{\LowestBranchC}[4]{%
  \node[main node] (#1) #2 {};
  \node[main node] (XYZ) [below of=#1] {};
  \node[main node] (#1b2) [below of=#1] {};
  \node[leaf node] (#1b1) [left of=#1b2, node distance=.6cm] {};
  \node[leaf node] (#1b3) [right of=#1b2, node distance=.6cm] {};

  \node[cprop node] (#1l0) [above of=#1, xshift=-0.2cm] {#3};
  \node[cpterm node] (#1l1) [below of=#1b1] {b};
  \node[cpterm node] (#1l3) [below of=#1b3] {c};

  \path
    (#1) edge (#1b1)
    (#1) edge (#1b2)
    (#1) edge (#1b3);

  \foreach \i in {1,...,4}{
    \LowestBranchA{XYZ\i}{ at ([shift={({(\i-\rootf)*.2},{(\i-\rootf)*-.3})}]XYZ)}{#3}{#4}
  }
  \path[rep edge] (XYZ) edge (XYZ4);
}

\newcommand{\LowestBranch}[4]{\LowestBranchA{#1}{#2}{#3}{#4}}
\newcommand{\xslant}{.5}
\newcommand{\yslant}{-.3}
\newcommand{\Xslant}{.5}
\newcommand{\Yslant}{-.3}

\newcommand{\LowestRep}[4]{%
  \path[rep edge] ({#2},{#3}) edge ({#2+(#4-\rootf)*\xslant},{#3+(#4-\rootf)*\yslant});
  \node[main node] (R#1b0) at ({#2},{#3}) {};
  \foreach \i in {1,...,#4}{
    \LowestBranch{R#1l\i}{at ({#2+(\i-\rootf)*\xslant},{#3+(\i-\rootf)*\yslant})}{}{a}
  }
}

\newcommand{\xsctm}{.47}
\newcommand{\timelinePFive}[1]{%
\begin{tikzpicture}[-,auto,xscale=\xsctm]
\newcommand{\xeAa}{0}
\newcommand{\xeAb}{3}
\newcommand{\xeAc}{4}
\newcommand{\xeBa}{13}
\newcommand{\xeBb}{16}
\newcommand{\xeBc}{17}
\newcommand{\xeCa}{26}
\newcommand{\xeCb}{29}
\newcommand{\xeCc}{30}
\draw[white] (-1,-2) rectangle (35,2);
\node[anchor=north east] at (35,2) {#1};

\draw[lightgray] (0,0) -- (34,0);
\foreach \i in {0,...,34}{
\draw[lightgray] (\i,-.1) -- (\i,.1);
}
\foreach \v/\i in {b/\xeAa,a/\xeAb,c/\xeAc, b/\xeBa,a/\xeBb,c/\xeBc, b/\xeCa,a/\xeCb,c/\xeCc}{
\node[anchor=south,text height=2em] at (\i,.1) {$\v$};
\node[font=\scriptsize] at (\i,-.3) {$\i$};
}
\path[->,>=stealth,shorten <= 1pt]
(\xeAa,-1.2) edge[bend right=3] node[below] {$\Cprd_0$} (\xeBa,-1.2)
(\xeBa,-1.2) edge[bend right=3] node[below] {$\Cprd_0$} (\xeCa,-1.2)
(\xeAa,-.5) edge[bend right=6] node[below] {$d_{12}$} (\xeAb,-.5)
(\xeAb,-.5) edge[bend right=6] node[below] {$d_{13}$} (\xeAc,-.5)
(\xeBa,-.5) edge[bend right=6] node[below] {$d_{12}$} (\xeBb,-.5)
(\xeBb,-.5) edge[bend right=6] node[below] {$d_{13}$} (\xeBc,-.5)
(\xeCa,-.5) edge[bend right=6] node[below] {$d_{12}$} (\xeCb,-.5)
(\xeCb,-.5) edge[bend right=6] node[below] {$d_{13}$} (\xeCc,-.5);

\draw [decorate,decoration={brace,amplitude=6pt},xshift=0pt,yshift=0pt]
(\xeAa,1.2) -- (\xeCc,1.2) node [black,midway,yshift=0.15cm] 
{\footnotesize $\tspanStar{\Pblock_0}$};
\draw [decorate,decoration={brace,amplitude=6pt},xshift=0pt,yshift=0pt]
(\xeAa,.6) -- (\xeAc,.6) node [black,midway,yshift=0.15cm] 
{\footnotesize $\tspanRepStar{\Pblock_0}$};
\draw [decorate,decoration={brace,amplitude=6pt},xshift=0pt,yshift=0pt]
(\xeBa,.6) -- (\xeBc,.6) node [black,midway,yshift=0.15cm] 
{\footnotesize $\tspanRepStar{\Pblock_0}$};
\draw [decorate,decoration={brace,amplitude=6pt},xshift=0pt,yshift=0pt]
(\xeCa,.6) -- (\xeCc,.6) node [black,midway,yshift=0.15cm] 
{\footnotesize $\tspanRepStar{\Pblock_0}$};
\end{tikzpicture}
}

\newcommand{\timelineQSixOne}[1]{%
\begin{tikzpicture}[-,auto,xscale=\xsctm]
\newcommand{\xeAa}{2}
\newcommand{\xeAb}{5}
\newcommand{\xeAc}{7}
\newcommand{\xeBa}{13}
\newcommand{\xeBb}{18}
\newcommand{\xeBc}{21}
\newcommand{\xeCa}{26}
\newcommand{\xeCb}{30}
\newcommand{\xeCc}{31}
\draw[white] (-1,-2) rectangle (35,2);
\node[anchor=north east] at (35,2) {#1};

\draw[lightgray] (0,0) -- (34,0);
\foreach \i in {0,...,34}{
\draw[lightgray] (\i,-.1) -- (\i,.1);
}
\foreach \v/\i in {b/\xeAa,a/\xeAb,c/\xeAc, b/\xeBa,a/\xeBb,c/\xeBc, b/\xeCa,a/\xeCb,c/\xeCc}{
\node[anchor=south,text height=2em] at (\i,.1) {$\v$};
\node[font=\scriptsize] at (\i,-.3) {$\i$};
}
\path[->,>=stealth,shorten <= 1pt]
(0,-1.2) edge[bend right=3] node[below] {$\Cto$} (\xeAa,-1.2)
(\xeAa,-1.2) edge[bend right=3] node[below] {$\Cprd_0+e_3$} (\xeBa,-1.2)
(\xeBa,-1.2) edge[bend right=3] node[below] {$\Cprd_0+e_6$} (\xeCa,-1.2)
(\xeAa,-.5) edge[bend right=6] node[below] {$d_{12}+e_1$} (\xeAb,-.5)
(\xeAb,-.5) edge[bend right=6] node[below] {$d_{13}+e_2$} (\xeAc,-.5)
(\xeBa,-.5) edge[bend right=6] node[below] {$d_{12}+e_4$} (\xeBb,-.5)
(\xeBb,-.5) edge[bend right=6] node[below] {$d_{13}+e_5$} (\xeBc,-.5)
(\xeCa,-.5) edge[bend right=6] node[below] {$d_{12}+e_7$} (\xeCb,-.5)
(\xeCb,-.5) edge[bend right=6] node[below] {$d_{13}+e_8$} (\xeCc,-.5);

\draw [decorate,decoration={brace,amplitude=6pt},xshift=0pt,yshift=0pt]
(\xeAa,1.2) -- (\xeCc,1.2) node [black,midway,yshift=0.15cm] 
{\footnotesize $\tspan{(\Pblock_0, \LL{})}$};
\draw [decorate,decoration={brace,amplitude=6pt},xshift=0pt,yshift=0pt]
(\xeAa,.6) -- (\xeAc,.6) node [black,midway,yshift=0.15cm] 
{\footnotesize $\tspanRep{(\Pblock_0, \LL{1})}$};
\draw [decorate,decoration={brace,amplitude=6pt},xshift=0pt,yshift=0pt]
(\xeBa,.6) -- (\xeBc,.6) node [black,midway,yshift=0.15cm] 
{\footnotesize $\tspanRep{(\Pblock_0, \LL{2})}$};
\draw [decorate,decoration={brace,amplitude=6pt},xshift=0pt,yshift=0pt]
(\xeCa,.6) -- (\xeCc,.6) node [black,midway,yshift=0.15cm] 
{\footnotesize $\tspanRep{(\Pblock_0, \LL{3})}$};
\end{tikzpicture}
}

\newcommand{\timelinePThree}[1]{%
\begin{tikzpicture}[-,auto,xscale=\xsctm]
\newcommand{\xeAa}{0}
\newcommand{\xeAb}{2}
\newcommand{\xeAc}{4}
\newcommand{\xeAd}{6}
\newcommand{\xeBa}{13}
\newcommand{\xeBb}{15}
\newcommand{\xeBc}{17}
\newcommand{\xeBd}{19}
\newcommand{\xeCa}{26}
\newcommand{\xeCb}{28}
\newcommand{\xeCc}{30}
\newcommand{\xeCd}{32}
\draw[white] (-1,-2) rectangle (35,2);
\node[anchor=north east] at (35,2) {#1};

\draw[lightgray] (0,0) -- (34,0);
\foreach \i in {0,...,34}{
\draw[lightgray] (\i,-.1) -- (\i,.1);
}
\foreach \v/\i in {a/\xeAa,a/\xeAb,a/\xeAc,a/\xeAd, a/\xeBa,a/\xeBb,a/\xeBc,a/\xeBd, a/\xeCa,a/\xeCb,a/\xeCc,a/\xeCd}{
\node[anchor=south,text height=2em] at (\i,.1) {$\v$};
\node[font=\scriptsize] at (\i,-.3) {$\i$};
}
\path[->,>=stealth,shorten <= 1pt]
(\xeAa,-1.2) edge[bend right=3] node[below] {$\Cprd_0$} (\xeBa,-1.2)
(\xeBa,-1.2) edge[bend right=3] node[below] {$\Cprd_0$} (\xeCa,-1.2)
(\xeAa,-.5) edge[bend right=6] node[below] {$\Cprd_{1}$} (\xeAb,-.5)
(\xeAb,-.5) edge[bend right=6] node[below] {$\Cprd_{1}$} (\xeAc,-.5)
(\xeAc,-.5) edge[bend right=6] node[below] {$\Cprd_{1}$} (\xeAd,-.5)
(\xeBa,-.5) edge[bend right=6] node[below] {$\Cprd_{1}$} (\xeBb,-.5)
(\xeBb,-.5) edge[bend right=6] node[below] {$\Cprd_{1}$} (\xeBc,-.5)
(\xeBc,-.5) edge[bend right=6] node[below] {$\Cprd_{1}$} (\xeBd,-.5)
(\xeCa,-.5) edge[bend right=6] node[below] {$\Cprd_{1}$} (\xeCb,-.5)
(\xeCb,-.5) edge[bend right=6] node[below] {$\Cprd_{1}$} (\xeCc,-.5)
(\xeCc,-.5) edge[bend right=6] node[below] {$\Cprd_{1}$} (\xeCd,-.5);

\draw [decorate,decoration={brace,amplitude=6pt},xshift=0pt,yshift=0pt]
(\xeAa,1.2) -- (\xeCd,1.2) node [black,midway,yshift=0.15cm] 
{\footnotesize $\tspanStar{\Pblock_0}$};
\draw [decorate,decoration={brace,amplitude=6pt},xshift=0pt,yshift=0pt]
(\xeAa,.6) -- (\xeAd,.6) node [black,midway,yshift=0.15cm] 
{\footnotesize $\tspanRepStar{\Pblock_0}$};
\draw [decorate,decoration={brace,amplitude=6pt},xshift=0pt,yshift=0pt]
(\xeBa,.6) -- (\xeBd,.6) node [black,midway,yshift=0.15cm] 
{\footnotesize $\tspanRepStar{\Pblock_0}$};
\draw [decorate,decoration={brace,amplitude=6pt},xshift=0pt,yshift=0pt]
(\xeCa,.6) -- (\xeCd,.6) node [black,midway,yshift=0.15cm] 
{\footnotesize $\tspanRepStar{\Pblock_0}$};
\end{tikzpicture}
}

\newcommand{\timelineQThreeOne}[1]{%
\begin{tikzpicture}[-,auto,xscale=\xsctm]
\newcommand{\xeAa}{2}
\newcommand{\xeAb}{5}
\newcommand{\xeAc}{7}
\newcommand{\xeAd}{8}
\newcommand{\xeBa}{13}
\newcommand{\xeBb}{15}
\newcommand{\xeBc}{20}
\newcommand{\xeBd}{21}
\newcommand{\xeCa}{26}
\newcommand{\xeCb}{29}
\newcommand{\xeCc}{32}
\newcommand{\xeCd}{33}
\draw[white] (-1,-2) rectangle (35,2);
\node[anchor=north east] at (35,2) {#1};

\draw[lightgray] (0,0) -- (34,0);
\foreach \i in {0,...,34}{
\draw[lightgray] (\i,-.1) -- (\i,.1);
}
\foreach \v/\i in {a/\xeAa,a/\xeAb,a/\xeAc,a/\xeAd, a/\xeBa,a/\xeBb,a/\xeBc,a/\xeBd, a/\xeCa,a/\xeCb,a/\xeCc,a/\xeCd}{
\node[anchor=south,text height=2em] at (\i,.1) {$\v$};
\node[font=\scriptsize] at (\i,-.3) {$\i$};
}
\path[->,>=stealth,shorten <= 1pt]
(0,-1.2) edge[bend right=3] node[below] {$\Cto$} (\xeAa,-1.2)
(\xeAa,-1.2) edge[bend right=3] node[below] {$\Cprd_0+e_4$} (\xeBa,-1.2)
(\xeBa,-1.2) edge[bend right=3] node[below] {$\Cprd_0+e_8$} (\xeCa,-1.2)
(\xeAa,-.5) edge[bend right=6] node[below] {$\Cprd_{1}+e_1$} (\xeAb,-.5)
(\xeAb,-.5) edge[bend right=6] node[below] {$\Cprd_{1}+e_2$} (\xeAc,-.5)
(\xeAc,-.5) edge[bend right=6] node[below, xshift=15pt] {$\Cprd_{1}+e_3$} (\xeAd,-.5)
(\xeBa,-.5) edge[bend right=6] node[below] {$\Cprd_{1}+e_5$} (\xeBb,-.5)
(\xeBb,-.5) edge[bend right=6] node[below] {$\Cprd_{1}+e_6$} (\xeBc,-.5)
(\xeBc,-.5) edge[bend right=6] node[below] {$\Cprd_{1}+e_7$} (\xeBd,-.5)
(\xeCa,-.5) edge[bend right=6] node[below] {$\Cprd_{1}+e_9$} (\xeCb,-.5)
(\xeCb,-.5) edge[bend right=6] node[below] {$\Cprd_{1}+e_{10}$} (\xeCc,-.5)
(\xeCc,-.5) edge[bend right=6] node[below, xshift=10pt] {$\Cprd_{1}+e_{11}$} (\xeCd,-.5);

\draw [decorate,decoration={brace,amplitude=6pt},xshift=0pt,yshift=0pt]
(\xeAa,1.2) -- (\xeCd,1.2) node [black,midway,yshift=0.15cm] 
{\footnotesize $\tspan{(\Pblock_0, \LL{})}$};
\draw [decorate,decoration={brace,amplitude=6pt},xshift=0pt,yshift=0pt]
(\xeAa,.6) -- (\xeAd,.6) node [black,midway,yshift=0.15cm] 
{\footnotesize $\tspanRep{(\Pblock_0, \LL{1})}$};
\draw [decorate,decoration={brace,amplitude=6pt},xshift=0pt,yshift=0pt]
(\xeBa,.6) -- (\xeBd,.6) node [black,midway,yshift=0.15cm] 
{\footnotesize $\tspanRep{(\Pblock_0, \LL{2})}$};
\draw [decorate,decoration={brace,amplitude=6pt},xshift=0pt,yshift=0pt]
(\xeCa,.6) -- (\xeCd,.6) node [black,midway,yshift=0.15cm] 
{\footnotesize $\tspanRep{(\Pblock_0, \LL{3})}$};
\end{tikzpicture}
}

\newcommand{\timelinePFour}[1]{%
\begin{tikzpicture}[-,auto,xscale=\xsctm]
\newcommand{\xeAa}{0}
\newcommand{\xeAb}{2}
\newcommand{\xeAc}{4}
\newcommand{\xeAd}{6}
\newcommand{\xeBa}{13}
\newcommand{\xeBb}{15}
\newcommand{\xeBc}{17}
\newcommand{\xeBd}{19}
\newcommand{\xeCa}{26}
\newcommand{\xeCb}{28}
\newcommand{\xeCc}{30}
\newcommand{\xeCd}{32}
\draw[white] (-1,-2) rectangle (35,2);
\node[anchor=north east] at (35,2) {#1};

\draw[lightgray] (0,0) -- (34,0);
\foreach \i in {0,...,34}{
\draw[lightgray] (\i,-.1) -- (\i,.1);
}
\foreach \v/\i in {a/\xeAa,a/\xeAb,a/\xeAc,a/\xeAd, a/\xeBa,a/\xeBb,a/\xeBc,a/\xeBd, a/\xeCa,a/\xeCb,a/\xeCc,a/\xeCd}{
\node[anchor=south,text height=2em] at (\i,.1) {$\v$};
\node[font=\scriptsize] at (\i,-.3) {$\i$};
}
\path[->,>=stealth,shorten <= 1pt]
(\xeAa,-1.2) edge[bend right=3] node[below] {$\Cprd_0$} (\xeAb,-1.2)
(\xeAb,-1.2) edge[bend right=3] node[below] {$\Cprd_0$} (\xeAc,-1.2)
(\xeAc,-1.2) edge[bend right=3] node[below] {$\Cprd_0$} (\xeAd,-1.2)
(\xeAa,-.5) edge[bend right=2] node[below, pos=.08] {$\Cprd_{1}$} (\xeBa,-.5)
(\xeBa,-.5) edge[bend right=14] node[below, pos=.5] {$\Cprd_{1}$} (\xeCa,-.5)
(\xeAb,-.5) edge[bend right=6] node[below, pos=.16] {$\Cprd_{1}$} (\xeBb,-.5)
(\xeBb,-.5) edge[bend right=10] node[below, pos=.76] {$\Cprd_{1}$} (\xeCb,-.5)
(\xeAc,-.5) edge[bend right=10] node[below, pos=.24] {$\Cprd_{1}$} (\xeBc,-.5)
(\xeBc,-.5) edge[bend right=6] node[below,pos=0.84] {$\Cprd_{1}$} (\xeCc,-.5)
(\xeAd,-.5) edge[bend right=14] node[below, pos=.5] {$\Cprd_{1}$} (\xeBd,-.5)
(\xeBd,-.5) edge[bend right=2] node[below, pos=0.92] {$\Cprd_{1}$} (\xeCd,-.5);

\draw [decorate,decoration={brace,amplitude=6pt},xshift=0pt,yshift=0pt]
(\xeAa,1.3) -- (\xeCd,1.3) node [black,midway,yshift=0.15cm] 
{\footnotesize $\tspanStar{\Pblock_0}$};

\draw [line width=2pt, white, decorate,decoration={brace,amplitude=6pt},xshift=0pt,yshift=0pt]
(\xeAd,.75) -- (\xeCd,.75);
\draw [decorate,decoration={brace,amplitude=6pt},xshift=0pt,yshift=0pt]
(\xeAd,.75) -- (\xeCd,.75) node [black,midway,yshift=0.15cm] 
{\footnotesize $\tspanRepStar{\Pblock_0}$};

\draw [line width=2pt, white, decorate,decoration={brace,amplitude=6pt},xshift=0pt,yshift=0pt]
(\xeAc,.7) -- (\xeCc,.7); 
\draw [decorate,decoration={brace,amplitude=6pt},xshift=0pt,yshift=0pt]
(\xeAc,.7) -- (\xeCc,.7) node [black,midway,yshift=0.15cm] 
{\footnotesize $\tspanRepStar{\Pblock_0}$};

\draw [line width=2pt, white, decorate,decoration={brace,amplitude=6pt},xshift=0pt,yshift=0pt]
(\xeAb,.65) -- (\xeCb,.65);
\draw [decorate,decoration={brace,amplitude=6pt},xshift=0pt,yshift=0pt]
(\xeAb,.65) -- (\xeCb,.65) node [black,midway,yshift=0.15cm] 
{\footnotesize $\tspanRepStar{\Pblock_0}$};

\draw [line width=2pt, white, decorate,decoration={brace,amplitude=6pt},xshift=0pt,yshift=0pt]
(\xeAa,.6) -- (\xeCa,.6);
\draw [decorate,decoration={brace,amplitude=6pt},xshift=0pt,yshift=0pt]
(\xeAa,.6) -- (\xeCa,.6) node [black,midway,yshift=0.15cm] 
{\footnotesize $\tspanRepStar{\Pblock_0}$};
\end{tikzpicture}
}

\newcommand{\timelinePEight}[1]{%
\begin{tikzpicture}[-,auto,xscale=\xsctm]
\newcommand{\yeA}{-2.0}
\newcommand{\yeB}{-1.5}
\newcommand{\yeC}{-1.}
\newcommand{\yeD}{-.5}
\newcommand{\yoA}{2.2}
\newcommand{\yoB}{1.4}
\newcommand{\yoC}{.7}
\newcommand{\yoD}{.5}
\draw[white] (-.2,-3.) rectangle (35.3,3.2);
\node[anchor=north east] at (35,3.2) {#1};
\begin{scope}
\clip (-.2,-2.9) rectangle (35.2,3.1);
\draw[lightgray] (0,0) -- (40,0);
\foreach \i in {0,...,40}{
\draw[lightgray] (\i,-.1) -- (\i,.1);
}
\foreach \x/\v in {0/b,3/a,4/a,5/a,6/a,8/c,10/b,13/a,14/a,15/a,16/a,18/c,20/b,23/a,24/a,25/a,26/a,28/c,33/b}{
\node[anchor=south,text height=2em] at (\x,.1) {$\v$};
\node[font=\scriptsize] at (\x,-.3) {$\x$};
}

\path[->,>=stealth,shorten <= 1pt]
(0,\yeA) edge[bend right=2] node[below, pos=0.025] {$\Cprd_0$} (33,\yeA)
(0,\yeB) edge[bend right=5] node[below, pos=0.1] {$\Cprd_1$} (10,\yeB)
(10,\yeB) edge[bend right=5] node[below] {$\Cprd_1$} (20,\yeB)
(0,\yeC) edge[bend right=7] node[below] {$d_{12}$} (3,\yeC)
(3,\yeC) edge[bend right=7] node[below] {$d_{13}$} (8,\yeC)
(3,\yeD) edge[bend right=10] node[below] {$\Cprd_{12}$} (4,\yeD)
(4,\yeD) edge[bend right=10] node[below, xshift=3pt] {$\Cprd_{12}$} (5,\yeD)
(5,\yeD) edge[bend right=10] node[below, xshift=6pt] {$\Cprd_{12}$} (6,\yeD)

(10,\yeC) edge[bend right=7] node[below] {$d_{12}$} (13,\yeC)
(13,\yeC) edge[bend right=7] node[below] {$d_{13}$} (18,\yeC)
(13,\yeD) edge[bend right=10] node[below] {$\Cprd_{12}$} (14,\yeD)
(14,\yeD) edge[bend right=10] node[below, xshift=3pt] {$\Cprd_{12}$} (15,\yeD)
(15,\yeD) edge[bend right=10] node[below, xshift=6pt] {$\Cprd_{12}$} (16,\yeD)
(20,\yeC) edge[bend right=7] node[below] {$d_{12}$} (23,\yeC)
(23,\yeC) edge[bend right=7] node[below] {$d_{13}$} (28,\yeC)
(23,\yeD) edge[bend right=10] node[below] {$\Cprd_{12}$} (24,\yeD)
(24,\yeD) edge[bend right=10] node[below, xshift=3pt] {$\Cprd_{12}$} (25,\yeD)
(25,\yeD) edge[bend right=10] node[below, xshift=6pt] {$\Cprd_{12}$} (26,\yeD)
(33,\yeA) edge[bend right=2] node[below, pos=0.025] {$\Cprd_0$} (66,\yeA)
(33,\yeB) edge[bend right=5] node[below, pos=0.1] {$\Cprd_1$} (43,\yeB)
(33,\yeC) edge[bend right=7] node[below] {$d_{12}$} (36,\yeC)
;

\draw [decorate,decoration={brace,amplitude=6pt},xshift=0pt,yshift=0pt]
(0,{\yoA+.0}) -- (55,{\yoA+.0}) node [black,midway,yshift=0.15cm] 
{\footnotesize $\tspanStar{\Pblock_{0}}$};

\draw [decorate,decoration={brace,amplitude=6pt},xshift=0pt,yshift=0pt]
(0,{\yoB+.05}) -- (33,{\yoB+.05}) node [black,midway,yshift=0.15cm, anchor=south west, xshift=-0.45cm] 
{\footnotesize $\maxtspanStar{\Pblock_{1}} = \maxtspanRepStar{\Pblock_{0}}$};
\draw [line width=2pt, white, decorate,decoration={brace,amplitude=6pt},xshift=0pt,yshift=0pt]
(0,{\yoB+.0}) -- (28,{\yoB+.0});
\draw [decorate,decoration={brace,amplitude=6pt},xshift=0pt,yshift=0pt]
(0,{\yoB+.0}) -- (28,{\yoB+.0}) node [black,midway,yshift=0.15cm, anchor=south east,xshift=0.45cm] 
{\footnotesize $\tspanStar{\Pblock_{1}} = \tspanRepStar{\Pblock_{0}}$};

\draw [decorate,decoration={brace,amplitude=6pt},xshift=0pt,yshift=0pt]
(10,{\yoC+.05}) -- (20,{\yoC+.05}) node [black,midway,yshift=0.15cm,xshift=0.48cm] 
{\footnotesize $\maxtspanRepStar{\Pblock_{1}}$};
\draw [line width=2pt, white, decorate,decoration={brace,amplitude=6pt},xshift=0pt,yshift=0pt]
(10,{\yoC+.0}) -- (18,{\yoC+.0});
\draw [decorate,decoration={brace,amplitude=6pt},xshift=0pt,yshift=0pt]
(10,{\yoC+.0}) -- (18,{\yoC+.0}) node [black,midway,yshift=0.15cm,xshift=-0.25cm] 
{\footnotesize $\tspanRepStar{\Pblock_{1}}$};

\draw [decorate,decoration={brace,amplitude=6pt},xshift=0pt,yshift=0pt]
(3,{\yoD+.05}) -- (8,{\yoD+.05}) node [black,midway,yshift=0.15cm,xshift=0.6cm] 
{\footnotesize $\maxtspanStar{\Pblock_{12}}$};
\draw [line width=2pt, white, decorate,decoration={brace,amplitude=6pt},xshift=0pt,yshift=0pt]
(3,{\yoD+.0}) -- (6,{\yoD+.0});
\draw [decorate,decoration={brace,amplitude=6pt},xshift=0pt,yshift=0pt]
(3,{\yoD+.0}) -- (6,{\yoD+.0}) node [black,midway,yshift=0.15cm,xshift=-0.4cm] 
{\footnotesize $\tspanStar{\Pblock_{12}}$};
\end{scope}
\end{tikzpicture}
}

\definecolor{vlgray}{HTML}{EFEFEF}
\definecolor{cbI}{HTML}{4477AA}
\definecolor{cbN}{HTML}{117733}
\definecolor{cbJ}{HTML}{DDCC77}

\newcommand{\xsc}{1.3}
\newcommand{\ysc}{.6}
\newcommand{\pinshift}{3pt}
\newcommand{\nbR}{4}
\newcommand{\nbC}{5}
\newcommand{\nbCh}{4}
\newcommand{\margP}{.025}

\newcommand{\PlenI}{r_I}
\newcommand{\PlenJ}{r_J}
\newcommand{\Plen}{r}
\newcommand{\PnbOI}{\ell}
\newcommand{\PnbOJ}{\ell'}

\tikzset{
shifte node/.style={font=\scriptsize, align=center},
occ node/.style={font=\scriptsize,  align=center, rectangle, rounded corners, minimum size=.85cm, xscale=\xsc, yscale=\ysc, draw, very thin, gray},
dots node/.style={font=\scriptsize, align=center, rectangle, rounded corners, minimum size=.85cm, xscale=\xsc, yscale=\ysc},
patt edge/.style={shorten >= -\pinshift, <-, thick}
}

\newcommand{\drawShiftNode}[3]{
  \node[shifte node, inner sep=1pt, fill=white, white, fill opacity=.8, #3] at (#1) {#2};
  \node[shifte node, inner sep=1pt, #3] at (#1) {#2};
}

\newcommand{\blocksPCombineVBckgrd}{%
\draw[white] (-.5-\margP,-.5-\margP) rectangle ({\nbC+\margP},{\nbR+\margP});
\foreach \c/\r in {1/3, 2/3, 3/3, 5/3}{
    \node[dots node] (n\c\r) at ({\c-1+.5},{\r-1+.5}) {$\vdots$};
}
\foreach \c/\r in {4/1,4/2,4/\nbR}{
    \node[dots node] (n\c\r) at ({\c-1+.5},{\r-1+.5}) {$\dots$};
}
\node[dots node] (n43) at ({4-1+.5},{3-1+.5}) {$\ddots$};
}

\newcommand{\blocksPCombineVPR}{%
\blocksPCombineVBckgrd{}
\foreach \c in {1,2,3}{
\foreach \r in {1,2}{
    \node[occ node] (n\c\r) at ({\c-.5},{\r-.5}) {};
  }
  \node[occ node] (n\c\nbR) at ({\c-.5},{\nbR-.5}) {};
}
\foreach \r in {1,2,\nbR}{
    \node[occ node] (n\nbC\r) at ({\nbC-.5},{\r-.5}) {};
}
\foreach \i in {1,2}{
  \node[cbI] at (-0.3,{\i-1+.5}) {$\patt_{I,\i}$};
  \node[occ node, double] at (.5,{\i-1+.5}) {};
  \node at (0.5,{\i-1+.5}) {$\Cto_{\i}$};
  \draw[thick, cbI] ({0+\margP},{\i-1+\margP}) rectangle ({\nbC-\margP,\i-\margP});
}
\node[cbI] at (-0.3,{\nbR-1+.5}) {$\patt_{I,\PlenJ}$};
\node[occ node, double] at (.5,{\nbR-1+.5}) {};
\node at (0.5,{\nbR-1+.5}) {$\Cto_{\PlenJ}$};
\draw[thick, cbI] ({0+\margP},{\nbR-1+\margP}) rectangle ({\nbC-\margP,\nbR-\margP});
\path[patt edge, cbI]
 (n11) edge (n21)
 (n21) edge (n31)
 (n31) edge (n41)
 (n41) edge (n51)
 (n12) edge (n22)
 (n22) edge (n32)
 (n32) edge (n42)
 (n42) edge (n52)
 (n14) edge (n24)
 (n24) edge (n34)
 (n34) edge (n44)
(n44) edge (n54);
\foreach \c/\r/\x/\y in {2/1/1/1,2/2/1/2,2/\nbR/1/\PlenJ, 3/1/2/1,3/2/2/2,3/\nbR/2/\PlenJ, \nbC/1/\PnbOI-1/1,\nbC/2/\PnbOI-1/2,\nbC/\nbR/\PnbOI-1/\PlenJ}{
  \drawShiftNode{{\c-1+.5},{\r-1+.5}}{$\Csc_{I,\y}[\x]$}{}
}

}

\newcommand{\blocksPCombineVPC}{%
\blocksPCombineVBckgrd{}
\foreach \c in {1,2,3}{
\foreach \r in {1,2}{
    \node[occ node] (n\c\r) at ({\c-.5},{\r-.5}) {};
  }
  \node[occ node] (n\c\nbR) at ({\c-.5},{\nbR-.5}) {};
}
\foreach \r in {1,2,\nbR}{
    \node[occ node] (n\nbC\r) at ({\nbC-.5},{\r-.5}) {};
}
\node[cbJ] at (0.5,-0.3) {$\patt_{J}$};
\node[occ node, double] at (.5,.5) {};
\node at (0.5,0.5) {$\Cto_{1}$};
\draw[thick, cbJ] ({0+\margP},{0+\margP}) rectangle ({1-\margP,\nbR-\margP});
\path[patt edge, cbJ]
 (n11) edge (n12)
 (n12) edge (n13)
 (n13) edge (n14);
\foreach \c/\r/\x/\y in {1/2/0/1, 1/\nbR/0/\PlenJ\!-\!1}{
  \drawShiftNode{{\c-1+.5},{\r-1+.5}}{$\Csc_{J}[\y]$}{}
}

}

\newcommand{\blocksPCombineV}{%
\blocksPCombineVBckgrd{}
\foreach \c in {1,2,3}{
\foreach \r in {1,2}{
    \node[occ node] (n\c\r) at ({\c-.5},{\r-.5}) {};
  }
  \node[occ node] (n\c\nbR) at ({\c-.5},{\nbR-.5}) {};
}
\foreach \r in {1,2,\nbR}{
    \node[occ node] (n\nbC\r) at ({\nbC-.5},{\r-.5}) {};
}
\node[cbN] at (-0.3,-0.3) {$\patt_{N}$};
\node[occ node, double] at (.5,.5) {};
\node at (0.5,0.5) {$\Cto_{1}$};
\draw[thick, cbN] ({0+\margP},{0+\margP}) rectangle ({\nbC-\margP,\nbR-\margP});
\path[patt edge, cbN]
 (n11) edge (n21)
 (n21) edge (n31)
 (n31) edge (n41)
 (n41) edge (n51)
 (n12) edge (n22)
 (n22) edge (n32)
 (n32) edge (n42)
 (n42) edge (n52)
 (n14) edge (n24)
 (n24) edge (n34)
 (n34) edge (n44)
(n44) edge (n54);
\path[patt edge, cbN]
 (n11) edge (n12)
 (n12) edge (n13)
 (n13) edge (n14);
\foreach \c/\r/\x in {2/1/1, 3/1/2, \nbC/1/\PnbOI-1, 1/2/\PnbOI, 2/2/\PnbOI+1, 3/2/\PnbOI+2, \nbC/2/2\PnbOI-1, \nbC/\nbR/\PlenJ\PnbOI-1}{
\drawShiftNode{{\c-1+.5},{\r-1+.5}}{$\Csc_{N}[\x]$}{}
}
\foreach \c/\r/\x in {1/\nbR/(\PlenJ\!-\!1)\PnbOI, 2/\nbR/(\PlenJ\!-\!1)\PnbOI+1, 3/\nbR/(\PlenJ\!-\!1)\PnbOI+2}{
\drawShiftNode{{\c-1+.5},{\r-1+.5}}{$\Csc_{N}[\x]$}{rotate=20}
}

}

\newcommand{\blocksPCombineVOccs}{%
\blocksPCombineVBckgrd{}
\foreach \c in {1,2,3}{
\foreach \r in {1,2}{
    \node[occ node] (n\c\r) at ({\c-.5},{\r-.5}) {};
  }
  \node[occ node] (n\c\nbR) at ({\c-.5},{\nbR-.5}) {};
}
\foreach \r in {1,2,\nbR}{
    \node[occ node] (n\nbC\r) at ({\nbC-.5},{\r-.5}) {};
}
\foreach \c/\r/\x/\y in {1/1/1/1, 2/1/2/1, 3/1/3/1, \nbC/1/\PnbOI/1, 1/2/1/2, 2/2/2/2, 3/2/3/2, \nbC/2/\PnbOI/2, 1/\nbR/1/\PlenJ, 2/\nbR/2/\PlenJ, 3/\nbR/3/\PlenJ, \nbC/\nbR/\PnbOI/\PlenJ}{
\node at ({\c-1+.5},{\r-1+.5}) {$\phantom{o'}o_{\y,\x}\phantom{o'}$};
}
}

\newcommand{\blocksPCombineHBckgrd}{%
\draw[white] (-1-\margP,-1-\margP) rectangle ({9.5+\margP},{.5+\nbR+\margP});
\foreach \c/\r in {1/3, 2/3, 3/3, 5/3, 6/3, 7/3, 9/3}{
    \node[dots node] (n\c\r) at ({\c-1+.5},{\r-1+.5}) {$\vdots$};
}
\foreach \c/\r in {4/1,4/2,4/\nbR, 8/1,8/2,8/\nbR}{
    \node[dots node] (n\c\r) at ({\c-1+.5},{\r-1+.5}) {$\dots$};
}
\node[dots node] (n43) at ({4-1+.5},{3-1+.5}) {$\ddots$};
\node[dots node] (n83) at ({8-1+.5},{3-1+.5}) {$\ddots$};
\foreach \c in {1,2,3,6,7}{
\foreach \r in {1,2}{
    \node[occ node] (n\c\r) at ({\c-.5},{\r-.5}) {};
  }
  \node[occ node] (n\c\nbR) at ({\c-.5},{\nbR-.5}) {};
}
\foreach \r in {1,2,\nbR}{
    \node[occ node] (n\nbC\r) at ({\nbC-.5},{\r-.5}) {};
    \node[occ node] (n9\r) at ({9-.5},{\r-.5}) {};
}
}

\newcommand{\blocksPCombineHOccs}{%
\blocksPCombineHBckgrd{}
\foreach \c/\r/\x/\y in {1/1/1/1, 2/1/2/1, 3/1/3/1, \nbC/1/\PnbOI/1, 1/2/1/2, 2/2/2/2, 3/2/3/2, \nbC/2/\PnbOI/2, 1/\nbR/1/\Plen, 2/\nbR/2/\Plen, 3/\nbR/3/\Plen, \nbC/\nbR/\PnbOI/\Plen}{
\node at ({\c-1+.5},{\r-1+.5}) {$\phantom{o'}o_{\y,\x}\phantom{o'}$};
}
\foreach \c/\r/\x/\y in {1/1/1/1, 2/1/2/1, \nbCh/1/\PnbOJ/1, 1/2/1/2, 2/2/2/2, \nbCh/2/\PnbOJ/2, 1/\nbR/1/\Plen, 2/\nbR/2/\Plen, \nbCh/\nbR/\PnbOJ/\Plen}{
\node at ({\c-1+.5+5},{\r-1+.5}) {$\phantom{o'}o'_{\y,\x}\phantom{o'}$};
}
}

\newcommand{\blocksPCombineHRC}{%
\blocksPCombineHBckgrd{}
\node[occ node, double] at (.5,.5) {};
\node[cbI] at (0.5,-0.3) {$\patt_{I}$};
\node at (0.5,0.5) {$\Cto_{I}$};
\draw[thick, cbI] ({0+\margP},{0+\margP}) rectangle ({\nbC-\margP,\nbR-\margP});
\path[patt edge, cbI]
 (n11) edge (n21)
 (n21) edge (n31)
 (n31) edge (n41)
 (n41) edge (n51)
 (n12) edge (n22)
 (n22) edge (n32)
 (n32) edge (n42)
 (n42) edge (n52)
 (n14) edge (n24)
 (n24) edge (n34)
 (n34) edge (n44)
(n44) edge (n54);
\path[patt edge, cbI]
 (n11) edge (n12)
 (n12) edge (n13)
 (n13) edge (n14);
\foreach \c/\r/\x in {2/1/1, 3/1/2, \nbC/1/\PnbOI-1, 1/2/\PnbOI, 2/2/\PnbOI+1, 3/2/\PnbOI+2, \nbC/2/2\PnbOI-1, \nbC/\nbR/\Plen\PnbOI-1}{
\drawShiftNode{{\c-1+.5},{\r-1+.5}}{$\Csc_{I}[\x]$}{}
}
\foreach \c/\r/\x in {1/\nbR/(\Plen\!-\!1)\PnbOI, 2/\nbR/(\Plen\!-\!1)\PnbOI+1, 3/\nbR/(\Plen\!-\!1)\PnbOI+2}{
\drawShiftNode{{\c-1+.5},{\r-1+.5}}{$\Csc_{I}[\x]$}{rotate=20}
}

\node[cbJ] at (5.5,-0.3) {$\patt_{J}$};
\node[occ node, double] at (5.5,.5) {};
\node at (5.5,0.5) {$\Cto_{J}$};
\draw[thick, cbJ] ({5+\margP},{0+\margP}) rectangle ({5+\nbCh-\margP,\nbR-\margP});
\path[patt edge, cbJ]
 (n61) edge (n71)
 (n71) edge (n81)
 (n81) edge (n91)
 (n62) edge (n72)
 (n72) edge (n82)
 (n82) edge (n92)
 (n64) edge (n74)
 (n74) edge (n84)
 (n84) edge (n94);
\path[patt edge, cbJ]
 (n61) edge (n62)
 (n62) edge (n63)
 (n63) edge (n64);
\foreach \c/\r/\x in {2/1/1, \nbCh/1/\PnbOJ-1, 1/2/\PnbOJ, 2/2/\PnbOJ+1, \nbCh/2/2\PnbOJ-1, \nbCh/\nbR/\Plen\PnbOJ-1}{
\drawShiftNode{{\c-1+5+.5},{\r-1+.5}}{$\Csc_{J}[\x]$}{}
}
\foreach \c/\r/\x in {1/\nbR/(\Plen\!-\!1)\PnbOJ, 2/\nbR/(\Plen\!-\!1)\PnbOJ+1}{
\drawShiftNode{{\c-1+5+.5},{\r-1+.5}}{$\Csc_{J}[\x]$}{rotate=20}
}
}

\newcommand{\blocksPCombineH}{%
\blocksPCombineHBckgrd{}
\node[cbN] at (-0.3,-0.3) {$\patt_{N}$};
\node[occ node, double] at (.5,.5) {};
\node at (0.5,0.5) {$\Cto_{N}$};
\draw[thick, cbN] ({0+\margP},{0+\margP}) rectangle ({\nbC+\nbCh-\margP,\nbR-\margP});
\path[patt edge, cbN]
 (n11) edge (n21)
 (n21) edge (n31)
 (n31) edge (n41)
 (n41) edge (n51)
 (n12) edge (n22)
 (n22) edge (n32)
 (n32) edge (n42)
 (n42) edge (n52)
 (n14) edge (n24)
 (n24) edge (n34)
 (n34) edge (n44)
 (n44) edge (n54) 
 (n61) edge (n71)
 (n71) edge (n81)
 (n81) edge (n91)
 (n62) edge (n72)
 (n72) edge (n82)
 (n82) edge (n92)
 (n64) edge (n74)
 (n74) edge (n84)
 (n84) edge (n94);
\path[patt edge, cbN]
 (n11) edge (n12)
 (n12) edge (n13)
 (n13) edge (n14);
\path[patt edge, cbN, bend right=15pt]
 (n11.35) edge (n61.145)
 (n12.35) edge (n62.145)
 (n14.35) edge (n64.145);
\foreach \c/\r/\x in {2/1/1, 3/1/2, \nbC/1/\PnbOI-1}{
\drawShiftNode{{\c-1+.5},{\r-1+.5}}{$\Csc_{N}[\x]$}{}
}
\foreach \c/\r/\x in {1/1/\PnbOI, 2/1/\PnbOI+1}{
\drawShiftNode{{\c-1+5+.5},{\r-1+.5}}{$\Csc_{N}[\x]$}{}
}
\foreach \c/\r/\x in {1/\nbR/(\Plen-1)(\PnbOI+\PnbOJ), 2/\nbR/(\Plen-1)(\PnbOI+\PnbOJ)+1, 3/\nbR/(\Plen-1)(\PnbOI+\PnbOJ)+2, 2/2/\PnbOI+\PnbOJ+1, 3/2/\PnbOI+\PnbOJ+2, \nbC/2/2\PnbOI+\PnbOJ-1, \nbC/\nbR/\Plen\PnbOI+(\Plen-1)\PnbOJ-1, 1/2/\PnbOI+\PnbOJ}{
\drawShiftNode{{\c-1+.5},{\r-1+.5}}{$\Csc_{N}[\x]$}{rotate=20}
}
\foreach \c/\r/\x in {1/\nbR/\Plen\PnbOI+(\Plen-1)\PnbOJ, 2/\nbR/\Plen\PnbOI+(\Plen-1)\PnbOJ+1, \nbCh/1/\PnbOI+\PnbOJ-1, 1/2/2\PnbOI+\PnbOJ, 2/2/2\PnbOI+\PnbOJ+1, \nbCh/2/2(\PnbOI+\PnbOJ)-1, \nbCh/\nbR/\Plen(\PnbOI+\PnbOJ)-1}{
\drawShiftNode{{\c-1+5+.5},{\r-1+.5}}{$\Csc_{N}[\x]$}{rotate=20}
}
}

\newlength{\ndhlen}
\newlength{\ndwlen}
\newcommand{\exEvATyp}{ptstart node}
\newcommand{\exEvBTyp}{ptend node}
\newcommand{\exEvCTyp}{ptstart node}
\newcommand{\headEx}[4]{
  \node[info node, xshift=-1.3\ndwlen, yshift=#1\ndhlen-1em] (r0) at (x0) {$\Clen=#2$};
  \node[info node, xshift=-1.3\ndwlen, yshift=#1\ndhlen-2em] (p0) at (x0) {$\Cprd=#3$};
  \node[info node, xshift=-1.3\ndwlen, yshift=#1\ndhlen] (t0) at (x0) {$\Cto=$ #4};
}
\newcommand{\threeEvEx}[9]{
  \node[ghost node] (x0) at (#1) {};
  \node[main node] (b0) at (#1) {};
  \headEx{1.45}{#2}{#3}{#4}
  \node[leaf node, xshift=\ndwlen] (l2) at (b0) {};
  \node[leaf node] (l1) [above of=l2] {};
  \node[leaf node] (l3) [below of=l2] {};
  \node (p1) [right of=l1, \exEvATyp] {#5};
  \node (p2) [right of=l2, \exEvBTyp] {#7};
  \node (p3) [right of=l3, \exEvCTyp] {#9};
  \path
    (b0) edge (l1)
    (b0) edge (l2)
    (b0) edge (l3);
\path[dotted, thin, ->, bend left=20, color=darkgray]
    (l1) edge node[dst node, right] {$#6$} (l2)
    (l2) edge node[dst node, right] {$#8$} (l3);
}
\newcommand{\twoEvEx}[7]{
  \node[ghost node] (x0) at (#1) {};
  \node[main node, xshift=.2\ndwlen] (b0) at (#1) {};
  \headEx{1.4}{#2}{#3}{#4}
  \node[ghost node, xshift=\ndwlen, yshift=-.5\ndhlen] (c1) at (b0) {};
  \node[leaf node] (l1) [above of=c1] {};
  \node[leaf node] (l2) [below of=l1] {};
  \node (p1) [right of=l1, \exEvATyp] {#5};
  \node (p2) [right of=l2, \exEvBTyp] {#7};
  \path
    (b0) edge (l1)
    (b0) edge (l2);
\path[dotted, thin, ->, bend left=20, color=darkgray]
    (l1) edge node[dst node, right] {$#6$} (l2);
}
\newcommand{\oneEvEx}[5]{
  \node[ghost node] (x0) at (#1) {};
  \node[main node, xshift=.5\ndwlen] (b0) at (#1) {};
  \headEx{1.4}{#2}{#3}{#4}
  \node[leaf node, xshift=\ndwlen] (l1) at (b0) {};
  \node (p1) [right of=l1, \exEvATyp] {#5};
  \path
    (b0) edge (l1);
}

\newcommand{\otherEvEx}[9]{
  \node[ghost node] (x0) at (#1) {};
  \node[main node, xshift=-.3\ndwlen, yshift=-.2\ndhlen] (b0) at (#1) {};
  \headEx{1.45}{#2}{#3}{#4}
  \node[ghost node, xshift=\ndwlen, yshift=-.5\ndhlen] (c1) at (b0) {};
  \node[leaf node] (l1) [above of=c1] {};
  \node[main node] (l2) [below of=l1] {};
  \node[leaf node, xshift=\ndwlen] (l3) at (l2) {};
  \node (p1) [right of=l1, \exEvATyp] {#5};
  \node (p3) [right of=l3, \exEvBTyp] {#7};
  \node[info node, xshift=-.5\ndwlen, yshift=-1em] (r1) at (l2) {$\Clen_1=#8, \Cprd_1=#9$};
  \path
    (b0) edge (l1)
    (b0) edge (l2)
    (l2) edge (l3);
\path[dotted, thin, ->, bend left=20, color=darkgray]
    (l1) edge node[dst node, right] {$#6$} (l2);
}

\usepackage{siunitx}

\newcommand{\urlcode}{\url{https://github.com/nurblageij/periodic-patterns-mdl}}

\newcommand{\VShortOnly}[1]{}
\newcommand{\VLongOnly}[1]{#1}
\newcommand{\custmath}[2]{\[#1\;#2\]}
\newcommand{\custfrac}[2]{\frac{#1}{#2}}

\usepackage[left=4.4cm,right=4.4cm,asymmetric]{geometry}
\usepackage{placeins}

\begin{document}

\title{Mining Periodic Patterns with a MDL Criterion} 

\author{Esther Galbrun\inst{1} \and
Peggy Cellier\inst{2} \and 
Nikolaj Tatti\inst{1,4} \and \\
Alexandre Termier\inst{2} \and
Bruno Cr\'emilleux\inst{3} 
}

\institute{Department of Computer Science, Aalto University, Finland \\ \url{{esther.galbrun, nikolaj.tatti}@aalto.fi} \and
Univ.\ Rennes, \{INSA, Inria\}, CNRS, IRISA, France \\ \url{{peggy.cellier, alexandre.termier}@irisa.fr} \and
Normandie Univ., UNICAEN, ENSICAEN, CNRS -- UMR GREYC, France \\ \url{bruno.cremilleux@unicaen.fr} \and
F-Secure, Finland 
}

\authorrunning{Galbrun \textit{et al.}}

\maketitle

\begin{abstract}
The quantity of event logs available is increasing rapidly, be they produced by industrial processes, computing systems, or life tracking, for instance. It is thus important to design effective ways to uncover the information they contain.
Because event logs often record repetitive phenomena, mining periodic patterns is especially relevant when considering such data.
Indeed, capturing such regularities is instrumental in providing condensed representations of the event sequences.

We present an approach for mining periodic patterns from event logs while relying on a Minimum Description Length (MDL) criterion to evaluate candidate patterns. Our goal is to extract a set of patterns that suitably characterises the periodic structure present in the data.
We evaluate the interest of our approach on several real-world event log datasets.

\end{abstract}
 
\section{Introduction}

{\em Event logs} are among the most ubiquitous types of data nowadays. 
They can be machine generated (server logs, database transactions, sensor data) or human generated (ranging from hospital records to life tracking, a.k.a.\ quantified self), and are bound to become ever more voluminous and diverse with the increasing digitisation of our lives and the advent of the Internet of Things (IoT).
Such logs are often the most readily available sources of information on a system or process of interest.
It is thus critical to have effective and efficient means to analyse them and extract the information they contain.

Many such logs monitor repetitive processes, and some of this repetitiveness is recorded in the logs.
A careful analysis of the logs can thus help understand the characteristics of the underlying recurrent phenomena.
However, this is not an easy task: a log usually captures many different types of events.
Events related to occurrences of different repetitive phenomena are often mixed together as well as with noise, and the different signals need to be disentangled to allow analysis.
This can be done by a human expert having a good understanding of the domain and of the logging system, but is tedious and time consuming.

{\em Periodic pattern mining} algorithms~\cite{Ozden98} have been proposed to tackle this problem.
These algorithms can discover periodic repetitions of sets or sequences of events amidst unrelated events.
They exhibit some resistance to noise, when it takes the form of slight variations in the inter-occurrence delay~\cite{Berberidis02} or of the recurrence being limited to only a portion of the data~\cite{Ma01}.
However, such algorithms suffer from the traditional plague of pattern mining algorithms: they output too many patterns (up to several millions), even when relying on condensed representations~\cite{LopezCueva12}.

Recent approaches have therefore focused on optimising the quality of the extracted {\em pattern set} as a whole~\cite{DRZ07a}, rather than finding individual high-quality patterns. 
In this context, the adaptation of the Minimal Description Length (MDL) principle~\cite{rissanen1978modeling,grunwald07} to pattern set mining has given rise to a fruitful line of work~\cite{krimp2011,bonchi11krimpro,tatti_long_2012,bhattacharyya_efficiently_2017}.
The MDL principle is a concept from information theory based on the insight that any structure in the data can be exploited to compress the data, and aiming to strike a balance between the complexity of the model and its ability to describe the data.

The most important structure of the data on which we focus here, i.e.\ of event logs, is the periodic recurrence of some events.  For a given event sequence, we therefore want to identify a set of patterns that captures the periodic structure present in the data, and we devise a MDL criterion to evaluate candidate pattern sets for this purpose.
First, we consider a simple type of model, representing event sequences with cycles over single events.
Then, we extend this model so that cycles over distinct events can be combined together. By simply letting our patterns combine not only events but also patterns recursively, we obtain an expressive language of periodic patterns. For instance, it allows us to express the following daily routine:

\smallskip
{\hspace{-2ex} \parbox{.95\textwidth}{
Starting Monday at $7$:$30$ AM, wake up, then, $10$ minutes later, prepare coffee, repeat every $24$ hours for $5$ days, repeat this every $7$ days for $3$ months 
}}
\smallskip

\noindent as a pattern consisting of two nested cycles, respectively with $24$ hours and $7$ days periods, over the events ``waking up'' and ``preparing coffee''.

In short, we propose a novel approach for mining periodic patterns using a MDL criterion. The main component of this approach---and our main contribution---is the definition of an expressive pattern language and the associated encoding scheme which allows to compute a MDL-based score for a given pattern collection and sequence. We design an algorithm for putting this approach into practise and perform an empirical evaluation on several event log datasets. 
We show that we are able to extract sets of patterns that compress the input sequences and to identify meaningful patterns. 

We start by reviewing the main related work, in Section~\ref{sec:related}. In Section~\ref{sec:problem}, we introduce our problem setting and a simple model consisting of cycles over single events, which we extend in Section~\ref{sec:problem_complex}. \VLongOnly{In Section~\ref{sec:comb}, we look at how patterns can be combined and compare costs. }We present an algorithm for mining periodic patterns that compress in Section~\ref{sec:algo} and evaluate our proposed approach over several event log datasets in Section~\ref{sec:xps}. We reach conclusions in Section~\ref{sec:conclusion}.  

\VShortOnly{We focus here on the high-level ideas, and refer the interested reader to our report~\cite{extended} that includes technical details, additional examples and experiments.}
\VLongOnly{This report extends our conference publication~\cite{conf} with technical details, numerous examples, and additional experiments.}

\section{Related Work}
\label{sec:related}

The first approaches for mining periodic patterns\VLongOnly{~\cite{Ozden98,Han98,Han99} were designed to augment traditional itemset and sequence mining techniques with the capacity to identify events whose occurrences are regularly spaced in time.
They} used extremely constrained definitions of the periodicity. In~\cite{Ozden98}, \emph{all} occurrences must be regularly spaced; In~\cite{Han98,Han99}, some missing occurrences are permitted but all occurrences must follow the same regular spacing.
As a result, these approaches are extremely sensitive to even small amounts of noise in the data.
Ma \textit{et al.}~\cite{Ma01} later proposed a more robust approach, which can extract periodic patterns in the presence of gaps of arbitrary size in the data\VLongOnly{: the recurrence can be interrupted and restarted, possibly with a different spacing. Such perturbations are frequent in real data}. \VLongOnly{

The}\VShortOnly{While the} above approaches require time to be discretized as a preprocessing (time steps of hour or day length, for example)\VLongOnly{, smoothing out small changes in inter-occurrence delays and limiting the search for the correct period to a predetermined range. These approaches might be too coarse grained, however, and are dependant on the discretization. Several}\VShortOnly{, several} solutions have been proposed to directly discover candidate periods from raw timestamp data, using the Fast Fourier Transform~\cite{Berberidis02} or statistical models~\cite{li2012mining,yuan2017pred}.\VLongOnly{ 

} All of the above approaches are susceptible to producing a huge number of patterns, making the exploitation of their results difficult.
The use of a {\em condensed representation} for periodic patterns~\cite{LopezCueva12} allows to significantly reduce the number of patterns output, without loss of information, but falls short of satisfactorily addressing the problem.

\VLongOnly{\medskip}

Considering pattern mining more in general, to tackle this pervasive issue of the overwhelming number of patterns extracted, research has focused on extracting {\em pattern sets}~\cite{DRZ07a}: finding a (small) set of patterns that together optimise some interest criterion.
One such criterion is based on the Minimum Description Length (MDL) principle~\cite{grunwald_model_2000}. Simply put, it states that \emph{the best model is the one that compresses the data best}. Following this principle, the \algname{Krimp} algorithm~\cite{krimp2011} was proposed, to select a subset of frequent itemsets that yields the best lossless compression of a transactional database. This algorithm was later improved~\cite{SLIM} and the approach extended to analyse event sequences~\cite{tatti_long_2012,goKRIMP,bhattacharyya_efficiently_2017}.
Along a somewhat different approach, Kiernan and Terzi proposed to use MDL to summarize event sequences~\cite{kiernan09}.  

To the best of our knowledge, the only existing method that combines periodic pattern mining and a MDL criterion was proposed by Heierman \textit{et al.}~\cite{DBLP:conf/icdm/HeiermanC03}.
This approach considers a single regular episode at a time and aims to select the best occurrences for this pattern, independently of other patterns. Instead, we use a MDL criterion in order to select a good collection of periodic patterns.

\section{Preliminary Notation and Problem Definition}
\label{sec:problem}

Next, we formally define the necessary concepts and formulate our problem, focusing on simple cycles.
\VLongOnly{
But first, let us clarify some of the notation we use throughout.

Lists are represented by enumerating their elements in order of occurrence, enclosed between $\lls$ and $\lle$, as in $\lls i_1, i_2, \dots \lle$ for instance, with $\llempty$ denoting the empty list.
We use $\opconcat$ to represent the concatenation of lists, as in 
\[\LL{a, b, c} = \LL{a} \opconcat \LL{b, c}\text{ and }
\LL{i_1, i_2, \dots, i_9} = \fnconcat_{k \in [1..9]} \LL{i_k}\;.\]
Given a list $L$, $L[k]$ returns the element at $k^{th}$ position (indexing starts at $1$).

We also use a simplified notation for lists, especially when using them as indices. Lists and single elements are then denoted respectively as upper-case and lower-case letters or numbers, and concatenation is simply represented by concatenating the corresponding letters. In this notation, we use $0$ to represent the empty list. For instance, the indices in $\Pblock_0$, $\Pblock_X$ and $\Pblock_{Xy}$ represent an empty list, a list $X$, and element $y$ concatenated to the list $X$, respectively. 

All logarithms are to base $2$.

Symbols used are listed on the last page of this report. 
}

\VShortOnly{\mpara{Event sequences and cycles.}}
\VLongOnly{\mpara{A timestamped event sequence as input data.}}
Our input data is a collection of timestamped occurrences of some events, which we call an \emph{event sequence}.
The events come from an alphabet $\ABC$\label{sym:ABC} and will be represented with lower case letters.
We assume that an event can occur only once per time step, so the data can be represented as a list of timestamp--event pairs, such as
\VShortOnly{\[ \seqex{1} = \lls (2,c),(3,c),(6,a),(7,a),(7,b),(19,a),(30,a),(31,c),(32,a),(37,b) \lle \;.\]}
\VLongOnly{
\begin{align*}
\seqex{1} &= \lls (2,c),(3,c),(6,a),(7,a),(7,b),(19,a),\\
&(30,a),(31,c),(32,a),(37,b),(42,a),(48,c),(54,a) \lle \;.
\end{align*} }
Whether timestamps represent days, hours, seconds, or something else depends on the application, the only requirement is that they be expressed as positive integers.
We denote as $\seq[\alpha]$\label{sym:seqalpha} the event sequence $\seq$\label{sym:seq} restricted to event $\alpha$\label{sym:alpha}, that is, the subset obtained by keeping only occurrences of event $\alpha$.
\VLongOnly{
For instance, we can represent $\seqex[a]{1}$, the event sequence above restricted to event $a$, simply as a list of timestamps:
\begin{align*}
\seqex[a]{1} = \langle6,7,19,30,32,42,54\rangle\;.
\end{align*} 
}
We denote as $\len{\seq}$\label{sym:lenS} the number of timestamp--event pairs contained in event sequence $\seq$, i.e.\ its \emph{length}, and $\tspan{\seq}$\label{sym:durationS} the time spanned by it, i.e.\ its \emph{duration}. That is, $\tspan{\seq} = \tSend(\seq) - \tSstart(\seq)$, where $\tSend(\seq)$\label{sym:tSend} and $\tSstart(\seq)$\label{sym:tSstart} represent the largest and smallest timestamps in $\seq$, respectively.
\VLongOnly{
Observe that $\len{\seq[\alpha]}$ equals the number of occurrences of $\alpha$ in the original sequence, and that $\tspan{\seq[\alpha]} \leq \tspan{\seq}.$
In the example above we have $\len{\seqex{1}} = 13$, $\len{\seqex[a]{1}} = 7$, $\tspan{\seqex{1}} = 52$ and $\tspan{\seqex[a]{1}} = 48$.}

\VShortOnly{\medskip}
\VLongOnly{\mpara{Cycles as periodic patterns.}}
Given such an event sequence, our goal is to extract a representative collection of cycles.
A \emph{cycle} is a periodic pattern that takes the form of an ordered list of occurrences of an event, where successive occurrences appear at the same distance from one another.
We will not only consider perfect cycles, where the inter-occurrence distance is constant, but will allow some variation.

A cycle is specified by indicating:
\VShortOnly{\begin{itemize}
\item the repeating event, called \emph{cycle event} and denoted as $\Cev$, 
\item the number of repetitions of the event, called \emph{cycle length}, $\Clen$, 
\item the inter-occurrence distance, called \emph{cycle period}, $\Cprd$, and 
\item the timestamp of the first occurrence, called \emph{cycle starting point}, $\Cto$.
\end{itemize}}
\VLongOnly{\begin{itemize}
\item the repeating event, called the \emph{cycle event} and denoted as $\Cev$\label{sym:Cev}, 
\item the number of repetitions of the event, called the \emph{cycle length} and denoted as $\Clen$\label{sym:Clen} ,
\item the inter-occurrence distance, called the \emph{cycle period} and denoted as $\Cprd$\label{sym:Cprd}, and 
\item the timestamp of the first occurrence, called the \emph{cycle starting point} and denoted as $\Cto$\label{sym:Cto}.
\end{itemize}}
Cycle lengths, cycle periods and cycle starting points take positive integer values (we choose to restrict periods to be integers for simplicity and interpretability). More specifically, we require $\Clen > 1$, $\Cprd > 0$ and $\Cto \geq 0$. 

In addition, since we allow some variation in the actual inter-occurrence distances, we need to indicate an offset for each occurrence in order to be able to reconstruct the original subset of occurrences, that is, to recover the original timestamps. For a cycle of length $\Clen$, this is represented as an ordered list of $\Clen-1$ signed integer offsets, called the \emph{cycle shift corrections} and denoted as $\Csc$\label{sym:Csc}.
Hence, a cycle is a 5-tuple $\cycle\label{sym:cycle} = (\Cev, \Clen, \Cprd, \Cto, \Csc)$.

\VLongOnly{Note that since the cycles we consider here involve one event each, we can process the occurrences of each event separately. In other words, we can split the original sequence $\seq$ into subsequences $\seq[\alpha]$, one for each event $\alpha$, and handle them separately.} 

\VLongOnly{\mpara{A cycle's cover.}}
For a given cycle $\cycle = (\Cev, \Clen, \Cprd, \Cto, \Csc)$, with $\Csc = \LL{ \Ces_1, \dots, \Ces_{\Clen-1} }$ we can recover the corresponding occurrences timestamps by reconstructing them recursively, starting from $\Cto$: $t_1 = \Cto$, $t_k = t_{k-1}+\Cprd+\Ces_{k-1}.$ 
Note that this is different from first reconstructing the occurrences while assuming perfect periodicity as $\Cto, \Cto+\Cprd, \Cto+2\Cprd, \dots, \Cto+(\Clen-1) \Cprd$, then applying the corrections, because in the former case the corrections actually accumulate.

Then, we overload the notation and denote the time spanned by the cycle as $\tspan{C}$\label{sym:Cspan}\VShortOnly{. }\VLongOnly{, that is 
\begin{align*}
\tspan{C} & =t_\Clen - t_1 \\
\VLongOnly{&= (t_{\Clen-1}+\Cprd+\Ces_{\Clen-1}) - \Cto \\
&= \big((t_{\Clen-2}+\Cprd+\Ces_{\Clen-2})+\Cprd+\Ces_{\Clen-1}\big) - \Cto \\}
&= (\Clen-1)\Cprd+\Ces_{1}+\dots+\Ces_{\Clen-1}\;.
\end{align*}}
Denoting as $\sumel{\Csc}$\label{sym:sumel} the sum of the shift corrections in $\Csc$, $\sumel{\Csc} = \sum_{\Ces \in \Csc} \Ces$, we have
\custmath{\tspan{C} = (\Clen-1)\Cprd+\sumel{\Csc}}{.}
Note that this assumes that the correction maintains the order of the occurrences. This assumption is reasonable since an alternative cycle that maintains the order can be constructed for any cycle that does not.

We denote as $\cov{\cycle}$\label{sym:cov} the corresponding set of reconstructed timestamp--event pairs \custmath{\cov{\cycle} = \{(t_1, \Cev), (t_2, \Cev), \dots, (t_r, \Cev)\}}{.} 
We say that a cycle covers an occurrence if the corresponding timestamp--event pair belongs to the reconstructed subset $\cov{\cycle}$.

Since we represent time in an absolute rather than relative manner and assume that an event can only occur once at any given timestamp, we do not need to worry about overlapping cycles nor about an order between cycles. Given a collection of cycles representing the data, the original list of occurrences can be reconstructed by reconstructing the subset of occurrences associated with each cycle, regardless of order, and taking the union.
We overload the notation and denote as $\cov{\ccycle}$ the set of reconstructed timestamp--event pairs for a collection $\ccycle$\label{sym:ccycle} of cycles $\ccycle =\{\cycle_1, \dots, \cycle_m\}$, that is \custmath{\cov{\ccycle} = \bigcup_{\cycle \in \ccycle} \cov{\cycle}}{.}

For a sequence $\seq$ and cycle collection $\ccycle$ we call \emph{residual} the timestamp--event pairs not covered by any cycle in the collection: \custmath{\residual{\ccycle, \seq}\label{sym:residual} = \seq \setminus \cov{\ccycle}}{.}

We associate a cost to each individual timestamp--event pair $o = (t, \alpha)$ and each cycle $\cycle$, respectively denoted as $\cl(o)$\label{sym:costCL} and $\cl(\cycle)$, which we will define shortly. Then, we can reformulate our problem of extracting a representative collection of cycles as follows:
\begin{problem}
\label{prob:cycles}
Given an event sequence $\seq$, find the collection of cycles $\ccycle$ minimising the cost
\[\cl(\ccycle, \seq) = \sum_{\cycle \in \ccycle} \cl(\cycle) + \sum_{o \in \residual{\ccycle, \seq}} \cl(o)\;.\]
\end{problem}

\mpara{Code lengths as costs.}
This problem definition can be instantiated with different choices of costs.
Here, we propose a choice of costs motivated by the MDL principle. Following this principle, we devise a scheme for encoding the input event sequence using cycles and individual timestamp--event pairs. The cost of an element is then the length of the code word assigned to it under this scheme, and the overall objective of our problem becomes finding the collection of cycles that results in the shortest encoding of the input sequence, i.e.\ finding the cycles that compress the data most. 
In the rest of this section, we present our custom encoding scheme.
\VShortOnly{Note that all logarithms are to base $2$.}

\VLongOnly{For each type of information, we need to determine the most appropriate way to encode it, given the type of patterns we are interested in finding.
The following should always be kept in mind
\begin{quote}
In MDL we are NEVER concerned with actual encodings; we are only concerned
with code length functions. (Peter D.\ Gr\"{u}nwald 2004)
\end{quote}

\mpara{Outline of code systems.}
Given a collection of symbols $Z$ that we might need to transmit, such as, in our case the alphabet of events over which our data sequence is expressed or the range of values that the periods might take, and a particular symbol $z$, all we are interested is the length of the code assigned to $z$, which we denote as $\cl(z)$, not the actual code.

Different code systems can be used, but we focus on those that possess the \emph{prefix property}, meaning that there will not be any two code words in the system such that one is a prefix of the other, making such code uniquely decodable. 

For a collection of symbols $Z$, where each symbol $z$ is associated with an occurrence frequency $\fr(z)$, the optimal prefix code is such that $\cl(z) = -\log(\fr(z))$. However, this requires that the receptors knows the occurrence frequencies.

\emph{Prequential coding} allows to obtain a code that is almost optimal, without knowing the frequencies. Such a code will assign shorter codes to, and hence favour, frequently occurring values.

\emph{Fixed-length codes}, as the name indicates, assign codes of equal length to all values, and hence do not favour any value. Each value is encoded with a code of length $\log(\abs{Z})$. 

\emph{Universal codes} allow to encode non-negative integers, assigning shorter codes to smaller numerical values. In particular, the code length assigned to $z$ is $\unic(z) = \log^*(z) + \log(c_0)$, where $c_0$ is a constant which must be adjusted to ensure that the Kraft inequality is satisfied, i.e.\ such that \[\sum_{z \in \mathbb{N}} 2^{-\unic(z)} \leq 1.\] 
How much small values are favoured compared to larger ones can be adjusted. To avoid wasting bits on unused values large values, $c_0$ can be adjusted to ensure that Kraft inequality is not only satisfied but holds with strict equality. That is, given some upper bound $v$ on the values to encode, we denote as $\univ{v}$ the code length obtained with an adjusted $c_0$ so that \[\sum_{z \in [1..v]} 2^{-\univ{v}(z)} = 1\;.\]

\mpara{Choosing the most appropriate encoding for cycles.}}
For each cycle we need to specify its event, length, period, starting point and shift corrections, that is \custmath{\cl(\cycle) = \cl(\Cev) + \cl(\Clen) + \cl(\Cprd) + \cl(\Cto) +\cl(\Csc)}{.}
It is important to look more closely at the range in which each of these pieces of information takes value, at what values---if any---should be favoured, and at how the values of the different pieces depend on one another.

\VLongOnly{Clearly, a cycle over event $\Cev$ cannot have a length greater than $\len{\seq^{(\Cev)}}$. On the other hand, if it has length $\Clen$, it cannot have a period greater than $\tspan{\seq^{(\Cev)}}/(\Clen-1)$. 
Furthermore, once $\Cto$ is known, the period is further restricted to $(\tSend(\seq^{(\Cev)}) - \Cto)/(\Clen-1)$. And vice-versa, if we first fix the period, it creates limitations on the values the length can take, which in turn affects the values the starting point can take. So, we see a clear dependency between these values.
Also note that the maximum values for the period and the starting point depend on the time span of the sequence, while the maximum value for the length depends on the number of occurrences of the event.
To avoid wasting bits, it might be useful to normalise the time scale to the smallest encountered time step.

\mpara{Encoding with fixed-length codes.}
A somewhat naive approach to encode a cycle is to use fixed-length codes for the event, length, period and starting point, and an adjusted universal code for the shift corrections.
The magnitude of an individual shift correction can be anywhere between $0$ and $\tspan{\seq}$. So if we let $m=\tspan{\seq}+1$, we can use a code word of length $\univ{m}(\abs{e}+1)$ to indicate the absolute value of shift correction $e$ and add one bit to indicate its direction.
Since we can easily determine that the length of a cycle can be no larger than $\len{\seq}$ and that, neglecting the shift corrections, its period and starting point can take values no larger than $\tspan{\seq}/2$ and $\tspan{\seq}$, respectively, we get
\begin{align*}
\cl(\cycle) =& \cl(\Cev) + \cl(\Clen) + \cl(\Cprd) + \cl(\Cto) +\cl(\Csc) \\
=& \log(\abs{\ABC}) + \log(\len{\seq}) + \log(\tspan{\seq}/2) + \log(\tspan{\seq}) \\
&+ \sum_{e \in \Csc} (\univ{m}(\abs{e}+1) + 1) \;.
\end{align*}

\mpara{Optimising the encoding.}
But we can do better, by exploiting the dependencies between the pieces of information.
To encode the cycles' events, we can use either fixed-length coding, as above, or codes based on the events' frequency in the original sequence. 
In the first case the length of the code word representing the event is constant across all cycles, regardless of the event and only depends on the size of the alphabet.
In the second case,}\VShortOnly{To encode the cycles' events, we can use codes based on the events' frequency in the original sequence, so that}
 events that occur more frequently in the event sequence will receive shorter code words: 
\custmath{\cl(\Cev) = -\log(\fr(\Cev)) = -\log(\custfrac{\len{\seq[\Cev]}}{\len{\seq}})}{.}
This requires that we transmit the number of occurrences of each event in the original event sequence. 
To optimise the overall code length, the length of the code word associated to each event should actually depend on the frequency of the event in the selected collection of cycles. However, this would require keeping track of these frequencies and updating the code lengths dynamically. Instead, we use the frequencies of the events in the input sequence as a simple proxy.  

\VShortOnly{Clearly, a cycle with event $\Cev$ cannot have a length greater than $\len{\seq^{(\Cev)}}$.}
Once the cycle event $\Cev$ and its number of occurrences are known, we can encode the cycle length with a code word of length
\custmath{\cl(\Clen) = \log(\len{\seq[\Cev]})}{,}
resulting in the same code length for large numbers of repetitions as for small ones.

\VLongOnly{Recall that \[\tspan{\cycle} = (\Clen-1)\Cprd+\sumel{\Csc}\;.\]} Clearly, a cycle spans at most the time of the whole sequence, i.e.\  ${\tspan{C} \leq \tspan{\seq}}$\VShortOnly{,}\VLongOnly{. Hence 
\[\Cprd \leq \Big\lfloor\frac{\tspan{\seq}-\sumel{\Csc}}{\Clen-1}\Big\rfloor\;,\]}
so that knowing the cycle length, the shift corrections, and the sequence time span, we can encode the cycle period with a code word of length 
\[\cl(\Cprd) = \log\Big(\Big\lfloor\frac{\tspan{\seq}-\sumel{\Csc}}{\Clen-1}\Big\rfloor\Big)\;.\]
\VLongOnly{ Note that the code word for the period of a cycle will be shorter if the cycle has greater length (since there are more repetitions, the period cannot be as long).}

Next, knowing the cycle length and period as well as the sequence time span, \VLongOnly{the starting point $\Cto$ can take any value between $\tSstart(\seq)$ and $\tSend(\seq) - \tspan{\cycle} = \tSend(\seq) -\sumel{\Csc} - (\Clen-1)\Cprd$. Hence, }we can specify the value of the starting point with a code word of length 
\[\cl(\Cto) = \log(\tspan{\seq} - \sumel{\Csc} - (\Clen-1)\Cprd + 1)\;.\] 
\VLongOnly{Note that if the cycle spans a larger part of the sequence, the range of the starting point is more restricted, and so it can be represented with a shorter code word.}

Finally, we encode the shift corrections as follows: each correction $e$ is represented by $\abs{e}$ ones, prefixed by a single bit to indicate the direction of the shift, with each correction separated from the previous one by a zero. 
For instance,  $\Csc =\LL{3, -2, 0, 4}$ would be encoded as 
$\signDG{0}\valDG{111}
\sepDG{0}\signDG{1}\valDG{11}
\sepDG{0}\signDG{0}
\sepDG{0}\signDG{0}\valDG{1111}\sepDG{0}$
with value digits, separating digits and sign digits, in italics, bold and normal font, respectively (the sign bit for zero is arbitrarily set to $0$ in this case).
As a result, the code length for a sequence of shift corrections $\Csc$ is
\custmath{\cl(\Csc) =  2\abs{\Csc} + \sum_{e \in \Csc} \abs{e}}{.}

\VLongOnly{\medskip}

Putting everything together, we can write the cost of a cycle $\cycle$ as
\begin{align*}
\cl(\cycle) \VLongOnly{=& \cl(\Cev) + \cl(\Clen) + \cl(\Cprd) + \cl(\Cto) +\cl(\Csc) \\}
=& \log(\len{\seq}) + \log\big(\big\lfloor\frac{\tspan{\seq}-\sumel{\Csc}}{\Clen-1}\big\rfloor\big) \\
&+ \log(\tspan{\seq} - \sumel{\Csc} - (\Clen-1)\Cprd + 1) \VLongOnly{\\
&}+ 2\abs{\Csc} + \sum_{e \in \Csc} \abs{e}
   \;.
\end{align*}

On the other hand, the cost of an individual occurrence $o = (t, \alpha)$ is simply the sum of the cost of the corresponding timestamp and event: 
\[\cl(o) = \cl(t)+\cl(\alpha) = \log(\tspan{\seq}+1) -\log(\custfrac{\len{\seq[\alpha]}}{\len{\seq}})\;.\]

Note that if our goal was to actually encode the input sequence, we would need to transmit the smallest and largest timestamps ($\tSstart(\seq)$ and $\tSend(\seq)$), the size of the event alphabet ($\abs{\ABC}$), as well as  the number of occurrences of each event ($\len{\seq[\alpha]}$ for each event $\alpha$) of the event sequence. We should also transmit the number of cycles in the collection ($\abs{\ccycle}$), which can be done, for instance with a code word of length $\log(\len{\seq})$. 
However, since our goal is to compare collections of cycles, we can simply ignore this, as it represents a fixed cost that remains constant for any chosen collection of cycles.

Finally, consider that we are given an ordered list of occurrences $\LL{t_1, t_2, \dots , t_l}$ of event $\alpha$, and we want to determine the best cycle with which to cover all these occurrences at once. Some of the parameters of the cycle are determined, namely the repeating event $\Cev$, the length $\Clen$, and the timestamp of the first occurrence~$\Cto$.
All we need to determine is the period $p$ that yields the shortest code length for the cycle. In particular, we want to find $p$ that minimises $\cl(\Csc)$.
The shift corrections are such that 
$\Csc_{k} = (t_{k+1} - t_{k}) - p$\VLongOnly{ (cf.\ the definition of a cycle's cover)}.
If we consider the list of inter-occurrence distances $d_1=t_{2} - t_{1}, d_2=t_{3} - t_{2}, \dots, d_{l-1} = t_{l} - t_{l-1}$, the problem of finding $p$ that minimises $\cl(\Csc)$ boils down to minimising 
$\sum_{d_i} \abs{d_i - p}.$ This is achieved by letting $p$ equal the geometric median of the inter-occurrence distances, which, in the one-dimensional case, is simply the median. 
Hence, for this choice of encoding for the shift corrections, the optimal cycle covering a list of occurrences can be determined by simply computing the inter-occurrences distances and taking their median as the cycle period.

\section{Defining Tree Patterns}
\label{sec:problem_complex}

So far, our pattern language is restricted to cycles over single events. In practise, however, several events might recur regularly together and repetitions might be nested with several levels of periodicity. To handle such cases, we now introduce a more expressive pattern language, that consists of a hierarchy of cyclic blocks, organised as a tree. 

Instead of considering simple cycles specified as 5-tuples $\cycle = (\Cev, \Clen, \Cprd, \Cto, \Csc)$ we consider more general patterns specified as triples $\patt\label{sym:patt} = (\Ptree, \Cto, \Csc)$, where $\Ptree$\label{sym:ptree} denotes the tree representing the hierarchy of cyclic blocks, while $\Cto$ and $\Csc$ respectively denote the starting point and shift corrections of the pattern, as with cycles. 

\mpara{Pattern trees.}
Each \emph{leaf node} in a pattern tree represents a simple block containing one event. Each \emph{intermediate node} represents a cycle in which the children nodes repeat at a fixed time interval. In other words, each intermediate node represents cyclic repetitions of a sequence of blocks.
The root of a pattern tree is denoted as $\Pblock_0$\label{sym:Pblock}. Using list indices, we denote the children of a node $\Pblock_{\Bid}$ as $\Pblock_{\Bid{}1}$, $\Pblock_{\Bid{}2}$, etc. \VLongOnly{We denote the ordered list of the children of node $\Pblock_{\Bid}$ as $\child{\Pblock_{\Bid}}$\label{sym:child}, that is, 
\custmath{\child{\Pblock_{\Bid}} = \LL{\Pblock_{\Bid{}1}, \Pblock_{\Bid{}2}, \dots}}{.}}
All children of an intermediate node except the \VLongOnly{left-most child}\VShortOnly{first one} are associated to their distance to the preceding child, called the \emph{inter-block distance}\VShortOnly{, denoted as $d_{\Bid{}}$ for node $\Pblock_{\Bid{}}$.}\VLongOnly{. This distance for node $\Pblock_{\Bid{}i}$ is denoted as $d_{\Bid{}i}$\label{sym:interBd}, i.e.\ $d_{\Bid{}i}$ represents the time that separates occurrences of node $\Pblock_{\Bid{}(i-1)}$ and node $\Pblock_{\Bid{}i}$.} Inter-block distances take non-negative integer values. 
Each intermediate node $\Pblock_{\Bid}$ is associated with the period $\Cprd_{\Bid}$ and length $\Clen_{\Bid}$ of the corresponding cycle.
Each leaf node $\Pblock_{Y}$ is associated with the corresponding occurring event $\alpha_{Y}$.

An example of an abstract pattern tree is shown in Fig.~\ref{fig:ex_tree_abs}. 
\VLongOnly{Some concrete pattern trees that we will use as examples are shown in Fig.~\ref{fig:ex_tree1a}--\ref{fig:ex_tree2}.}
We call \emph{height} and \emph{width} of the pattern tree---and by extension of the associated pattern---respectively the number of edges along the longest branch from the root to a leaf node and the number of leaf nodes in the tree.

\VLongOnly{
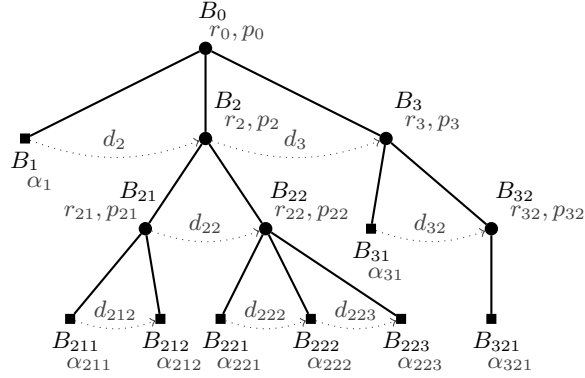
\begin{figure}
\centering
\begin{tikzpicture}[-,auto,node distance=1.2cm, thick]
  \node[main node] (b0) {};
  \node[main node] (b2) [below of=b0] {};
  \node[leaf node] (b1) [left of=b2, node distance=2.4cm] {};
  \node[main node] (b3) [right of=b2, node distance=2.4cm] {};

  \node[ghost node] (b2x) [below of=b2] {};
  \node[main node] (b21) [left of=b2x, node distance=.8cm] {};
  \node[main node] (b22) [right of=b2x, node distance=.8cm] {};

  \node[ghost node] (b3x) [below of=b3, xshift=0.6cm] {};
  \node[leaf node] (b31) [left of=b3x, node distance=.8cm] {};
  \node[main node] (b32) [right of=b3x, node distance=.8cm] {};

  \node[leaf node] (b212) [below of=b21, xshift=.2cm] {};
  \node[leaf node] (b211) [left of=b212, node distance=1.2cm] {};

  \node[leaf node] (b222) [below of=b22, xshift=.6cm] {};
  \node[leaf node] (b221) [left of=b222, node distance=1.2cm] {};
  \node[leaf node] (b223) [right of=b222, node distance=1.2cm] {};

  \node[leaf node] (b321) [below of=b32] {};

  \node[label node] (l0) [above of=b0, xshift=-0.2cm] {\BlockMark{0}};
  \node[label node] (l2) [above of=b2] {\BlockMark{2}};
  \node[label node] (l3) [above of=b3] {\BlockMark{3}};
  \node[label node] (l21) [above of=b21, xshift=-.4cm] {\BlockMark{21}};
  \node[label node] (l22) [above of=b22] {\BlockMark{22}};
  \node[label node] (l32) [above of=b32] {\BlockMark{32}};

  \node[lterm node] (l1) [below of=b1] {\BlockMark{1}};
 
  \node[lterm node] (l31) [below of=b31] {\BlockMark{31}};

  \node[lterm node] (l212) [below of=b212] {\BlockMark{212}};
  \node[lterm node] (l211) [below of=b211] {\BlockMark{211}};

  \node[lterm node] (l222) [below of=b222] {\BlockMark{222}};
  \node[lterm node] (l221) [below of=b221] {\BlockMark{221}};
  \node[lterm node] (l223) [below of=b223] {\BlockMark{223}};

  \node[lterm node] (l321) [below of=b321] {\BlockMark{321}};

  \node[prop node] (p0) [below of=l0] {$r_0, p_0$};
  \node[prop node] (p2) [below of=l2] {$r_2, p_2$};
  \node[prop node] (p3) [below of=l3] {$r_3, p_3$};
  \node[prop node] (p21) [below of=l21, anchor=east, xshift=-0.6cm] 
{$r_{21}, p_{21}$};
  \node[prop node] (p22) [below of=l22, anchor=east, xshift=0.2cm] 
{$r_{22}, p_{22}$};
  \node[prop node] (p32) [below of=l32] {$r_{32}, p_{32}$};

  \node[pterm node] (p1) [below of=l1] {$\alpha_1$};
  \node[pterm node] (p31) [below of=l31] {$\alpha_{31}$};

  \node[pterm node] (p212) [below of=l212] {$\alpha_{212}$};
  \node[pterm node] (p211) [below of=l211] {$\alpha_{211}$};

  \node[pterm node] (p222) [below of=l222] {$\alpha_{222}$};
  \node[pterm node] (p221) [below of=l221] {$\alpha_{221}$};
  \node[pterm node] (p223) [below of=l223] {$\alpha_{223}$};
  \node[pterm node] (p321) [below of=l321] {$\alpha_{321}$};

  \path
    (b0) edge (b1)
    (b0) edge (b2)
    (b0) edge (b3)
    (b2) edge (b21)
    (b2) edge (b22)
    (b3) edge (b31)
    (b3) edge (b32)
    (b21) edge (b211)
    (b21) edge (b212)
    (b22) edge (b221)
    (b22) edge (b222)
    (b22) edge (b223)
    (b32) edge (b321);

\path[dotted, thin, ->, bend right=20, color=darkgray]
    (b1) edge node[above] {$d_{2}$} (b2)
    (b2) edge node[above] {$d_{3}$} (b3)
    (b21) edge node[above] {$d_{22}$} (b22)
    (b31) edge node[above] {$d_{32}$} (b32)
    (b211) edge node[above] {$d_{212}$} (b212)
    (b221) edge node[above] {$d_{222}$} (b222)
    (b222) edge node[above] {$d_{223}$} (b223);

\end{tikzpicture}
\caption{Abstract pattern tree.}
\label{fig:ex_tree_abs}
\end{figure}
}

For a given pattern, we can construct a tree of event occurrences by expanding the pattern tree recursively, that is, by appending to each intermediate node the corresponding number of copies of the associated subtree, recursively.
We call this expanded tree the \emph{expansion tree} of the pattern, as opposed to the contracted \emph{pattern tree} that more concisely represents the pattern.

\VLongOnly{When a pattern tree is expanded, several copies of a node can be generated as a result of repetitions in possibly nested cycles. Each node in an expansion is identified with a pair $(n,L)$, where $n$ is the node of the pattern tree that generated the expansion node, and $L$ is a list indicating the specific combination of repetitions of ancestors that produced it.

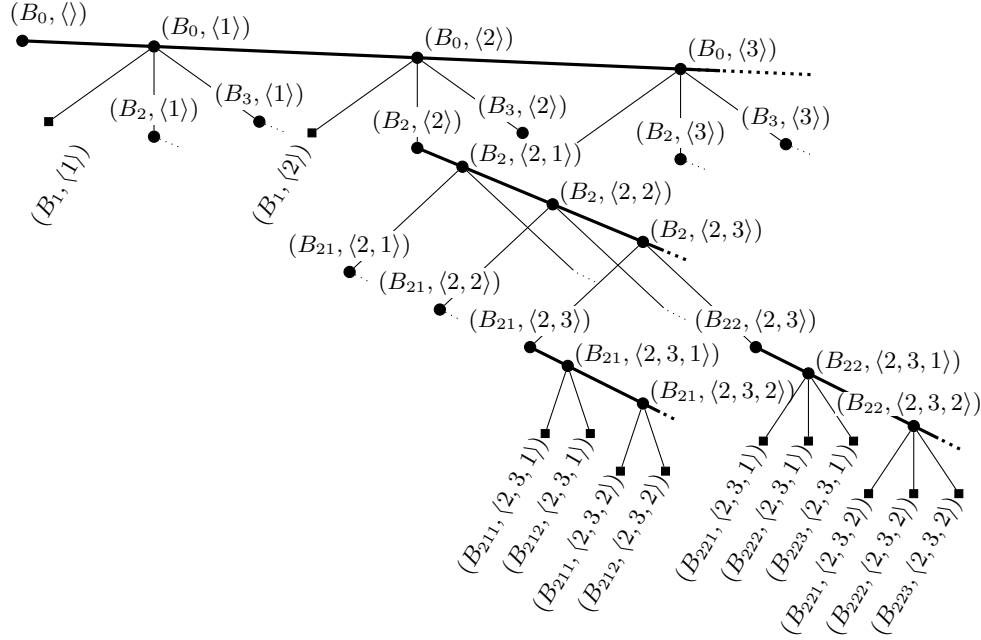
\begin{figure*}
\centering
\newcommand{\ndStylA}{main node}
\newcommand{\ndStylB}{main node}
\newcommand{\ndStylXA}{main node}
\newcommand{\ndStylXB}{main node}
\newcommand{\ndStylC}{main node}
\newcommand{\xAslant}{1.}
\newcommand{\yAslant}{-.5}
\newcommand{\xBslant}{1.4}
\newcommand{\yBslant}{-.7}
\renewcommand{\xslant}{1.2}
\renewcommand{\yslant}{-.5}
\renewcommand{\Xslant}{3.5}
\renewcommand{\Yslant}{-.15}
\newcommand{\nbRepsA}{3}
\newcommand{\nbRepsB}{3}
\newcommand{\nbRepsCA}{2}
\newcommand{\nbRepsCB}{2}
\begin{tikzpicture}[-,auto,node distance=1.2cm]
  \path[rep edge] (0,0) edge ({((\nbRepsA+.15)-\rootf)*\Xslant},{((\nbRepsA+.15)-\rootf)*\Yslant});
  \path[rep edge, dotted] 
 ({(\nbRepsA-\rootf)*\Xslant},{(\nbRepsA-\rootf)*\Yslant}) edge ({((\nbRepsA+.5)-\rootf)*\Xslant},{((\nbRepsA+.5)-\rootf)*\Yslant});

  \node[main node] (R00) at (0,0) {};
  \node[lmain node, xshift=10pt] (R00l) [above of=R00] {\OccMark{$(\Pblock_0, \llempty)$}};
  \foreach \j in {1,...,\nbRepsA}{

    \renewcommand{\ndStylA}{leaf node}
    \renewcommand{\ndStylB}{main node}
    \renewcommand{\ndStylC}{main node}
    \ifthenelse{\j=3}{
      \renewcommand{\ndStylA}{ghost node}
    }{}

    \node[main node] (R{\j}b0) at ({(\j-\rootf)*\Xslant},{(\j-\rootf)*\Yslant}) {};
    \node[\ndStylB] (R{\j}b2) [below of=R{\j}b0] {};
    \node[\ndStylA] (R{\j}b1) [left of=R{\j}b2, node distance=1.4cm, yshift=.2cm] {};
    \node[\ndStylC] (R{\j}b3) [right of=R{\j}b2, node distance=1.4cm, yshift=.2cm] {};

    \path
    (R{\j}b0) edge (R{\j}b1)
    (R{\j}b0) edge (R{\j}b2)
    (R{\j}b0) edge (R{\j}b3);

      \path[dotted]
        (R{\j}b2) edge ([shift={({.3*\xslant},{.3*\yslant})}]R{\j}b2);
    \ifthenelse{\j=2}{}{
      \path[dotted]
        (R{\j}b3) edge ([shift={({.3*\xslant},{.3*\yslant})}]R{\j}b3);
    }

    \node[li node, yshift=.1cm] (R{\j}l0) at (R{\j}b0)  {\OccMark{$(\Pblock_0, \LL{\j})$}};
    \ifthenelse{\j=1}{
      \node[l\ndStylA, xshift=4pt, yshift=-16pt] (R{\j}l1) [above of=R{\j}b1] {\OccMark{$(\Pblock_1, \LL{\j})$}};
    }{\node[l\ndStylA] (R{\j}l1) [above of=R{\j}b1] {\OccMark{$(\Pblock_1, \LL{\j})$}};}
    \node[l\ndStylB] (R{\j}l2) [above of=R{\j}b2] {\OccMark{$(\Pblock_2, \LL{\j})$}};
    \node[l\ndStylC] (R{\j}l3) [above of=R{\j}b3] {\OccMark{$(\Pblock_3, \LL{\j})$}};

    \ifthenelse{\j=2}{
      \node[ghost node] (Ya) at (R{\j}b2) {};
      \node[ghost node] (Yb) at ([shift={({.1*\xslant},{.1*\yslant})}]R{\j}b2) {};
      \node[ghost node] (Yc) at ([shift={({.25*\xslant},{.25*\yslant})}]R{\j}b2) {};
      \node[ghost node] (Yd) at ([shift={({(\nbRepsB+.25-\rootf)*\xslant},{(\nbRepsB+.25-\rootf)*\yslant})}]R{\j}b2) {};
      \node[ghost node] (Ye) at ([shift={({(\nbRepsB+.5-\rootf)*\xslant},{(\nbRepsB+.5-\rootf)*\yslant})}]R{\j}b2) {};
      \path[rep edge] (Ya) edge (Yd);
      \path[rep edge, dotted] (Yd) edge (Ye);

      \foreach \k in {1,...,\nbRepsB}{
        \node[main node] (R{\j}{\k}b2) at ([shift={({(\k-\rootf)*\xslant},{(\k-\rootf)*\yslant})}]R{\j}b2) {};
        \node[li node] (R{\j}{\k}b2l) at (R{\j}{\k}b2)  {\OccMark{$(\Pblock_2, \LL{\j, \k})$}};

        \renewcommand{\ndStylXA}{main node}
        \ifthenelse{\k=3}{
          \renewcommand{\ndStylXB}{main node}
        }{
          \renewcommand{\ndStylXB}{ghost node}
        }

        \node[ghost node] (R{\j}{\k}b2x) [below of=R{\j}{\k}b2, node distance=1.4cm] {};
        \node[\ndStylXA] (R{\j}{\k}b21) [left of=R{\j}{\k}b2x, node distance=1.5cm] {};
        \node[\ndStylXB] (R{\j}{\k}b22) [right of=R{\j}{\k}b2x, node distance=1.5cm] {};

        \path
        (R{\j}{\k}b2) edge (R{\j}{\k}b21)
        (R{\j}{\k}b2) edge (R{\j}{\k}b22);

        \ifthenelse{\k=3}{}{
        \path[dotted]
        (R{\j}{\k}b21) edge ([shift={({.3*\xAslant},{.3*\yAslant})}]R{\j}{\k}b21)
        (R{\j}{\k}b22) edge ([shift={({.3*\xAslant},{.3*\yAslant})}]R{\j}{\k}b22);
        }

        \node[l\ndStylXA] (R{\j}{\k}l21) [above of=R{\j}{\k}b21] {\OccMark{$(\Pblock_{21}, \LL{\j, \k})$}};
        \node[l\ndStylXB] (R{\j}{\k}l22) [above of=R{\j}{\k}b22] {\OccMark{$(\Pblock_{22}, \LL{\j, \k})$}};

        \ifthenelse{\k=3}{

          \node[ghost node] (YAa) at (R{\j}{\k}b21) {};
          \node[ghost node] (YAb) at ([shift={({.1*\xAslant},{.1*\yAslant})}]R{\j}{\k}b21) {};
          \node[ghost node] (YAc) at ([shift={({.25*\xAslant},{.25*\yAslant})}]R{\j}{\k}b21) {};
          \node[ghost node] (YAd) at ([shift={({(\nbRepsCA+.25-\rootf)*\xAslant},{(\nbRepsCA+.25-\rootf)*\yAslant})}]R{\j}{\k}b21) {};
          \node[ghost node] (YAe) at ([shift={({(\nbRepsCA+.5-\rootf)*\xAslant},{(\nbRepsCA+.5-\rootf)*\yAslant})}]R{\j}{\k}b21) {};
          \path[rep edge] (YAa) edge (YAd);
          \path[rep edge, dotted] (YAd) edge (YAe);

          \foreach \l in {1,...,\nbRepsCA}{
            \node[main node] (R{\j}{\k}{\l}b21) at ([shift={({(\l-\rootf)*\xAslant},{(\l-\rootf)*\yAslant})}]R{\j}{\k}b21) {};

            \node[ghost node] (R{\j}{\k}{\l}b21x) [below of=R{\j}{\k}{\l}b21, node distance=.9cm] {};
            \node[leaf node] (R{\j}{\k}{\l}b212) [right of=R{\j}{\k}{\l}b21x, node distance=.3cm] {};
            \node[leaf node] (R{\j}{\k}{\l}b211) [left of=R{\j}{\k}{\l}b21x, node distance=.3cm] {};

            \node[li node] (R{\j}{\k}{\l}b21l) at (R{\j}{\k}{\l}b21)  {\OccMark{$(\Pblock_{21}, \LL{\j, \k, \l})$}};
            \node[le node] (R{\j}{\k}{\l}l211) [below of=R{\j}{\k}{\l}b211] {\OccMark{$(\Pblock_{211}, \LL{\j, \k, \l})$}};
            \node[le node] (R{\j}{\k}{\l}l212) [below of=R{\j}{\k}{\l}b212] {\OccMark{$(\Pblock_{212}, \LL{\j, \k, \l})$}};

            \path
            (R{\j}{\k}{\l}b21) edge (R{\j}{\k}{\l}b211)
            (R{\j}{\k}{\l}b21) edge (R{\j}{\k}{\l}b212);

            }

          \node[ghost node] (YBa) at (R{\j}{\k}b22) {};
          \node[ghost node] (YBb) at ([shift={({.1*\xBslant},{.1*\yBslant})}]R{\j}{\k}b22) {};
          \node[ghost node] (YBc) at ([shift={({.25*\xBslant},{.25*\yBslant})}]R{\j}{\k}b22) {};
          \node[ghost node] (YBd) at ([shift={({(\nbRepsCB+.25-\rootf)*\xBslant},{(\nbRepsCB+.25-\rootf)*\yBslant})}]R{\j}{\k}b22) {};
          \node[ghost node] (YBe) at ([shift={({(\nbRepsCB+.5-\rootf)*\xBslant},{(\nbRepsCB+.5-\rootf)*\yBslant})}]R{\j}{\k}b22) {};
          \path[rep edge] (YBa) edge (YBd);
          \path[rep edge, dotted] (YBd) edge (YBe);

          \foreach \m in {1,...,\nbRepsCB}{
            \node[main node] (R{\j}{\k}{\m}b22) at ([shift={({(\m-\rootf)*\xBslant},{(\m-\rootf)*\yBslant})}]R{\j}{\k}b22) {};

            \node[leaf node] (R{\j}{\k}{\m}b222) [below of=R{\j}{\k}{\m}b22, node distance=.9cm] {};
            \node[leaf node] (R{\j}{\k}{\m}b221) [left of=R{\j}{\k}{\m}b222, node distance=.6cm] {};
            \node[leaf node] (R{\j}{\k}{\m}b223) [right of=R{\j}{\k}{\m}b222, node distance=.6cm] {};

            \ifthenelse{\m=2}{
              \node[li node, yshift=4pt, xshift=-32pt] (R{\j}{\k}{\m}b22l) at (R{\j}{\k}{\m}b22)  {\OccMark{$(\Pblock_{22}, \LL{\j, \k, \m})$}};}{
              \node[li node] (R{\j}{\k}{\m}b22l) at (R{\j}{\k}{\m}b22)  {\OccMark{$(\Pblock_{22}, \LL{\j, \k, \m})$}};
            }

            \node[le node] (R{\j}{\k}{\m}l222) [below of=R{\j}{\k}{\m}b222] {\OccMark{$(\Pblock_{222}, \LL{\j, \k, \m})$}};
            \node[le node] (R{\j}{\k}{\m}l221) [below of=R{\j}{\k}{\m}b221] {\OccMark{$(\Pblock_{221}, \LL{\j, \k, \m})$}};
            \node[le node] (R{\j}{\k}{\m}l223) [below of=R{\j}{\k}{\m}b223] {\OccMark{$(\Pblock_{223}, \LL{\j, \k, \m})$}};

            \path
            (R{\j}{\k}{\m}b22) edge (R{\j}{\k}{\m}b221)
            (R{\j}{\k}{\m}b22) edge (R{\j}{\k}{\m}b222)
            (R{\j}{\k}{\m}b22) edge (R{\j}{\k}{\m}b223);

            }
        }{}%
      }      
    }{}%
  }
\end{tikzpicture}
\caption{Expansion of the pattern tree from Fig.~\ref{fig:ex_tree_abs}.}
\label{fig:ex_exp_abs}
\end{figure*}

The expansion tree of the pattern tree of Fig.~\ref{fig:ex_tree_abs} is shown in Fig.~\ref{fig:ex_exp_abs}. Node $(\Pblock_0, \llempty)$ is the root of the expansion tree, $(\Pblock_{0}, \LL{1})$ is the node generated as the first repetition of pattern node $\Pblock_{0}$, and $(\Pblock_{21}, \LL{2, 3})$ is the node generated from node $\Pblock_{21}$ in the third repetition of pattern node $\Pblock_{2}$ nested within the second repetition of pattern node $\Pblock_0$.

The notation used to identify nodes in pattern trees and expansion trees allows to easily navigate the trees.
In particular, the left-most leaf among the descendants of a given node $\Pblock_{\Bid}$ can be obtained by going down the left-most branch, looking at nodes $\Pblock_{\Bid1}$, $\Pblock_{\Bid11}$, etc.\ until reaching a leaf. We denote that node, the left-most leaf descendant of $\Pblock_{\Bid}$ as $\Lchild{\Pblock_{\Bid}}$\label{sym:Lchild}. 
Similarly, we denote as $\Lchild{(n, L)}$ the left-most leaf descendant of node $(n, L)$ in the expansion tree, which is such that $\Lchild{(n, L)} = (\Lchild{n}, L')$, where $L' = L \opconcat \LL{1,1\dots}$, that is, $L'$ is the list $L$ trailing with ones. That is, in addition to selecting always the left-most child, we always select the first repetition of a node when travelling the expansion tree until reaching a leaf. 
Note that $\Lchild{\Pblock_{\Bid}} = \Pblock_{\Bid}$ and $\Lchild{(\Pblock_{\Bid}, L)} = (\Pblock_{\Bid}, L)$ if $\Pblock_{\Bid}$ itself is a leaf node.

}\VLongOnly{We use the recursive notation $\BinfoRP{r_{\Bid}}{p_{\Bid}} \Bstart{}\activity{\Pblock_{\Bid1}}\,\BinfoD{d_{\Bid{}2}}\, \activity{\Pblock_{\Bid2}} \dots\Bend{}$ to represent a block $\Pblock_{\Bid}$.
With this notation, $\Ptree_1$ from Fig.~\ref{fig:ex_tree1a} is represented as  
\[\BinfoRP{4}{2} \Bstart{}\activity{a}\Bend{}\] 
and $\Ptree_7$ from Fig.~\ref{fig:ex_tree2} as  
\[\BinfoRP{3}{10} \Bstart{}\activity{b} \BinfoD{3} \BinfoRP{4}{1} \Bstart{}\activity{a}\Bend{} \BinfoD{1} \activity{c}\Bend{}\;.\] 

\mpara{Reconstructing a pattern's cover.}}
We can enumerate the event occurrences of a pattern by traversing its expansion tree and recording the encountered leaf nodes. \VLongOnly{The expansion tree is traversed in a depth-first left-to-right manner, first travelling through all children in a repetition of a block before moving on to the next repetition. }\VShortOnly{We denote as $\occsStar{\patt}$ this list of timestamp--event pairs reconstructed from the tree, prior to correction.}
\VLongOnly{For instance, the traversal of the expansion tree shown in Fig.~\ref{fig:ex_exp_abs}, starts from the root node $(\Pblock_0 , \llempty{})$ and first reaches $(\Pblock_0 , \LL{1})$. Then, children nodes $(\Pblock_1, \LL{1})$, $(\Pblock_2, \LL{1})$ and $(\Pblock_3, \LL{1})$, and their descendants, should be traversed before travelling to the next repetition of $\Pblock_0$, $(\Pblock_0 , \LL{2})$. Simply put, pattern edges (represented as thin lines in Fig.~\ref{fig:ex_exp_abs}) take priority over repetition edges (represented as thick lines).

We define the following recursive function:
\[ \mapOids{\Pblock_{\Bid}, l}\label{sym:mapOids} = \left\{ \begin{array}{l@{}l}
     \LL{\Pblock_{\Bid}, l} & \quad \mbox{if $\Pblock_{\Bid}$ is a leaf},\\[.5em]
     \fnconcat_{k \in [1..\Clen_{\Bid}]} \fnconcat_{\Pblock_{\Bid{}i} \in \child{\Pblock_{\Bid}}} & \mapOids{\Pblock_{\Bid{}i}, l \opconcat \LL{k-1}} \\  & \hfill \mbox{otherwise}.\end{array} \right. \] 
The list of leaf nodes encountered in the expansion tree during the traversal can be obtained as $\mapOids{\Ptree} = \mapOids{\Pblock_0, \llempty}$.

Using a similar recursive function, following the same traversal of the expansion tree, we can construct the perfect event occurrences. That is, we can recursively construct the list of uncorrected timestamps--events pairs produced by a pattern tree $\Ptree$, which we denote as $\occsStar{\Ptree} = \occsStar{\Pblock_0}$.

For this purpose, we first define a function $\shift{S}{t_s}$\label{sym:shift} that shifts a set of event occurrences $S$ by a specified value $t_s$, that is, 
\[ \shift{S}{t_s} = \{(t_i+t_s, \alpha_i), \quad \forall (t_i, \alpha_i) \in S\}.\]
For instance
\begin{align*}
 \opshift(&\LL{(2,c),(3,c),(6,a),(7,a)},-1) \\ =& \LL{(1,c),(2,c),(5,a),(6,a)}.
\end{align*}

Overloading the notation, we let $\occsStar{\Pblock_{\Bid}}$\label{sym:occStar} denote the list of occurrences associated with $\Pblock_{\Bid}$.
If $\Pblock_{\Bid}$ is a leaf, $\occsStar{\Pblock_{\Bid}}$ is a one-element list
\[ \occsStar{\Pblock_{\Bid}} = \LL{(0, \alpha_{\Bid})}.\]
If $\Pblock_{\Bid}$ is an intermediate node, we let $O(\Pblock_{\Bid})$ denote the concatenation of the lists of occurrences of its children, each one shifted by the accumulated inter-block distances:
\[
O(\Pblock_{\Bid}) = \fnconcat_{\Pblock_{\Bid{}i} \in \child{\Pblock_{\Bid}}} \shift{\occsStar{\Pblock_{\Bid{}i}}}{ \sum_{1 < j \leq i} d_{\Bid{}j} }\;.
\]
Then the list of occurrences is obtained by concatenating $\Clen_{\Bid}$ copies of $O(\Pblock_{\Bid})$, shifted according to the period $\Cprd_{\Bid}$:
\[
\occsStar{\Pblock_{\Bid}} = \fnconcat_{k \in [1..\Clen_{\Bid}]} \shift{ O(\Pblock_{\Bid}) }{ (k-1) \cdot \Cprd_{\Bid} }\;.
\]
Finally, if the starting point of pattern $\patt$ is $\tau$, we have $\occsStar{\patt} = \shift{\occsStar{\Ptree}}{\tau}$.

The occurrences appear in the list in the order in which they are generated during the expansion, which does not necessarily match the order of the timestamps.
More specifically, if the sequence of timestamps in $\occsStar{\Ptree}$ is not monotone, we say that the pattern tree $\Ptree$ (and the associated pattern $\patt$) is \emph{interleaved}.
If a pattern tree is not interleaved, all events constituting a repetition of a block must occur at latest when an event of the following repetition occurs. If several events occur at the same time, we say that the pattern tree has \emph{overlaps}.
For example, pattern trees $\Ptree_3$ and $\Ptree_4$ cover the same occurrences, but $\Ptree_4$ is interleaved while $\Ptree_3$ is not. Both patterns $\Ptree_6$ and $\Ptree_7$ have overlaps, but $\Ptree_7$ is interleaved while $\Ptree_6$ is not.

We denote as $o_i$ the $i^\text{th}$ event occurrence generated by $\Ptree$, and let $\occsStar{o_i}$ be the corresponding timestamp--event pair and $\mapOids{o_i}$ be the corresponding expansion leaf node, i.e.\ mapping $o_i$ to the elements at position $i$ in $\occsStar{\Ptree}$ and $\mapOids{\Ptree}$, respectively.}

As for the simple cycles, we will not only consider perfect patterns but will allow some variations. For this purpose, a list of shift corrections $\Csc$ is provided with the pattern, which contains a correction for each occurrence except the first one, i.e.\ $\abs{\Csc} = \abs{\occsStar{\patt}}-1$. \VLongOnly{

By applying the shift corrections in $\Csc$ to the perfect occurrences in $\occsStar{\patt}$, we can generate the list of corrected occurrences for pattern $\patt$, denoted as $\occs{\patt}$\label{sym:occs}.
The corrections are listed in $\Csc$ in the same order as the leaf nodes are encountered in the expansion tree. Therefore, the correction associated to occurrence $o_i$ is the element at position $i-1$ in $\Csc$, i.e.\ $\Csc[i-1]$, which we also denote as $\Csc(o_i)$ or $\Csc((n,L))$, where $(n,L)$ is the corresponding expansion node. For ease of notation we let $\Csc(o_1)=0$, since the left most occurrence $o_1$ has no correction.

}However, as for simple cycles, corrections accumulate over successive occurrences, and we cannot recover the list of corrected occurrences $\occs{\patt}$ by simply adding the individual corrections to the elements of $\occsStar{\patt}$.
Instead, we first have to compute the accumulated corrections for each occurrence.
\VLongOnly{In addition to its own correction, the corrections that should be applied to an occurrence come from the offsets of its left siblings in multi-events blocks and the offsets of previous repetitions in cycles the occurrence belongs to.

Algorithm~\ref{alg:coco} shows the procedure---named \algCoCo{}---that can be used to collect the occurrences whose individual corrections impact occurrence $o$ (recall that $\Lchild{}$ returns the left-most leaf descendant of a node). Then, the correction to be applied to the timestamp of $o$ is \[\cume{o}\label{sym:cume} = \Csc(o) + \sum_{o_k \in \algCoCo(o)} \Csc(o_k)\;.\]
The corrected occurrence timestamps can thus be reconstructed by shifting the perfect timestamp by the corresponding correction, i.e.\ $\occs{o_i} = \occsStar{o_i} +\cume{o_i}$.

\begin{algorithm}[tb]
\caption{\algCoCo: Collect occurrence corrections.}
\label{alg:coco}
\begin{algorithmic}[1]
\Require An occurrence $o$
\Ensure A set of occurrences whose corrections apply to $o$
\If{$o = (\Pblock_0, \LL{})$}\Comment{Root of pattern}
\State $\omega \gets \emptyset$
\EndIf
\If{$o = (\Pblock_{Xy}, Uv)$}
\State $\omega \gets \{ \Lchild{(\Pblock_{Xy'}, Uv)},\, y' < y \}$ 
\Comment{Left-siblings}
\State $\omega \gets \omega \,\cup\,\, \{ \Lchild{(\Pblock_{X}, Uv')}, v' < v \}$
\Comment{Previous repetitions}
\State $\omega \gets \omega \,\cup\, $\algCoCo$((\Pblock_{X}, U))$ 
\Comment{Recurse for parent}
\EndIf
\State \textbf{return} $\omega$
\end{algorithmic}
\end{algorithm}
}

\mpara{Encoding the patterns.}
To transmit a pattern, we need to encode its pattern tree, as well as its starting point and shift corrections. 
Furthermore, to encode the pattern tree, we consider separately its event sequence, its cycle lengths, its top-level period, and the other values, as explained below.

First we encode the event in the leaves of the pattern tree, traversing the tree from left to right, depth-first\VLongOnly{, enclosing blocks between parenthesis}. \VLongOnly{The string representing the events in the pattern tree is defined recursively as follows:
\[ \evtseqfun{\Pblock_{\Bid}} = \left\{ \begin{array}{ll}
      \text{`}\Cev_{\Bid}\text{'} & \mbox{if $\Pblock_{\Bid}$ is a leaf},\\
     \text{`('} \opconcat \big( \fnconcat_{\Pblock_Y \in \child{\Pblock_{\Bid}}} \evtseqfun{\Pblock_{Y}} \big) \opconcat \text{`)'} & \mbox{otherwise}.\end{array} \right. \]}
We denote as $\evtseq$\label{sym:evtseq} the string\VLongOnly{ $\evtseqfun{\Pblock_0}$ for the top-level block of the tree of a pattern,} representing its event sequence.
We encode each symbol $s$ in the string $\evtseq$ using a code of length $\cl(s)$, where $\cl(s)$ depends on the frequency of $s$, adjusted to take into account the additional symbols `(' and `)', used to delimit blocks.
\VLongOnly{In particular, we set the code length for the extended alphabet as
\[\cl(\text{`('})=\cl(\text{`)'}) = -\log(\frac{1}{3})\]
for the block delimiters, and 
\[\cl(\text{`)'}) = -\log\big(\frac{\len{\seq[\Cev]}}{3\len{\seq}}\big)\]
for the original events.}
\VShortOnly{ In particular, we set the code lengths for the extended alphabet such that
$\cl(\text{`('})=\cl(\text{`)'}) = -\log(1/3)$
for the block delimiters, and 
$\cl(\alpha) = -\log(\len{\seq[\Cev]}/(3\len{\seq}))$
for the original events.}

Next, we encode the cycle lengths, i.e.\ the values $\Clen_{\Bid}$ associated to each intermediate node $\Pblock_{\Bid}$ encountered while traversing the tree depth-first and from left to right, as a sequence of values, and denote this sequence $\lens$. 
For a block $\Pblock_{\Bid}$ the number of repetitions of the block cannot be larger than the number of occurrences of the least frequent event participating in the block\VShortOnly{, denoted as $\lensfun{\Pblock_{\Bid}}$}. \VLongOnly{Formally, the cycle length $\Clen_{\Bid}$ of a block $\Pblock_{\Bid}$, can take at most a value $\lensfun{\Pblock_{\Bid}}$ defined recursively as follows:
\[ \lensfun{\Pblock_{\Bid}} = \left\{ \begin{array}{ll}
      \abs{\seq[\Cev_{\Bid}]} & \mbox{if $\Pblock_{\Bid}$ is a leaf},\\
     \min_{\Pblock_Y \in \child{\Pblock_{\Bid}}} \lensfun{\Pblock_Y} & \mbox{otherwise},\end{array} \right. \] }
We can thus encode the sequence of cycle lengths $\lens$ with code of length
\[\cl(\lens) = \sum_{\Clen_\Bid \in \lens} \cl(\Clen_\Bid) = \sum_{\Clen_\Bid \in \lens}  \log\big(\lensfun{\Pblock_{\Bid}}\big)\;.\]

Knowing the cycle lengths $\lens$ and the structure of the pattern tree from its event sequence $\evtseq$, we can deduce the total number of events covered by the pattern\VShortOnly{.}\VLongOnly{, $N(\Pblock_0)$, using the following formula
\[ N(\Pblock_{\Bid}) = \left\{ \begin{array}{ll}
     1 & \mbox{if $\Pblock_{\Bid}$ is a leaf},\\
     \Clen_{\Bid} \cdot \sum_{\Pblock_Y \in \child{\Pblock_{\Bid}}} N(\Pblock_Y) & \mbox{otherwise}.\end{array} \right. \] 

} The shift corrections for the pattern consist of the correction to each event occurrence except the first one\VLongOnly{ (assumed not to require correction)}. This ordered list of\VLongOnly{ $N(\Pblock_{0})-1$} values can be transmitted using the same encoding as for the simple cycles.  

In simple cycles, we had a unique period characterising the distances between occurrences.
Instead, with these more complex patterns, we have a period $\Cprd_{\Bid}$ for each intermediate node $\Pblock_{\Bid}$, as well as an inter-block distance $d_{\Bid}$ for each node $\Pblock_{\Bid}$ that is not the left-most child of its parent.

First, we transmit the period of the root node of the pattern tree, $\Pblock_0$. In a similar way as with simple cycles, we can deduce the largest possible value for $\Cprd_0$ from $\Clen_0$ and $\Csc$. Since we do not know when the events within the main cycle occur, we assume what would lead to the largest possible value for $\Cprd_0$, that is, we assume that all the events within each repetition of the cycle happen at once, so that each repetition spans no time at all.
\VLongOnly{The corrections that must be taken into account are those applying to the left-most leaf of each repetition of the main cycle. 
These are exactly the corrections accumulated in $\cume{o_{za}}$ where $o_{za}$ is the first occurrence of the last repetition of the main cycle, i.e.\ $o_{za} = \Lchild{(\Pblock_0, \LL{\Clen_0})}.$

Thus we have
\[\cl(\Cprd_0) = \log\big(\Big\lfloor\frac{\tspan{\seq}-\cume{o_{za}}}{\Clen_0-1}\Big\rfloor\big)\;.\]

Once the main period is known, we can use the same principle as for simple cycles to transmit the starting point and we have
\[\cl(\Cto) = \log(\tspan{\seq} - \cume{o_{za}} - (\Clen_0-1)\Cprd_0 + 1)\;.\] 

We denote as $\tspanStar{\Pblock_\Bid{}}$\label{sym:tspanStar} the time spanned by the entire cycle of block $\Pblock_\Bid{}$, that is, the time spanned by the $\Clen_\Bid{}$ repetitions of the block. We denote as $\tspanRepStar{\Pblock_\Bid{}}$\label{sym:tspanRepStar} the time spanned by a single repetition of the block. Note that here we consider the perfect occurrences of the block, before applying the corrections. In this case all repetitions span the same time, which might no longer be true after correction.
In Fig.~\ref{fig:ex_timelinesP8} we provide a timeline schema of the first occurrences of pattern $(\Ptree_8, 0, \nullC)$, i.e.\ the pattern consisting of the pattern tree $\Ptree_8$ from Fig.~\ref{fig:ex_tree2}, with starting point $0$ and no shift corrections. We indicate the time spanned by different blocks and their maximum value assuming interleaving is not allowed.

\begin{figure}[tbp]
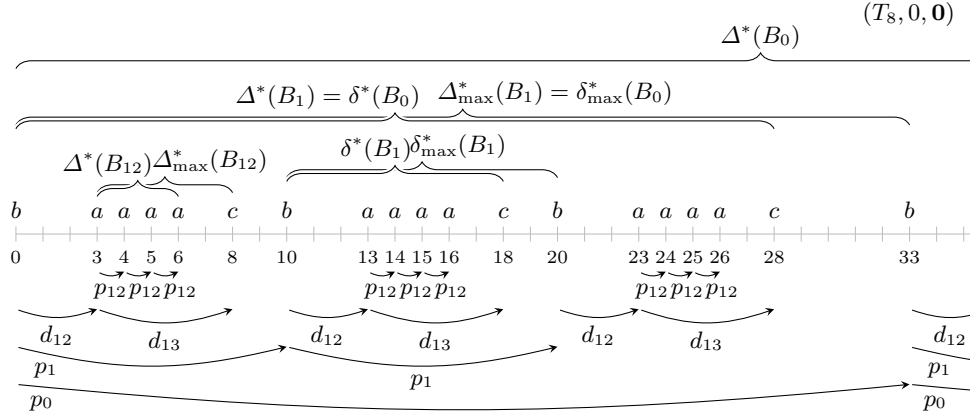

\centering
\renewcommand{\xsctm}{.36}
\timelinePEight{$(\Ptree_8, 0, \nullC)$}
\caption{Pattern $(\Ptree_8, 0, \nullC)$ partially shown on timeline (maximum time spans assume interleaving is not allowed).}
\label{fig:ex_timelinesP8}
\end{figure}

Suppose we know $\tspanStar{\Pblock_\Bid{}}$. Then, in order for $\Clen_\Bid{}$ repetitions (equally long, but potentially spanning no time at all) to happen within time $\tspanStar{\Pblock_\Bid{}}$, $\Cprd_\Bid{}$ must satisfy $\Cprd_\Bid{} \leq \lfloor\tspanStar{\Pblock_\Bid{}}/(\Clen_\Bid{}-1)\rfloor$ and can therefore be represented with a code word of length
\[\cl(\Cprd_\Bid{}) = \log\Big(\Big\lfloor\frac{\tspanStar{\Pblock_\Bid{}}}{\Clen_\Bid{}-1}\Big\rfloor\big)\;.\]

If we do not allow interleaving, each repetition can span at most $\lfloor\tspanStar{\Pblock_\Bid{}}/\Clen_\Bid{}\rfloor$, and also no longer than $\Cprd_\Bid{}$. 
On the other hand, if we do allow interleaving, each repetition can have a time span of at most $\tspanStar{\Pblock_\Bid{}} - \Clen_\Bid{} +1$.
Thus, the maximum time span of a repetition is
\[ \maxtspanRepStar{\Pblock_{\Bid}}\label{sym:maxtspanRepStar} = \left\{ \begin{array}{l}
     \tspanStar{\Pblock_\Bid{}} - \Clen_\Bid{} +1 \\ \hspace{1.5cm} \mbox{if interleaving is allowed},\\
     \min(\Cprd_\Bid{}, \lfloor\tspanStar{\Pblock_\Bid{}}/\Clen_\Bid{}\rfloor) \quad \mbox{otherwise}.\end{array} \right. \] 

Obviously, the sum of the distances between the children of the block cannot be larger than the time span of a repetition. Therefore, we can represent the distances between the children of $\Pblock_\Bid{}$ with code words such that 
\[\sum_{\Pblock_{\Bid{}i} \in \child{\Pblock_{\Bid}}, i > 1} \cl(d_{\Bid{}i}) = (\abs{\child{\Pblock_\Bid{}}}-1) \cdot \log\big(\maxtspanRepStar{\Pblock_{\Bid}} + 1\big)\;.\]

We can then determine the maximum span of each child of a block.
If interleaving is allowed, the child can span as much time as is left in the time span of its parent after accounting for the distances of the left siblings:
\[ \maxtspanStar{\Pblock_{\Bid{}i}}\label{sym:maxtspanStar} = \maxtspanRepStar{\Pblock_\Bid{}} - \sum_{1 \leq j \leq i} \interd{\Bid{}j}. \] 
Alternatively, if interleaving is not allowed, all events of the child must occur before the first event of the next sibling:
\[ \maxtspanStar{\Pblock_{\Bid{}i}} = \left\{ \begin{array}{l@{}l}
     \maxtspanRepStar{\Pblock_\Bid{}} & - \sum_{j \neq i} \interd{\Bid{}j} \\ & \quad \mbox{if $\Pblock_{\Bid{}i}$ is the right-most child},\\
     \interd{\Bid{}(i+1)} & \mbox{otherwise}.\end{array} \right. \] 
Note that $\interd{\Bid{}(i+1)}$ is not defined if $\Pblock_{\Bid{}i}$ is the right-most child of the block.

Applying the formulas above recursively allows to compute the length of the code words needed to represent all the periods and inter-block distances in the tree, for a known value $\tspanRepStar{\Pblock_0}$.

Looking at the last occurrence of the main cycle $(\Pblock_0, \LL{\Clen_0})$, we have
\[\Cto + (\Clen_0-1)\Cprd_0 + \tspanRepStar{\Pblock_0} + \cume{o_{zz}} \leq \tSend(\seq)\;,\]
and hence
\[\maxtspanRepStar{\Pblock_0} = \tSend(\seq) - \cume{o_{zz}} - (\Clen_0-1)\Cprd_0 - \Cto \;,\]
where $\cume{o_{zz}}$ denotes the accumulated corrections that apply to the event having the largest uncorrected timestamp.

If interleaving is not allowed, that event is the right-most leaf node of the expansion tree, i.e.\ the last element in the occurrence list.  Besides, if interleaving is not allowed, we also have $\tspanRepStar{\Pblock_0} \leq \Cprd_0$.

On the other hand, if interleaving is allowed the event having the largest uncorrected timestamp is not necessarily the last one in the list of occurrences (see $\occsStar{\Ptree_6}$ in Fig.~\ref{fig:ex_tree1b} for instance). Since it depends on periods and inter-block distances within the block, which have not been specified at that point, we cannot determine which event has the largest timestamp. 
Hence, we compute $\cume{o_{i}}$ for all occurrences $o_i$ that correspond to the right most child of a block and take the minimum (possibly a negative value) as $\cume{o_{zz}}$.

To compute the periods and inter-block distances, we can use the actual value $\tspanRepStar{\Pblock_0}$, which we first need to transmit explicitly after the value of $\Cto$, with a code word of length 
$\log\big( \maxtspanRepStar{\Pblock_0} + 1\big)$.
Instead, we could use the upper-bound on $\maxtspanRepStar{\Pblock_0}$, which we do not need to transmit. It is probably more economical to transmit the value explicitly.}

We denote as $\inDP$\label{sym:inDP} the collection of all the periods (except $\Cprd_0$) and inter-block distances in the tree\VLongOnly{ (as well as $\tspanRepStar{\Pblock_0}$, if necessary)}, that need to be transmitted to fully describe the pattern.
\VLongOnly{The corresponding code length is
\[\cl(\inDP) = \sum_{v \in \inDP} \cl(v)\;,\]
where the code length of each element can be computed using the formulas presented above.

}To put everything together, the code used to represent a pattern $\patt = (\Ptree, \Cto, \Csc)$ has length
\VLongOnly{\begin{align*}
\cl(\patt) &= \cl((\Ptree, \Cto, \Csc)) \\
&= \cl(\evtseq) + \cl(\lens) + \cl(\Cprd_{0}) + \cl(\inDP) + \cl(\Cto) + \cl(\Csc)\;.
\end{align*}}
\VShortOnly{\[\cl(\patt) = \cl(\evtseq) + \cl(\lens) + \cl(\Cprd_{0}) + \cl(\inDP) + \cl(\Cto) + \cl(\Csc)\;.\]}

\VLongOnly{
\mpara{From simpler patterns to more complex ones.}
Let us have a look at what happens to the encoding of a simple cycle, when using this more complex encoding scheme to represent it.
Consider a simple cycle $\cycle = (\Cev, \Clen, \Cprd, \Cto, \Csc)$. Using the more complex encoding it can be represented as $\patt = (\Ptree, \Cto, \Csc)$, where the cycle is represented using a more general pattern formalism $\Ptree = \BinfoRP{\Clen}{\Cprd}\Bstart{}\activity{\Cev}\Bend{}$.
Both encodings are very similar, with $\lens = \LL{\Clen}$, $\Cprd_{0} = \Cprd$ and $\inDP = \llempty$, $\evtseq = \text{`( \Cev )'} $.
The code word representing the cycle length, $\cl(\Clen)$, depends only on the frequency of occurrence of the event, which is fixed. The corrections accumulated for the first occurrence of the last repetition of the main cycle are equal to the sum of the corrections in $\Csc$,\VLongOnly{ hence $\cume{o_{za}} = \sumel{\Csc}$,} so that the length of the code words representing the cycle period and starting point also remain the same. The corrections are the same and encoded the same way under both encodings. The only difference comes from the different way to encode the event, which is longer under the more complex encoding, to accommodate for the additional symbols which allow to represent (nested) event sequences.\VLongOnly{ That is, for any event $\alpha$, its code length under the more complex pattern encoding $\cl_P(\Cev)$ is larger than its code length under the simpler cycle encoding, $\cl_C(\Cev)$, due to the over-head of having block delimiters.}

Note that the actual value of $\Cto$ does not impact the code length of a pattern.
\VLongOnly{If we consider two cycles 
\[\cycle_1 = (\Cev_1, \Clen_1, \Cprd_1, \Cto_1, \Csc_1)\text{ and }\cycle_2 = (\Cev_2, \Clen_2, \Cprd_2, \Cto_2, \Csc_2)\]
 such that $\Cto_1 \neq \Cto_2$ but all other values are equal, then $\cl(\cycle_1) = \cl(\cycle_2)$. Simply put, translation does not affect the cost of a cycle or pattern.

}On the other hand, the values of the corrections\VLongOnly{, through $\cume{o_{za}}$} impact the length of the code words representing the starting point and the main period.\VLongOnly{ For this reason, given two cycles with the same length and period but with different corrections (i.e.\ such that $\Clen_1 = \Clen_2$ and $\Cprd_1 = \Cprd_2$, but $\Csc_1 \neq \Csc_2$), the code words representing their respective periods and starting points will differ (i.e.\ we will have $\cl(\Clen_1) = \cl(\Clen_2)$ but $\cl(\Cprd_1) \neq \cl(\Cprd_2)$ and $\cl(\Cto_1) \neq \cl(\Cto_2)$).}
}

\VLongOnly{\section{Combining patterns and comparing costs}}
\VShortOnly{\section{Algorithm for Mining Periodic Patterns that Compress}
\label{sec:algo}}
\label{sec:comb}

Recall that for a given input sequence $\seq$, our goal is to find a collection of patterns $\ccycle$ that minimises the cost
\[\cl(\ccycle, \seq) = \sum_{\patt \in \ccycle} \cl(\patt) + \sum_{o \in \residual{\ccycle, \seq}} \cl(o)\;.\]

It is useful to compare the cost of different patterns, or sets of patterns, on a subset of the data, i.e.\ compare $\cl(\ccycle', \seq')$ for different sets of patterns $\ccycle'$ and some subsequence $\seq' \subseteq \seq$.
In particular, we might compare the cost of a pattern $\patt$ to the cost of representing the same occurrences separately. This means comparing
\[\cl(\{\patt\}, \cov{\patt}) = \cl(\patt) \quad\text{and}\quad 
\cl(\emptyset, \cov{\patt}) = \sum_{o \in \cov{\patt}} \cl(o)\;.\]
If $\cl(\{\patt\}, \cov{\patt}) < \cl(\emptyset, \cov{\patt})$, we say that pattern $\patt$ is \emph{cost-effective}.
In addition, we compare patterns in terms of their cost-per-occurrence ratio defined, for a pattern $\patt$, as 
\custmath{\custfrac{\cl(\patt)}{\abs{\cov{\patt}}}}{,}
and say that a pattern is more \emph{efficient} when this ratio is smaller.\VLongOnly{

Furthermore, in order to reduce the number of candidate patterns considered and to retain only the most promising ones, we use a procedure called $\algFilterTopKCov$ that takes as input a collection of patterns $\mathcal{K}$ together with some integer $k$ and returns only those patterns from $\mathcal{K}$ that are among the top-$k$ most efficient ones for some occurrence they cover.}

\medskip

A natural way to build patterns is to start with the simplest patterns, i.e.\ cycles over single events, and combine them together into more complex, possibly multi-level multi-event patterns. 
\VLongOnly{Therefore, we now look at how the cost of patterns relates to the cost of the building blocks they are constructed from.
We start by looking at the cost of covering $k$ occurrences ($k \geq 3$) with a simple cycle as compared to representing them separately. In other words, we look in more details at what it takes for a cycle to be cost-effective.

\mpara{Simple cycles vs.\ residuals.}
Assume we have a candidate cycle $C$ of length $k \geq 3$, covering $k$ occurrences of event $\alpha$, and we want to check whether this cycle is cost-effective, i.e.\ compare the cost of representing this $k$-subsequence with $C$ to the cost of representing it with individual occurrences
\[\cl(\{C\}, \cov{C}) = \cl(C) \quad\text{and}\quad 
\cl(\emptyset, \cov{C}) = \sum_{o \in \cov{C}} \cl(o)\;.\]

The cost of representing the individual occurrences separately is
\[\cl(\emptyset, \cov{C}) = k \cdot (\cl(t) + \cl(\alpha)) = k \big( \log(\tspan{\seq}+1) -\log(\frac{\len{\seq[\alpha]}}{\len{\seq}})\big) \]
and the cost for representing the same occurrences with cycle $C$ is
\[\cl(C) = \cl(\alpha) + \beta + \cl(\Clen) + \cl(\Cprd) + \cl(\Cto) +\cl(\Csc)\;,\]
where $\beta$ denotes the length of the code for one pair of block delimiters.
The cost of corrections in the cycle is
\[\cl(\Csc) = 2(k-1) + \sum_{e \in \Csc} \abs{e}\]
and the code length of the period and starting point of a cycle satisfy, respectively,
\[\cl(\Cprd) < \log(\frac{\tspan{\seq}+1}{k-1}) \quad \text{and} \quad \cl(\Cto) < \cl(t),\]
so that
\[\cl(C) < \cl(\alpha) + \beta + \cl(\Clen) + \log\big(\frac{\tspan{\seq}+1}{k-1}\big) + \cl(t)  + 2k-2 + \sum_{e \in \Csc} \abs{e}\;.\]

If we let
\begin{align*}
W(k) =\,& (k-1) (\cl(t) + \cl(\alpha))- \beta - \cl(\Clen) - \log\big(\frac{\tspan{\seq}+1}{k-1}\big)  - 2k+2\\
=\,& (k-2) \log(\tspan{\seq}+1) + (k-1) \cl(\alpha) - \beta - \log(\len{\seq[\Cev]}) + \log(k-1) - 2k+2\;,
\end{align*}
we have
\[\sum_{e \in \Csc} \abs{e} < W(k) \implies \cl(\{C\}, \cov{C}) < \cl(\emptyset, \cov{C})\;.\]
In other words, if the sum of the absolute shift corrections in a cycle $C$ of length $k$ is less than $W(k)$, then the cost of representing the occurrences with $C$ is smaller than the cost of representing them separately.

Furthermore, we can state the following:
\begin{lemma}
\label{lemma:cycles}
Given a sequence $\seq$, if $C$ is a cycle of length $k$ over event $\alpha$ with corrections $\Csc$ satisfying $\sum_{e \in \Csc} \abs{e} < W(k)$, and if extending $C$ to cover one further occurrence of event $\alpha$ does not increase the sum of the absolute corrections by more than $\log(\tspan{\seq}+1) - 2$, then the cost of representing the $k+1$ occurrences with the extended cycle is smaller than the cost of representing them separately, i.e.\ the extended cycle remains cost-effective.
\end{lemma}

\begin{proof}
Assume we have a cycle $C$ with corrections $\Csc$, satisfying $\sum_{e \in \Csc} \abs{e} < W(k)$. Let $C'$ be the cycle obtained by extending $C$ to cover one further occurrence, i.e.\ $C'$ is a cycle of length $k+1$, and let $\Csc'$ be the associated corrections.
Since
\[W(k+1) - W(k) = \log(\tspan{\seq}+1) + \cl(\alpha) + \log(k/(k-1)) - 2 > \log(\tspan{\seq}+1) - 2\;,\] 
we have
\begin{align*} 
 & \sum_{e \in \Csc'} \abs{e} - \sum_{e \in \Csc} \abs{e} \leq \log(\tspan{\seq}+1) - 2 \\
\implies & \sum_{e \in \Csc'} \abs{e} - \sum_{e \in \Csc} \abs{e} < W(k+1) - W(k) \\
\implies & \sum_{e \in \Csc'} \abs{e}  < W(k+1) - W(k) + \sum_{e \in \Csc} \abs{e} \\
\implies & \sum_{e \in \Csc'} \abs{e}  < W(k+1) \\
\implies & \cl(\{C'\}, \cov{C'}) < \cl(\emptyset, \cov{C'}) \;.
\end{align*} 
\end{proof}

For a simple criterion to decide whether to extend a cycle
we compare the magnitude of the new correction to $\log(\tspan{\seq}+1) - 2$.}

\VLongOnly{\mpara{Vertical combination: Nesting cycles.}
First, let us consider a practical example. Imagine that the following sequence is part of the input:
\begin{align*}
\seqex{2} = \lls & (2,a), (5,a), (7,a), (8,a), (13,a), (15,a), \\
& (20,a), (21,a), (26,a), (29,a), (32,a), (33,a) \lle \;.
\end{align*}

We can represent this sequence with simple cycles, using three patterns over pattern tree $\Ptree_1$ from Fig.~\ref{fig:ex_tree1a} with starting points $2$, $13$, and $26$, respectively. 

Using this notation, the first option is to represent the sequence with the collection 
\begin{align*} 
\ccycle_1 &= \{\patt_{1,1}, \patt_{1,2}, \patt_{1,3} \}\\ &= \{ (\Ptree_1, 2, \LL{1,0,-1}), (\Ptree_1, 13, \LL{0,3,-1}), (\Ptree_1, 26, \LL{1,1,-1})\}\;.
\end{align*} 

Alternatively, we can represent the sequence using four patterns over pattern tree $\Ptree_2$ from Fig.~\ref{fig:ex_tree1a} with starting points $2$, $5$, $7$ and $8$, respectively:
\begin{align*}
\ccycle_2 &= \{\patt_{2,1}, \patt_{2,2}, \patt_{2,3}, \patt_{2,4} \} \\
&= \{(\Ptree_2, 2, \LL{-2,0}), (\Ptree_2, 5, \LL{-3,1}), \\ 
&\phantom{= \{} (\Ptree_2, 7, \LL{0,-1}), (\Ptree_2, 8, \LL{0,-1})\}\;.
\end{align*}

But it can also be represented as a single pattern containing two nested cycles, namely as patterns over pattern trees $\Ptree_3$ or $\Ptree_4$ from Fig.~\ref{fig:ex_tree1a}, respectively, depending whether the inner cycle is $\Ptree_1$ or $\Ptree_2$. So, we can represent the sequence with a single pattern, with either 
\begin{align*}
 \ccycle_3 &= \{\patt_{3,1}\} = \{(\Ptree_3, 2, \LL{1, 0, -1, -2, 0, 3, -1, 0, 1, 1, -1})\},\text{ or }\\
 \ccycle_4 &= \{\patt_{4,1}\} = \{(\Ptree_4, 2, \LL{-2, 0, 1, -3, 1, 0, 0, -1, -1, 0, -1})\}\;.
\end{align*}

Note that with this type of pattern combining two nested cycles over the same event, the list of corrections for the combined pattern is a simple combination of corrections for the basic cycles: 
\[\Csc_{3,1} = \Csc_{1,1} \opconcat \LL{\Csc_{2,1}[1]} \opconcat \Csc_{1,2} \opconcat \LL{\Csc_{2,1}[2]} \opconcat \Csc_{1,3}\]
where $\Csc_{x,y}$ is the list of shift corrections for pattern $\patt_{x,y}$ and $\Csc_{x,y}[i]$ is the correction at position $i$ in that list.

Let us look at the code lengths for these different patterns.
For this example, we have
\[
\begin{array}{r@{}lr@{}lr@{}l}
\tSstart(\seqex{2})&=0,&\tSend(\seqex{2})&=34,&\tspan{\seqex{2}}&=34,\\ 
\multicolumn{2}{c}{\text{ and }}&\len{\seqex[a]{2}}&= 12\;.\\
\end{array}
\]
We list the code lengths for the different elements in Tables~\ref{tab:ex_clC1}--\ref{tab:ex_clC3-4}.
In Fig.~\ref{fig:ex_timelinesP3-P4} we provide a timeline schema of the occurrences of $\patt_{3,1}$ as well as of the occurrences of $(\Ptree_3, 0, \nullC)$ and $(\Ptree_4, 0, \nullC)$, i.e.\ the occurrences of pattern trees $\Ptree_3$ and $\Ptree_4$ with starting point $0$ and no corrections.

\begin{figure}[tbp]
\centering
\setlength{\ndhlen}{1.cm}
\setlength{\ndwlen}{.8cm}
\newcommand{\treeMedH}[5]{
  \node[main node] (b1#1) at (#2) {};
  \node[info node] (l1rp#1) [above of=b1#1, xshift=-.5cm, node distance=.3\ndhlen] {$\Clen=#5, \Cprd=7$};
  \node[info node] (l1t#1) [above of=l1rp#1, node distance=1.2em] {$\Cto=#4$};
  \node[leaf node] (b2#1) [below of=b1#1] {};
  \node[pterm node, xshift=-.2cm] (p2#1) [below of=b2#1] {$#3$};
  \path
    (b1#1) edge (b2#1);
}
\newcommand{\treeMedHP}[5]{
  \node[main node] (b1#1) at (#2) {};
  \node[info node] (l1rp#1) [above of=b1#1, xshift=-.5cm, node distance=.3\ndhlen] {$\Clen=#5, \Cprd=7$};
  \node[info node] (l1t#1) [above of=l1rp#1, node distance=1.2em] {$\Cto=#4$};
  \node[main node] (b2#1) [below of=b1#1] {};

  \node[info node] (l2r#1) [right of=b2#1, node distance=.8\ndwlen] {$\Clen=3$};
  \node[info node] (l2p#1) [below of=l2r#1, node distance=1.2em] {$\Cprd=2$};

  \node[leaf node] (b3#1) [below of=b2#1] {};
  \node[pterm node, xshift=-.2cm] (p2#1) [below of=b3#1] {$#3$};
  \path
    (b1#1) edge (b2#1)
    (b2#1) edge (b3#1);
}

\newcommand{\treeMedCH}[3]{
  \node[main node] (b2#1) at (#2) {};
  \node[info node] (l1rp#1) [above of=b2#1, xshift=-.5cm, node distance=.3\ndhlen] {$\Clen=5, \Cprd=7$};
  \node[info node] (l1t#1) [above of=l1rp#1, node distance=1.2em] {$\Cto=#3$};

  \node[leaf node] (b22#1) [below of=b2#1] {};
  \node[leaf node] (b21#1) [left of=b22#1, node distance=1.5\ndwlen] {};
  \node[main node] (b23#1) [right of=b22#1, node distance=1.5\ndwlen] {};
  \node[leaf node] (b231#1) [below of=b23#1] {};

  \node[info node] (l2r#1) [right of=b23#1, node distance=.8\ndwlen] {$\Clen=3$};
  \node[info node] (l2p#1) [below of=l2r#1, node distance=1.2em] {$\Cprd=2$};

  \node[pterm node, xshift=-.2cm] (p22#1) [below of=b22#1] {$a$};
  \node[pterm node, xshift=-.2cm] (p23#1) [below of=b21#1] {$b$};
  \node[pterm node, xshift=-.2cm] (p211#1) [below of=b231#1] {$b$};
  \path
    (b2#1) edge (b21#1)
    (b2#1) edge (b22#1)
    (b2#1) edge (b23#1)
    (b23#1) edge (b231#1);

\path[dotted, thin, ->, bend right=20, color=darkgray]
    (b21#1) edge node[below] {$d=2$} (b22#1)
    (b22#1) edge node[below] {$d=1$} (b23#1);
}

\newcommand{\treeMedV}[3]{
  \node[ghost node] (b0#1) at (#2) {};
  \node[main node] (b2#1) [below of=b0#1] {};
  \node[ghost node] (b2x#1) [below of=b2#1] {};
  \node[info node] (l2rp#1) [above of=b2#1, xshift=-.5cm, node distance=.3\ndhlen] {$\Clen=3, \Cprd=2$};
  \node[info node] (l2t#1) [above of=l2rp#1, node distance=1.2em] {$\Cto=#3$};
  \node[leaf node] (b21#1) [left of=b2x#1, node distance=\ndwlen] {};
  \node[leaf node] (b22#1) [right of=b2x#1, node distance=\ndwlen] {};
  \node[pterm node, xshift=-.2cm] (p21#1) [below of=b21#1] {$c$};
  \node[pterm node, xshift=-.2cm] (p22#1) [below of=b22#1] {$e$};
  \path
    (b2#1) edge (b21#1)
    (b2#1) edge (b22#1);
\path[dotted, thin, ->, bend right=20, color=darkgray]
    (b21#1) edge node[below] {$d=1$} (b22#1);
}
\newcommand{\treeMedCV}[2]{
  \node[main node] (b0#1) at (#2) {};
  \node[info node] (l0rp#1) [above of=b0#1, xshift=-.5cm, node distance=.3\ndhlen] {$\Clen=12, \Cprd=7$};
  \node[info node] (l0t#1) [above of=l0rp#1, node distance=1.2em] {$\Cto=3$};

  \node[main node] (b2#1) [below of=b0#1] {};
  \node[ghost node] (b2x#1) [below of=b2#1] {};
  \node[info node] (l2p#1) [right of=b2#1, node distance=\ndwlen] {$\Cprd=2$};
  \node[info node] (l2r#1) [above of=l2p#1, node distance=1.2em] {$\Clen=3$};

  \node[leaf node] (b21#1) [left of=b2x#1, node distance=\ndwlen] {};
  \node[leaf node] (b22#1) [right of=b2x#1, node distance=\ndwlen] {};
  \node[pterm node, xshift=-.2cm] (p21#1) [below of=b21#1] {$c$};
  \node[pterm node, xshift=-.2cm] (p22#1) [below of=b22#1] {$e$};
  \path
    (b0#1) edge (b2#1)
    (b2#1) edge (b21#1)
    (b2#1) edge (b22#1);
 \path[dotted, thin, ->, bend right=20, color=darkgray]
    (b21#1) edge node[below] {$d=1$} (b22#1);

}

\setlength{\ndhlen}{1.cm}
\setlength{\ndwlen}{.8cm}

\begin{tikzpicture}[-,auto,node distance=\ndhlen, thick]
\node[anchor=west] at (-\ndwlen,1.5\ndhlen) {a) $\algGrowH$:};
\treeMedH{h}{0,0}{b}{2}{6}
\treeMedH{f}{2.6\ndwlen,0}{a}{4}{5}
\treeMedHP{g}{5.2\ndwlen,0}{b}{5}{5}
\node at (7.\ndwlen,-.5\ndhlen) {$\longrightarrow$};
\treeMedCH{y}{9.5\ndwlen,0}{2}

\node[anchor=west] at (-\ndwlen,-3\ndhlen) {b) $\algGrowV$:};
\treeMedV{a}{0,-3.5\ndhlen}{3}
\treeMedV{b}{2.6\ndwlen,-3.5\ndhlen}{12}
\node at (4.3\ndwlen,-5\ndhlen) {$\dots$};
\treeMedV{d}{6\ndwlen,-3.5\ndhlen}{102}
\node at (7.7\ndwlen,-5\ndhlen) {$\longrightarrow$};
\treeMedCV{x}{9.5\ndwlen,-3.5\ndhlen}
\end{tikzpicture}
\caption{Examples of growing patterns through combinations.}
\label{fig:ex_tree_grow}
\end{figure}
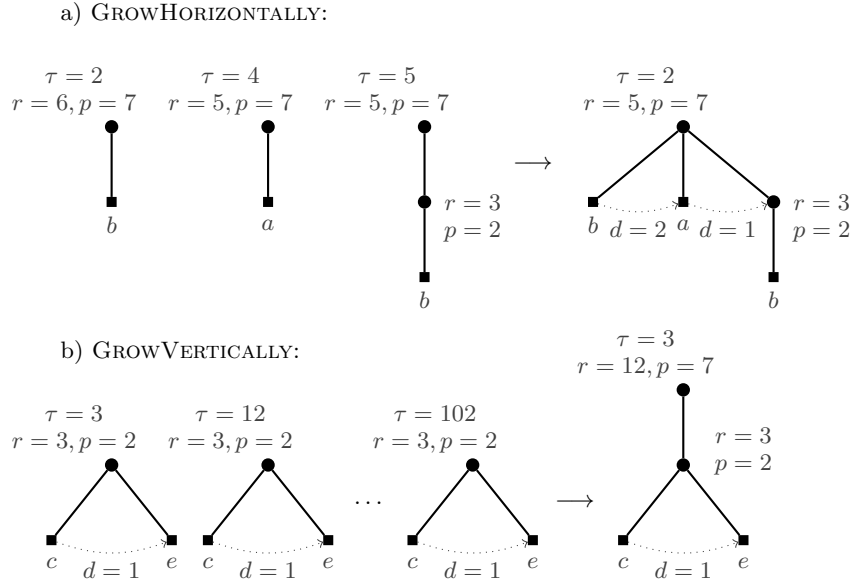

\medskip

Now, let us turn to the general case.}
Assume that we have a pattern tree $\Ptree_I$ which occurs multiple times in the event sequence. In particular, assume that it occurs at starting points $\Cto_{1}$, $\Cto_{2}$, $\dots$, $\Cto_{\PlenJ}$ \VLongOnly{(where the starting points are ordered) }and that this sequence of starting points itself can be represented as a cycle of length $\PlenJ$ and period $p_J$. \VLongOnly{In other words, 
if we denote  as $\alpha$ the left-most event of $\Ptree_I$, i.e.\ the event associated to the starting point of $\Ptree_I$, the sequence consisting of the starting points of the different occurrences of $\Ptree_I$ can be represented by a pattern $(\Ptree_J, \Cto_{1}, \Csc_J)$ where $\Ptree_J = \BinfoRP{\PlenJ}{p_J}\Bstart{}\activity{\alpha}\Bend{}$ is a cycle of length $\PlenJ$ and period $p_J$ over event $\alpha$, with shift corrections 
\[\Csc_J = \LL{(\Cto_{i}-\Cto_{i-1})-p_J \text{ for } i \in [2,\PlenJ]}\;.\]

}In such a case, the occurrences of $\Ptree_I$ might be combined together and represented as a nested pattern tree\VLongOnly{ $\Ptree_N = \BinfoRP{\PlenJ}{p_J}\Bstart{}\activity{\Ptree_I}\Bend{}$}. \VLongOnly{We refer to such a combination as \emph{vertical combination}, since it produces patterns of greater depth than the original ones. }$\algGrowV$ is the procedure which takes as input a collection $\ccycle_I$ of patterns over a tree $\Ptree_I$\VLongOnly{, i.e.\ $\ccycle_I = \{(\Ptree_I, \Cto_{1}, \Csc_{I,1}), \dots (\Ptree_I, \Cto_{\PlenJ}, \Csc_{I,\PlenJ})\}$} and returns the nested pattern\VLongOnly{, covering the same timestamp--event pairs,} obtained by combining them together as depicted in Fig.~\ref{fig:ex_tree_grow}(b).\VLongOnly{

This situation is illustrated in Fig.~\ref{fig:ex_combineV}.

\begin{lemma}
\label{lemma:vertical}
Let $\ccycle_I = \{(\Ptree_I, \Cto_{1}, \Csc_{I,1}), \dots (\Ptree_I, \Cto_{\PlenJ}, \Csc_{I,\PlenJ})\}$ be a collection of patterns consisting of $\PlenJ$ occurrences of the same pattern tree $\Ptree_I$ and $\patt_N = \algGrowV(\ccycle_I)$ be the nested pattern obtained by combining the patterns in $\ccycle_I$.
If the cycle $\patt_J$ over the starting points of the patterns in $\ccycle_I$ satisfies 
\[\cl(\patt_J) < (\PlenJ-1) \cdot \cl((\Ptree_I, \Cto_{1}, \emptyC)) \;,\] then 
\[\cl(\{\patt_N\}, \cov{\ccycle_I}) < \cl(\ccycle_I, \cov{\ccycle_I})\;.\]
\end{lemma}

\begin{proof}
The code length of the event sequence in $\Ptree_N$, i.e.\ $\evtseq_N = \text{`(}\Ptree_I\text{)'}$ equals the code length to encode the event sequence in $\Ptree_I$ plus the code length for one pair of block delimiters and satisfies
\[\cl(\evtseq_N) < \cl(\evtseq_{I,\PlenJ}) + \cl(\evtseq_J).\]
Once nested, the time spans in $T$ can only become more constrained, so that $\cl(\inDP_N) \leq \cl(\inDP_{I,\PlenJ})$. 
The shift corrections for the nested pattern can be written as
\[\Csc_{N} = \Csc_{I,1} \opconcat \LL{\Csc_{J}[1]} \opconcat \Csc_{I,2} \opconcat \LL{\Csc_{J}[2]}  \dots \LL{\Csc_{J}[\PlenJ-1]} \opconcat \Csc_{I,\PlenJ},\]
so that
\[\cl(\Csc_N) = \cl(\Csc_J) + \sum_{i \in [1,\PlenJ]} \cl(\Csc_{I,i})\;.\]
For the remaining elements, we have
\begin{align*}
\cl(\lens_N) &= \cl(\lens_{I,\PlenJ}) + \cl(\lens_J) \\
\cl(\Cprd_{0N}) &= \cl(\Cprd_{0J}) \\
\cl(\Cto_N) &= \cl(\Cto_J)
\end{align*}

Hence, the following holds for the code length of the nested pattern $\patt_N$ when compared to the code length for the inner patterns $\patt_{I,i}$ and the outer pattern $\patt_J$:
\[\cl(\patt_N) < \cl(\patt_J) + \cl(\evtseq_{I,\PlenJ}) + \cl(\lens_{I,\PlenJ}) + \cl(\inDP_{I,\PlenJ}) + \sum_{i \in [1,\PlenJ]} \cl(\Csc_{I,i})\;.\]

We can then compare the code length of the outer pattern to the code length of the structure of all but one of the inner patterns $\patt_J$, that is
\begin{align*}
& \cl(\patt_J) < (\PlenJ-1) \cdot \cl((\Ptree_I, \Cto_1, \emptyC)) \\
\implies & \cl(\patt_J) + \cl(\evtseq_{I,\PlenJ}) + \cl(\lens_{I,\PlenJ}) + \cl(\inDP_{I,\PlenJ}) + \sum_{i \in [1,\PlenJ]} \cl(\Csc_{I,i}) \\ 
& <  (\PlenJ-1) \cdot \cl((\Ptree_I, \Cto_1, \emptyC)) + \cl(\evtseq_{I,\PlenJ}) + \cl(\lens_{I,\PlenJ}) + \cl(\inDP_{I,\PlenJ}) + \sum_{i \in [1,\PlenJ]} \cl(\Csc_{I,i}) \\
\implies & \cl(\patt_N) = \cl(\{\patt_N\}, \cov{\ccycle_I}) < \sum_{i \in [1,\PlenJ]} (\Ptree_{I,i}, \Cto_{I,i}, \Csc_{I,i}) = \cl(\ccycle_I, \cov{\ccycle_I})\;.
\end{align*}
\end{proof}

\mpara{Horizontal combination: Concatenating cycles.}
Again, let us first consider a practical example.
Imagine that the following sequence is part of the input:
\begin{align*}
\seqex{3} = \lls & (2,b), (5,a), (7,c), (13,b), (18,a), (21,c), \\
& (26,b), (30,a), (31,c) \lle \;.
\end{align*}

We can represent this sequence with single cycles of length $3$ and period $13$, over events $b$, $a$, and $c$ and with starting points $2$, $5$, and $7$, respectively. 
The cycle over $a$ corresponds to pattern tree $\Ptree_2$ from Fig.~\ref{fig:ex_tree1a}, the other two cycles correspond to similar pattern trees but over event $b$ and $c$, so we denote them respectively as $\Ptree_{2b}$ and $\Ptree_{2c}$. 
This corresponds to the following collection:
\begin{align*} 
\ccycle_5 &= \{\patt_{5,1}, \patt_{5,2}, \patt_{5,3}\} \\ 
&= \{ (\Ptree_{2b}, 2, \LL{-2,0}), (\Ptree_2, 5, \LL{0,-1}), (\Ptree_{2c}, 7, \LL{1,-3})\}\;.
\end{align*}

We can also use a more complex pattern tree, concatenating the three events. This corresponds to using pattern tree $\Ptree_5$ from Fig.~\ref{fig:ex_tree1b}: 
\begin{align*} 
\ccycle_6 &= \{\patt_{6,1}\} \\ &= \{(\Ptree_5, 2, \LL{0, 1, -2, 2, 2, 0, 1, 0})\}\;.
\end{align*}

Let us look at the code lengths for these different patterns.
For this example, we have 
\[
\begin{array}{r@{}lr@{}lr@{}l}
\tSstart(\seqex{3})&=0,&\tSend(\seqex{3})&=34,&\tspan{\seqex{3}}&=34,\\ 
\text{ and } & \multicolumn{4}{r}{\len{\seqex[a]{3}} =\len{\seqex[b]{3}}=\len{\seqex[c]{3}}}&= 3\;.\\
\end{array}
\]
We list the code lengths for the different elements in Tables~\ref{tab:ex_clC5}--\ref{tab:ex_clC6}.
In Fig.~\ref{fig:ex_timelinesP5} we provide a timeline schema of the occurrences of $\patt_{6,1}$ as well as of the occurrences of $(\Ptree_5, 0, \nullC)$.

\medskip 

Given}\VShortOnly{

On the other hand, given}
a collection of patterns that occur close to one another and share similar periods, we might want to combine them together into a concatenated pattern by merging the roots of their respective trees. \VLongOnly{We refer to such a combination as \emph{horizontal combination}, since it produces patterns of greater width than the original ones.

To understand what this means in terms of cost, we focus on the basic case where we have two patterns $\patt_I$ and $\patt_J$, such that $\Ptree_I = \BinfoRP{\Plen}{p_I}\Bstart{}\activity{T}\Bend{}$ and $\Ptree_J = \BinfoRP{\Plen}{p_J}\Bstart{}\activity{T'}\Bend{}$, both patterns have top-level blocks of the same length $\Plen$, and with starting points $\Cto_I \leq \Cto_J$.
We compare the cost of these two patterns to the code length for the pattern that concatenates them, that is, pattern $\patt_N$ with  $\Ptree_N = \BinfoRP{\Plen}{p_N}\Bstart{}\activity{T} \BinfoD{d_N} \activity{T'}\Bend{}$ covering the same event occurrences in the original sequence. $\PnbOI$ and $\PnbOJ$ denote the number of occurrence in one repetition of the top-level block of patterns $\patt_I$ and $\patt_J$ respectively, that is $\abs{\occsStar{T}} = \PnbOI$ and $\abs{\occsStar{T'}} = \PnbOJ$. 
This situation is illustrated in Fig.~\ref{fig:ex_combineH}.

Since the shift corrections are applied relatively within a block, concatenating $T$ and $T'$ only impacts the first event occurrence of each repetition of the top-level block in either pattern, i.e.\ the left-most leaf in $T$ and in $T'$.
We must look at the timestamps of occurrences of the first event in $T$ and in $T'$, let's denote the timestamp of the $i^{th}$ occurrence of these events as $t(o_{i,1})$ and $t(o'_{i,1})$ respectively.

Looking at the position at which these occurrences are produced by the different patterns, we have 
\begin{align*}
\Csc_I(o_{i,1}) &= \Csc_I[(i-1)\PnbOI] &
\Csc_J(o'_{i,1}) &= \Csc_J[(i-1)\PnbOJ] \\
\Csc_N(o_{i,1}) &= \Csc_N[(i-1)(\PnbOI+\PnbOJ)] &
\Csc_N(o'_{i,1}) &= \Csc_N[i\PnbOI+(i-1)\PnbOJ]\;. 
\end{align*}

Per $(\Ptree_I, \Cto_I, \Csc_I)$ we have
\begin{align}
t(o_{1,1}) &= \Cto_I\;, \label{a1x} \\
t(o_{2,1}) &= \Cto_I + p_I + \Csc_I(o_{2,1})\;, \label{a2x}\\
t(o_{3,1}) &= \Cto_I + 2 p_I + \Csc_I(o_{3,1}) + \Csc_I(o_{2,1})\;, \label{a3x}
\end{align}
and per $(\Ptree_N, \Cto_N, \Csc_N)$
\begin{align}
t(o_{1,1}) &= \Cto_N\;, \label{a1z}\\
t(o_{2,1}) &= \Cto_N + p_N + \Csc_N(o_{2,1})\;, \label{a2z}\\
t(o_{3,1}) &= \Cto_N + 2 p_N + \Csc_N(o_{3,1}) + \Csc_N(o_{2,1})\;. \label{a3z}
\end{align}
Hence, from eq.~\ref{a1x} and eq.~\ref{a1z} we get
\[ \Cto_N = \Cto_I. \]
And generalising from eq.~\ref{a2x} and eq.~\ref{a2z} we get
\[ \Csc_N(o_{i,1}) = (p_I - p_N) + \Csc_I(o_{i,1}).\]
And therefore, we let $p_N = p_I$ so that $\Csc_N(o_{i,j}) = \Csc_I(o_{i,j})$ for all event occurrences of $\patt_I$.

Furthermore, we have per $(\Ptree_J, \Cto_J, \Csc_J)$ 
\begin{align}
t(o'_{1,1}) &= \Cto_J\;, \label{b1y}\\
t(o'_{2,1}) &= \Cto_J + p_J + \Csc_J(o'_{i,2})\;, \label{b2y}\\
t(o'_{3,1}) &= \Cto_J + 2 p_J + \Csc_J(o'_{i,3}) + \Csc_J(o'_{i,2})\;, \label{b3y}
\end{align}
and per $(\Ptree_N, \Cto_N, \Csc_N)$
\begin{align}
t(o'_{1,1}) &= \Cto_N + d_N + \Csc_N(o'_{1,1})\;, \label{b1z} \\
t(o'_{2,1}) &= t(o_{2,1}) + d_N + \Csc_N(o'_{2,1})\;, \label{b2z}\\
t(o'_{3,1}) &= t(o_{3,1}) + d_N + \Csc_N(o'_{3,1})\;. \label{b3z}
\end{align}

Hence, from eq.~\ref{b1y} and eq.~\ref{b1z} we get
\[ d_N = (\Cto_J - \Cto_I) - \Csc_N(o'_{1,1})\;.\]
and therefore we let $d_N = (\Cto_J - \Cto_I)$.
From eq.~\ref{b2y} and eq.~\ref{b2z} we get
\begin{align*}
\Cto_J +& p_J + \Csc_J(o'_{2,1}) \\ 
&= \Cto_N + p_N + \Csc_N(o_{2,1}) + d_N + \Csc_N(o'_{2,1})\;,\\
&= \Cto_N + p_N + \Csc_N(o_{2,1}) + (\Cto_J - \Cto_I) - \Csc_N(o'_{1,1}) + \Csc_N(o'_{2,1})\;,
\end{align*}
and hence
\[(p_J - p_N) + \Csc_J(o'_{2,1}) =  \Csc_N(o'_{2,1}) -  \Csc_N(o'_{1,1}) +  \Csc_N(o_{2,1}) \;.\]
More generally, we have
\[(p_J - p_N) + \Csc_J(o'_{i,1}) =  \Csc_N(o'_{i,1}) -  \Csc_N(o'_{(i-1),1}) +  \Csc_N(o_{i,1}) \;,\]
and using $p_N = p_I$ and $\Csc_N(o_{i,1}) = \Csc_I(o_{i,1})$:
\[ \Csc_N(o'_{i,1})  = (p_J - p_I) + \Csc_I(o_{i,1}) - \Csc_J(o'_{i,1})  + \Csc_N(o'_{(i-1),1})\;.\]

In the best case, the patterns are well aligned, in the sense that $\Csc_I(o_{i,1}) = \Csc_J(o'_{i,1})$,
so then, summing up the shift corrections above, which are the only ones that differ between the old patterns and the new one, we get
\[ \sum_{i \in [1,\Plen-1]} \abs{\Csc_N(o'_{i,1})} = \frac{\Plen(\Plen-1)}{2}\abs{p_J - p_I}\;.\]
We use this as a filter for patterns to concatenate requiring that 
\[\sum_{i \in [1,\Plen-1]} \abs{\Csc_N(o'_{i,1})} \leq \sum_{i \in [1,\Plen-1]} \abs{\Csc_J(o'_{i,1})}\;,\]
i.e.\
\[\abs{p_J - p_I} \leq \frac{2}{\Plen(\Plen-1)} \sum_{i \in [1,\Plen-1]} \abs{\Csc_J(o'_{i,1})} \;.\]
This can be interpreted as requiring that the difference in period between the two concatenated patterns does not produce shift corrections larger than in the original patterns.

} $\algGrowH$ is the procedure which takes as input a collection of patterns and returns the pattern obtained by concatenating them together in order of increasing starting points as depicted in Fig.~\ref{fig:ex_tree_grow}(a).\VLongOnly{ More specifically, let the input collection be $\{ \patt_i \}$, where each pattern is a cycle of length $\Clen_i$ and period $\Cprd_i$ over a pattern tree $\Ptree_i$ (possibly a single event) with starting point $\Cto_i$, and assume that the patterns in the collection are indexed in order of increasing starting points, i.e.\ in the order in which they occur in the data. The resulting pattern tree $\Ptree_N$ is a cycle of length $\Clen_N = \min(\Clen_i)$ and period $\Cprd_N=\Cprd_1$ over the concatenation of $\Ptree_1, \Ptree_2, \dots$, where the distance between $\Ptree_{i-1}$ and $\Ptree_{i}$ is set to $d_{i}= \Cto_i - \Cto_{i-1}$, and with $\Cto_N = \Cto_1$.
}

\VShortOnly{

\medskip}

\VLongOnly{\section{Algorithm for Mining Periodic Patterns that Compress}
\label{sec:algo}}

\VLongOnly{We are now ready to present our main algorithm for mining a collection of periodic patterns that compresses the input sequence.}
As outlined in Algorithm~\ref{alg:mine_patterns}, our proposed algorithm consists of three stages: \textit{(i)} extracting cycles (line~\ref{alg:main-cycles}), \textit{(ii)} building tree patterns from cycles (lines~\ref{alg:main-comb-start}--\ref{alg:main-comb-end}) and \textit{(iii)} selecting the final pattern collection (line~\ref{alg:main-select}).
We now present each stage in turn\VShortOnly{ at a high-level}.

\begin{algorithm}[tb]
\caption{Mining periodic patterns that compress.}
\label{alg:mine_patterns}
\begin{algorithmic}[1]
\Require A multi-event sequence $\seq$, a number $k$ of top candidates to keep
\Ensure A collection of patterns $\mathcal{P}$
\State $\mathcal{I} \gets \algCycles(\seq, k)$ \label{alg:main-cycles}
\State $\mathcal{C} \gets \emptyset; \mathcal{V} \gets \mathcal{I}; \mathcal{H} \gets \mathcal{I}$ \label{alg:main-comb-start}
\While{$\mathcal{H} \neq \emptyset$ \textbf{or} $\mathcal{V} \neq \emptyset$} \label{alg:main-while}
\State $\mathcal{V}' \gets \algCombineV(\mathcal{H}, \mathcal{P}, \seq, k)$ \label{alg:main-combV}
\State $\mathcal{H}' \gets \algCombineH(\mathcal{V}, \mathcal{P}, \seq, k)$ \label{alg:main-combH}
\State $\mathcal{C} \gets \mathcal{C} \cup \mathcal{H} \cup \mathcal{V}; 
\mathcal{V} \gets \mathcal{V}'; \mathcal{H} \gets \mathcal{H}'$ \label{alg:main-comb-end}
\EndWhile
\State $\mathcal{P} \gets \algFilterFinal(\mathcal{C}, \seq)$ \label{alg:main-select}
\State \textbf{return} $\mathcal{P}$
\end{algorithmic}
\end{algorithm}

\mpara{Extracting cycles.}
\VLongOnly{The first stage of the algorithm consists in extracting cycles (line~\ref{alg:main-cycles}).
The algorithm used for the initial mining of cycles is given as Algorithm~\ref{alg:cycles}.}
Considering each event in turn, we use two different routines to mine cycles from the sequence of timestamps obtained by restricting the input sequence to the event of interest, combine and filter their outputs to generate the set \VLongOnly{$\mathcal{I}$ }of initial candidate patterns.\VLongOnly{

} The first routine, $\algCyclesDP{}$\VLongOnly{ (line~\ref{alg:cycles-dp})}, uses dynamic programming. Indeed, if we allow neither gaps in the cycles nor overlaps between them, finding the best set of cycles for a given sequence corresponds to finding an optimal segmentation of the sequence, and since our cost is additive over individual cycles, we can use dynamic programming to solve it optimally~\cite{bellman:61:on}.\VLongOnly{

} The second routine, $\algCyclesFold{}$\VLongOnly{ (line~\ref{alg:cycles-fold})}, extracts cycles using a heuristic which allows for gaps and overlappings. It collects triples $(t_0, t_1, t_2)$ such that $\abs{\abs{t_2 - t_1} - \abs{t_1 - t_0}} \leq \ell$, where $\ell$ is set so that the triple can be beneficial when used to construct longer cycles. 
Triples are then chained into longer cycles.\VLongOnly{ A triple $(t_{-1}, t_0, t_{+1})$, can be seen as an elementary cycle with a single shift correction $e =\abs{(t_0-t_{-1}) - (t_{+1}-t_{0})}$. Since we are looking for triples that could produce cost-effective cycles, we only keep triples for which $e < \log(\tspan{\seq}+1) - 2$, following Lemma~\ref{lemma:cycles}. Triples $(t_{-1}, t_0, t_{+1})$ and $(t'_{-1}, t'_0, t'_{+1})$ are chained together if $t_{0} = t'_{-1}$ and $t_{+1} = t'_{0}$, producing $(t_{-1}, t_0, t_{+1}, t'_{+1})$, and so on.

} Finally, the set $\mathcal{C}$ of cost-effective cycles obtained by merging the output of the two routines is filtered \VLongOnly{with $\algFilterTopKCov$, }to keep only the $k$ most efficient patterns for each occurrence\VLongOnly{ (line~\ref{alg:cycles-filter})} for a user-specified $k$, and returned.

\VLongOnly{\begin{algorithm}[tb]
\caption{$\algCycles$: Mines simple cycles from the data sequence.}
\label{alg:cycles}
\begin{algorithmic}[1]
\Require A sequence $\seq$
\Ensure A collection of cycles $\mathcal{C}$
\State $\mathcal{C} \gets \emptyset$
\State $l_{\max} \gets \log(\tspan{\seq}+1) - 2$ \label{alg:cycles-crit}
\For{\textbf{each} event $\alpha \in \omega$} \label{alg:cycles-event}
\State $\mathcal{C} \gets \mathcal{C} \cup \algCyclesDP(\seq[\alpha])$ \label{alg:cycles-dp}
\State $\mathcal{C} \gets \mathcal{C} \cup \algCyclesFold(\seq[\alpha], l_{\max})$ \label{alg:cycles-fold}
\EndFor
\State $\algFilterTopKCov(\mathcal{C}, \seq, k)$ \label{alg:cycles-filter}
\State \textbf{return} $\mathcal{C}$
\end{algorithmic}
\end{algorithm}}

\VLongOnly{\medskip}

\mpara{Building tree patterns from cycles.}
The second stage of the algorithm builds tree patterns, starting from the cycles produced in the previous stage.
That is, while there are new candidate patterns, the algorithm performs combination rounds, trying to generate more complex patterns through vertical and horizontal combinations.
If desired, this stage can be skipped, thereby restricting the pattern language to simple cycles. 

In a round of vertical combinations performed by $\algCombineV$ (line~\ref{alg:main-combV}), each distinct pattern tree represented among the new candidates in $\mathcal{H}$ is considered in turn. Patterns over that tree are collected and $\algCyclesFold$ is used to mine cycles from the corresponding sequence of starting points.\VLongOnly{ This time, the threshold used to mine the cycles is derived from the cost of the considered pattern tree, in accordance with Lemma~\ref{lemma:vertical}.} For each obtained cycle, a nested pattern is produced by combining the corresponding candidates using $\algGrowV$ (see Fig.~\ref{fig:ex_tree_grow}(b)). \VLongOnly{The set of candidates produced through these vertical combinations is filtered, and returned as $\mathcal{V}'$. 
The procedure $\algCombineV$ for generating candidate patterns by means of vertical combinations is shown in Algorithms~\ref{alg:combineV}.

\begin{algorithm}[tb]
\caption{$\algCombineV$: Combine patterns vertically.}
\label{alg:combineV}
\begin{algorithmic}[1]
\Require A collection of new candidate patterns $\mathcal{H}$, and other candidate patterns $\mathcal{C}$, a sequence $\seq$, a number $k$ of top candidates to keep
\Ensure A collection of patterns resulting from vertical combinations $\mathcal{V}'$
\State $\mathcal{V}' \gets \emptyset$
\For{\textbf{each} distinct $\Ptree_c \in \mathcal{H}$}
\State $\mathcal{C} \gets \{(\Ptree_x, \Cto_x, \Csc_x) \in \mathcal{H} \cup \mathcal{C}, \text{ such that } \Ptree_x=\Ptree_c\}$
\State $l_{\max} \gets \cl((\Ptree_1, \Cto_1, \emptyC))$
\For{\textbf{each} cycle $(r, p, O) \in \algCyclesFold(\{\Cto_x \in \mathcal{C}\}, l_{\max})$}
\State $\mathcal{K} \gets \{(\Ptree_y, \Cto_y, \Csc_y) \in \mathcal{C}, \text{ such that } \Cto_y \in O\}$
\State $K \gets \algGrowV(\mathcal{K})$
\If{$\cl(\{K\}, \cov{\mathcal{K}}) < \cl(\mathcal{K}, \cov{\mathcal{K}})$}
\State $\mathcal{V}' \gets \mathcal{V}' \cup \{ K \}$
\EndIf
\EndFor
\EndFor
\State $\mathcal{V}' \gets \algFilterTopKCov(\mathcal{V}', \seq, k)$
\State \textbf{return} $\mathcal{V}'$
\end{algorithmic}
\end{algorithm}
}

In a round of horizontal combinations performed by $\algCombineH$ (line~\ref{alg:main-combH}), \VLongOnly{pairs of candidates such that \textit{(i)} at least one of the two patterns was produced in the previous round, and \textit{(ii)} their starting points are closer than the period of the earliest occurring of the two patterns are considered for concatenation.
A}\VShortOnly{a} graph $G$ is constructed, with vertices representing candidate patterns and with edges connecting pairs of candidates $\mathcal{K} = \{\patt_I, \patt_J\}$ for which the concatenated pattern $\patt_N = \algGrowH(\mathcal{K})$ 
satisfies $\cl(\{\patt_N\}, \cov{\mathcal{K}}) < \cl(\mathcal{K}, \cov{\mathcal{K}})$. A new pattern is then produced for each clique of $G$, by applying $\algGrowH$ to the corresponding set of candidate patterns.
\VLongOnly{ The set $\mathcal{H}'$ of new patterns is then filtered and returned.
The procedure $\algCombineH$ for generating candidate patterns by means of horizontal combinations is shown in Algorithms~\ref{alg:combineH}.

To limit the number of concatenations generated and evaluated when testing pairs of patterns, we require that the periods of two patterns be similar enough not to produce shift corrections larger than in the patterns of the pair, as discussed in Section~\ref{sec:comb}.

Note that if we obtain, as a result from a horizontal combination, a pattern a the following shape
\[\BinfoRP{r_0}{p_0} \Bstart{}\BinfoRP{r_1}{p_1} \Bstart{}\activity{T_a}\Bend{} \BinfoD{d} \BinfoRP{r_1}{p_1} \Bstart{}T_b\Bend{}\Bend{}\]
 we will factorise it into
\[\BinfoRP{r_0}{p_0} \Bstart{}\BinfoRP{r_1}{p_1} \Bstart{} T_a \BinfoD{d} T_b\Bend{}\Bend{}\;,\]
if it results in shorter code length, as is often the case.

\begin{algorithm}[tb]
\caption{$\algCombineH$: Combine patterns horizontally.}
\label{alg:combineH}
\begin{algorithmic}[1]
\Require A collection of new candidate patterns $\mathcal{V}$, and other candidate patterns $\mathcal{C}$, a sequence $\seq$, a number $k$ of top candidates to keep
\Ensure A collection of patterns resulting from horizontal combinations $\mathcal{H}'$
\State $\mathcal{H}' \gets \emptyset; G \gets \emptyset$
\State $\mathcal{C} \gets$ pattern pairs $(\patt_a, \patt_b) \in (\mathcal{V} \cup \mathcal{C})^2$, such that 
$(\patt_a \in \mathcal{V}$ or $\patt_b \in \mathcal{V})$ and  $\Cto_b \leq \Cto_a + \Cprd_{0a}$ 
\For{\textbf{each} pair of patterns $\mathcal{K} = (\patt_a, \patt_b) \in \mathcal{C}$}
\State $K \gets \algGrowH(\mathcal{K})$
\If{$\cl(\{K\}, \cov{\mathcal{K}}) < \cl(\mathcal{K}, \cov{\mathcal{K}})$}
\State $\mathcal{H}' \gets \mathcal{H}' \cup \{ K \}$
\State $G \gets G \cup \{(a,b)\}$
\EndIf
\EndFor
\State $\mathcal{H}' \gets \mathcal{H}' \cup \{ \algGrowH(\mathcal{K})$ for each clique $\mathcal{K}$ in the graph $G \}$
\State $\mathcal{H}' \gets \algFilterTopKCov(\mathcal{H}', \seq, k)$
\State \textbf{return} $\mathcal{H}'$
\end{algorithmic}
\end{algorithm}}

\mpara{Selecting the final pattern collection.}
Selecting the final set of patterns to output among the candidates in $\mathcal{C}$ is very similar to solving a weighted set cover problem.
\VLongOnly{Each candidate pattern can be seen as a set containing the occurrences it covers and associated to a weight representing its code length. A singleton set is associated to each occurrence whose weight is the cost of encoding that occurrence as a residual.
}
Therefore, the selection is done using a simple variant of the greedy algorithm for this problem, denoted as $\algFilterFinal$ (line~\ref{alg:main-select})\VShortOnly{.}\VLongOnly{, that works as follows.
Initially, the set $\mathcal{P}$ of selected patterns is empty.
Let $\mathcal{O}$ be the set of event occurrences covered so far, also initially empty.
In each round, the pattern $P$ with smallest value of $\cl(P)/\abs{\occs{P} \setminus \mathcal{O}}$ among remaining candidates, i.e.\ the most efficient when considering only uncovered occurrences, is selected. If $P$ is cost-effective for the remaining uncovered occurrences, it is added to $\mathcal{P}$, $\mathcal{O}$ is updated and the selection proceeds to the next round. Otherwise the selection stops and $\mathcal{P}$ is returned.
}

\section{Experiments}
\label{sec:xps}

In this section, we evaluate the ability of our algorithm to find patterns that compress the input event sequences. We make the code and the prepared datasets publicly available.\footnote{\urlcode{}}
To the best of our knowledge, no existing algorithm carries out an equivalent task and we are therefore unable to perform a comparative evaluation against competitors. To better understand the behaviour of our algorithm, we first performed experiments on synthetic sequences. We then applied our algorithm to real-world sequences including process execution traces, smartphone applications activity, and life-tracking. We evaluate our algorithm's ability to compress the input sequences and present some examples of extracted patterns.

\VShortOnly{
\begin{table}[tbp]
\caption{Statistics of the event log sequences used in the experiments.}
\label{tab:data-stats-traces}
\centering
\begin{tabular}{@{\hspace{1ex}}l@{\hspace{4ex}}r@{\hspace{2ex}}r@{\hspace{2ex}}r@{\hspace{2.5ex}}r@{\hspace{1.5ex}}r@{\hspace{2ex}}r@{\hspace{2ex}}r@{\hspace{2.5ex}}r@{\hspace{1.5ex}}r@{\hspace{1ex}}}
\toprule
 & $\len{\seq}$ & $\tspan{\seq}$ & $\abs{\ABC}$ & \multicolumn{2}{c}{$\len{\seq[\alpha]}$} &  & $\clEmpty$ & \multicolumn{2}{c}{RT (s)} \\ 
 &  &  &  & $\omed$ & $\max$ &  &  & cycles & overall \\ 
\midrule
\dstTZap{0} & $181644$ & $181643$ & $443$ & $22$ & $36697$ &  & $4154277$ & $2094$ & $35048$ \\ 
\dstBugz{0} & $16775$ & $16774$ & $91$ & $6$ & $3332$ &  & $303352$ & $112$ & $522$ \\ 
\dstSamba{} & $28751$ & $7461$ & $119$ & $44$ & $2905$ &  & $520443$ & $214$ & $2787$ \\ 
\dstSachaG{15} & $65977$ & $221445$ & $141$ & $231$ & $4389$ &  & $1573140$ & $2963$ & $14377$ \\ 
\cmidrule{2-10} \\ [-1.2em]
\multicolumn{10}{l}{\dstUbiAR{\iAbs} ($31$ sequences)} \\ 
\midrule
min & $413$ & $11391$ & $10$ & $23$ & $194$ &  & $6599$ & $1$ & $1$ \\ 
median & $23859$ & $87591$ & $87$ & $52$ & $2131$ &  & $486633$ & $232$ & $1020$ \\ 
max & $167863$ & $17900307$ & $241$ & $129$ & $6101$ &  & $3733349$ & $2297$ & $28973$ \\ 
\bottomrule
\end{tabular}
\end{table}
}

For a given event sequence, the main objective of our algorithm is to mine and select a good collection of periodic patterns, in the sense that the collection should allow to compress the input sequence as much as possible. 
Therefore, the main measure that we consider in our experiments is the \emph{compression ratio}, defined as the ratio between the length of the code representing the input sequence with the considered collection of patterns and the length of the code representing the input sequence with an empty collection of patterns, i.e.\ using only individual event occurrences, given as a percentage.
For a given sequence $S$ and collection of patterns $\ccycle$ the compression ratio is defined as
\[ \prcCl{}\label{sym:prcCl} = 100 \cdot \clCC / \clEmpty\;,\]
with smaller values associated to better pattern collections.

\VLongOnly{
\subsection{Mining synthetic sequences}
\label{ssec:xps_synthe}

We begin by probing the behaviour of our algorithm on synthetic sequences containing planted periodic patterns.
 
First we generate sequences that contain a single pattern. Each pattern consists of a basis of one to three events, repeated in a cycle, in two nested cycles or in three nested cycles, that is building pattern trees of depth $1$, $2$ and $3$ respectively.
The simplest basis consists of event $a$, with the period of the inner cycle being either greater than five (specifically, in $[5, 9]$) or greater than $10$ (specifically, in $[10, 24]$).
To build more complex patterns, we use event $a$ followed by event $b$ at distance $4$, i.e.\ \Bstart{}\activity{a} \BinfoDT{4} \activity{b}\Bend{}, as well as event $a$ followed by event $c$ at distance $1$, followed by event $d$ at distance $2$, i.e.\ \Bstart{}\activity{a} \BinfoDT{1} \activity{c} \BinfoDT{2} \activity{d}\Bend{}.

Each resulting perfect synthetic sequence can then be perturbed with \emph{shift noise}, i.e.\ by displacing the occurrences by a few time steps either forward or backward, or with \emph{additive noise}, i.e.\ by adding sporadic occurrences. 
Displacement noise is parameterised, on one hand, by the maximum absolute shift by which the occurrences might be displaced and, on the other hand, by the fraction of occurrences that are displaced. We refer to these two parameters as the \emph{level} and the \emph{density} of the noise, respectively.
For additive noise, we insert occurrences of event $a$ at random timestamps. This type of noise has a single parameter, \emph{density}, fixing the number of of sporadic occurrences as compared to the number of occurrences of the event in the unperturbed sequence.
The generated sequences contain from about fifty up to over two thousand occurrences. 

In each round, we mine each generated sequence in turn for periodic patterns, check whether the planted pattern was recovered exactly and compare the length of the code for encoding the perturbed sequence using either the planted pattern, denoted as $\cl_H$, or those that have been selected by the algorithm, denoted as $\cl_F$.
The first round of experiments is run on sequences with only shift noise. 
The second and third rounds of experiments are run on sequences with additive noise of density $0.1$ and density $0.5$ respectively.
The fourth round is run on sequences with only shift noise, but letting the occurrences of the planted pattern interleave, unlike in the three previous rounds.

In Fig.~\ref{fig:synthe_1_cr}--\ref{fig:synthe_4_cr}, we plot the compression ratio achieved by the planted pattern versus the compression ratio achieved by the pattern collection selected by the algorithm for each of the twenty sequences generated with each considered combination of parameters, for the four rounds respectively.
A different take on the same results is presented in Fig.~\ref{fig:synthe_1}--\ref{fig:synthe_4}, where we show the distribution of $\prcCl_F - \prcCl_H$ among the twenty sequences generated with each combination of parameters as boxplots, for the four rounds respectively. 
A value of $\prcCl_F - \prcCl_H=0$ means that the patterns selected by our algorithm achieve the same compression as the planted patterns, while positive (resp.\ negative) values of $\prcCl_F - \prcCl_H$ correspond to selected patterns achieving longer (resp.\ shorter) code length than with planted patterns.
On the left next to each boxplot, we indicate the number of sequences for which the planted pattern was recovered exactly.

Next, we consider sequences containing multiple planted patterns. For this purpose, we consider the pool of sequences generated in each of the four rounds with single patterns above and generate new sequences by selecting between two and five sequences from the pool and combining them together. The patterns can be combined either with or without overlap, that is, either letting a sequence start before or after the preceding sequence ends. 
The results for the runs over these synthetic sequences containing multiple planted patterns are presented in Fig.~\ref{fig:synthe_comb}.

We see from Fig.~\ref{fig:synthe_1} that when no spurious occurrences are inserted the planted pattern is recovered exactly in most cases for simple patterns of depth one, while the performance deteriorates and fewer planted patterns are recovered for more complex patterns and greater depths, as also visible from Fig.~\ref{fig:synthe_1_cr}. This is expected since recovering multi-event patterns requires that the corresponding cycles are properly recovered in the first stage of the algorithm for each of the events that make up the pattern. Even in the absence of noise, the algorithm might miss the planted pattern, e.g.\ because it merges successive nested repetition of a cycle that appear close to each other.
When the sequences involve interleaving (Fig.~\ref{fig:synthe_4_cr} and~\ref{fig:synthe_4}) the algorithm behaves in a similar way, except for the more complex basis with depths two and three, which are expectedly impacted more strongly by interleaving, resulting in more degraded performances. 

Spurious occurrences break the planted patterns which are no longer recovered by the algorithm. With low density of additive noise the algorithm often selects patterns very similar to the planted one but covering also the spurious occurrences, using shift corrections to accommodate them (Fig.~\ref{fig:synthe_2}). This is typical of the dynamic programming cycle mining, which is able to find cycles with many repetitions but does not allow to skip any occurrence, which are thus incorporated at the cost of increased corrections.
When the density of noise becomes fairly large, the inserted occurrences might actually generate new patterns that can result in shorter code length than the planted pattern, as can be observed in Fig.~\ref{fig:synthe_3}. Indeed, except for the patterns over single event $a$ with long periods, the difference in compression ratios is negative in the majority of cases.

When several planted patterns are combined without overlap, the algorithm is able to recover them all exactly in roughly half of the cases for patterns taken from pools with no additive noise, with or without interleaving ($44$ and $51\%$, respectively, see Fig.~\ref{fig:synthe_comb}).
In most cases the patterns selected by the algorithm yield a longer code length than the planted patterns, except in the presence of dense additive noise. 

Note that the requirement that the planted pattern(s) should be recovered exactly is very strict, as it means that the pattern(s) selected by the algorithm should cover the exact same occurrences as the planted ones, with the exact same pattern tree. 
Closer inspection of the results reveals that the algorithm is able to recover large fragments of the planted patterns in most cases. More specifically, in cases where it fails to recover planted patterns with height greater than one, the algorithm is in general able to identify cycles that constitute large fragments of different repetitions of the inner cycle of the pattern, but merely omitting a few occurrences in these fragment prevents the algorithm from combining them into vertical patterns of greater height. Designing a procedure that is able to build on the extracted fragments from different repetitions to recover the omitted occurrences could make the retrieval of this type of patterns more robust, but is clearly not trivial.

\subsection{Mining real-world sequences}
\label{ssec:xps_quantitative}

Next, we apply our algorithm to real-world datasets.}

\mpara{Datasets.}
Our first two datasets come from a collaboration with STMicroelectronics and are execution traces of a set-top box based on the STiH418 SoC\footnote{STiH418 description: \url{http://www.st.com/resource/en/data_brief/stih314.pdf}} running STLinux.
Both traces are a log of system actions (interruptions, context switches and system calls) taken by the KPTrace instrumentation system developed at STMicroelectronics.
The \dstTZapX{} \VShortOnly{sequence}\VLongOnly{dataset} corresponds to 3 successive changes of channel (``zap''), while the \dstBugzX{} \VShortOnly{sequence}\VLongOnly{dataset} corresponds to logging a display blackout bug into the bug tracking system of ST. \VLongOnly{Each dataset contains two traces, one for either of the two cores of the box, named respectively \dstTZap{0} and \dstTZap{1}, on one hand, \dstBugz{0} and \dstBugz{1}, on the other hand.}
For our analysis of these traces, we do not consider timestamps, only the succession of events.

The \dstUbi{} dataset was obtained from the UCI Machine learning repository.\!\footnote{\url{https://archive.ics.uci.edu/ml/datasets/UbiqLog+(smartphone+lifelogging)}} It contains traces collected from the smartphones of users over the course of two months. For each of $31$ users \VLongOnly{(we excluded those whose data was not encoded using Hindu-Arabic numerals), }we obtain a sequence recording what applications are run on that user's smartphone.   
\VShortOnly{We consider absolute timestamps with a granularity of one minute.}
\VLongOnly{We either consider absolute timestamps with a granularity of one minute or only the succession of events, and denote the corresponding collections of sequences respectively as \dstUbiAR{\iAbs} and \dstUbiAR{\iRel}. }

The \dstSamba{} dataset consists of a single sequence recording the emails identifying the authors of commits on the git repository of the samba network file system\footnote{\url{https://git.samba.org/}} from $1996$ to $2016$.
We consider timestamps with a granularity of one day. \VLongOnly{User commits are instantaneous.} We aggregated together users that appeared fewer than $10$ times\VLongOnly{ as ``other''}.

The \dstSacha{} dataset \VShortOnly{consists of a single sequence containing }\VLongOnly{contains }records from the \textit{quantified awesome} life log\footnote{\url{http://quantifiedawesome.com/records}} recording the daily activities of its author between November 2011 and January 2017. The daily activities are associated to start and end timestamps, and are divided between categories organised into a hierarchy. Categories with fewer than $200$ occurrences were aggregated to their parent category. Each resulting category is represented by an event. Adjacent occurrences of the same event were merged together.
\VShortOnly{We consider absolute timestamps with a granularity of $15$ minutes.}
\VLongOnly{We either consider absolute timestamps with a granularity of one minute or only the succession of events, and denote the corresponding sequences respectively as \dstSachaAR{\iAbs} and \dstSachaAR{\iRel}. 
Further, we investigate what happens when we coarsen the time granularity, from the original one minute to $15$ minutes, $30$ minutes, $1$ hour, half a day and a full day. The corresponding sequences are denoted \dstSachaG{15}, \dstSachaG{30}, \dstSachaG{60}, \dstSachaG{720} and \dstSachaG{1440}, respectively. }

\VLongOnly{When considering absolute timestamps for occurrences involving non-instant processes (e.g.\ daily activities, running applications), each process might be associated with three different events representing its start, its end, and the process happening for a duration smaller than the time granularity respectively.
When considering only the succession of events or, in other words, focusing on the order in which things happen rather than the specific times, we only consider the starting time of the process and each process is hence associated with only one event.}
 
\VShortOnly{Table~\ref{tab:data-stats-traces} presents the statistics of the sequences used in our experiments.}
\VLongOnly{Tables~\ref{tab:data-stats-traces}--\ref{tab:data-stats-UbiqLog-IS-rel} present the statistics of the sequences used in our experiments.}
We indicate the length ($\len{\seq}$) and duration ($\tspan{\seq}$) of each sequence, the size of its alphabet ($\abs{\ABC}$), as well as the median and maximum length of the event subsequences ($\len{\seq[\alpha]}$).
We also indicate the code length of the sequence when encoded with an empty collection of patterns ($\clEmpty$), as well as the running time of the algorithm (RT, in seconds) for mining and selecting the patterns, as well as for the first stage of mining cycles for each separate event.

\mpara{Measures.}
Beside\VLongOnly{ the code length and} the compression ratio\VShortOnly{ ($\prcCl$)} achieved with the selected pattern collections, we consider several other characteristics\VShortOnly{ (see Table~\ref{tab:res-short})}.
For a given pattern collection $\ccycle$, we denote the set of residuals $\residual{\ccycle, \seq}$ simply as $\resSet$\label{sym:resSet} and look at what fraction of the code length is spent on them, denoted as $\ratioClR\label{sym:ratioClR} = \sum_{o \in \resSet} \cl(o)/\clCC$. \VLongOnly{Note that when the pattern collection is empty $\ratioClR=1$, since only residuals are used, and hence the code length results entirely from residuals.
$\abs{\resSet}$ and $\abs{\ccycle}$ are the number of residuals (individual event occurrences) and the number of patterns in the collection, respectively.} We also look at the number of patterns of different types in $\ccycle$\VShortOnly{:}\VLongOnly{, specifically, \textit{(i)}}\VShortOnly{ ($\nbS$)} simple cycles, i.e.\ patterns with \VLongOnly{both width and height equal to $1$}\VShortOnly{$\text{width}=1$ and $\text{height}=1$}, \VLongOnly{\textit{(ii)}}\VShortOnly{($\nbV$)} vertical patterns, \VLongOnly{having a width of $1$ and a height strictly greater than $1$}\VShortOnly{with $\text{width}=1$ and $\text{height}>1$}, \VLongOnly{\textit{(iii)}}\VShortOnly{($\nbH$)} horizontal patterns, \VLongOnly{having a height of $1$ and a width strictly greater than $1$}\VShortOnly{with $\text{width}>1$ and $\text{height}=1$}, and \VLongOnly{\textit{(iv)}}\VShortOnly{($\nbTD$)} proper two-dimensional patterns, \VLongOnly{with both height and width greater than $1$}\VShortOnly{with $\text{width}>1$ and $\text{height}>1$}.
Finally, we look at the \VLongOnly{fraction of patterns in $\ccycle$ that cover strictly more than three occurrences, i.e.\
\[\nbOTC = \abs{\{ P \in \ccycle, \abs{\cov{P}} > 3\}}/\abs{\ccycle}\;,\]
where $\cov{P}$ denotes the set of timestamp--event pairs covered by a pattern $P$, and the median and }maximum cover size of patterns in $\ccycle$\VShortOnly{, denoted as $\nbOmax$}. 

\mpara{Results.}
\VShortOnly{In addition to the final collection of patterns returned by the algorithm after potentially a few rounds of combinations (denoted as $\collF$), we also consider intermediate collections of patterns, namely a collection selected among cycles mined during the first stage of the algorithm (denoted as $\collS$), including patterns from the first round of horizontal combinations ($\collH$), of vertical combinations ($\collV$), and both, i.e.\ at the end of the first round of combinations ($\collVH$).}
\VLongOnly{To better understand the role of the pattern combinations, in addition to looking at the final collection of patterns returned by the algorithm (denoted as $\collF$), we also consider intermediate collections of patterns, namely a collection selected among simple cycles mined during the initial phase of the algorithm (denoted as $\collS$), a collection selected among simple cycles and patterns resulting from the first round of horizontal combinations (denoted as $\collH$), from the first round of vertical combinations (denoted as $\collV$) and from both, or in other words among the candidate patterns obtain at the end of the first round of combinations (denoted as $\collVH$).

Table~\ref{tab:res-long-traces} shows the results for application trace log sequences \dstTZap{0}, \dstTZap{1}, \dstBugz{0}, \dstBugz{1} and \dstSamba{}.
Table~\ref{tab:res-long-sacha_1} shows the results for \dstSacha{} sequences when considering timestamps with different time granularities, as well as when considering only the event succession.
Tables~\ref{tab:res-long-UbiqLog-ISE-abs_1}--\ref{tab:res-long-UbiqLog-ISE-abs_5} show the results for the sequences from the \dstUbiAR{\iAbs} dataset, while tables~\ref{tab:res-long-UbiqLog-IS-rel_1}--\ref{tab:res-long-UbiqLog-IS-rel_5} show the results for the sequences from the \dstUbiAR{\iRel} dataset.

For each sequence and pattern collection we indicate the compression ratio ($\prcCl{}$), the code length ($\cl_{\ccycle}$), the fraction of code used for residual ($\ratioClR{}$), the number of residuals ($\abs{\resSet}$) and of patterns ($\abs{\ccycle}$), the number of simple, vertical, horizontal and two-dimensional patterns ($\nbS$, $\nbV$, $\nbH$, and $\nbTD$, respectively), the fraction of patterns covering more than three occurrences ($\nbOTC$) as well as the median ($\nbOmed$) and the maximum ($\nbOmax$) cover size of patterns in the collection.
}

\VShortOnly{A summary of the results for the separate event sequences \dstTZap{0}, \dstBugz{0}, \dstSamba{} and \dstSacha{}, is presented in Table~\ref{tab:res-short}.}
\VLongOnly{Table~\ref{tab:res-agg} shows aggregated results for the \dstUbiAR{\iAbs} and \dstUbiAR{\iRel} datasets, where we indicate the range of values taken for the different sequences in each subset.}
\VShortOnly{Fig.~\ref{fig:xps-prcCl-ubi-abs} shows the compression ratios achieved on event sequences from the \dstUbiAR{\iAbs} dataset.}
\VLongOnly{Fig.~\ref{fig:xps-prcCl-other}--\ref{fig:xps-prcCl-ubi-rel} show the compression ratios achieved for sequences from the different datasets.}

\VLongOnly{
\begin{table}
\caption{Aggregated results for \dstUbi{} sequences.}
\label{tab:res-agg}
\centering
\begin{tabular}{@{}c@{\hspace{1ex}}r@{\hspace{.7ex}}r@{\hspace{.7ex}}c@{/}c@{/}c@{/}c@{\hspace{.4ex}}r@{}}
\toprule
 & $\prcCl{}$ & $\ratioClR{}$ & \parbox[b][1em][b]{.2cm}{\centering $\nbS$} & \parbox[b][1em][b]{.2cm}{\centering $\nbV$} & \parbox[b][1em][b]{.2cm}{\centering $\nbH$} & \parbox[b][1em][b]{.2cm}{\centering $\nbTD$} & $\nbOmax$ \\ 
\cmidrule{2-8} \\ [-.5em]
\multicolumn{8}{c}{\dstUbiAR{\iAbs} (31)} \\ 
\midrule
$\collS$ & [$40.18$, $85.52$] & [$0.22$, $0.60$] & [$41$, $9468$] & [$0$, $0$] & [$0$, $0$] & [$0$, $0$] & [$17$, $388$] \\ 
$\collV$ & [$40.17$, $85.52$] & [$0.23$, $0.60$] & [$41$, $9445$] & [$0$, $57$] & [$0$, $0$] & [$0$, $0$] & [$17$, $388$] \\ 
$\collH$ & [$30.08$, $84.33$] & [$0.24$, $0.60$] & [$31$, $3113$] & [$0$, $0$] & [$5$, $2256$] & [$0$, $0$] & [$17$, $2328$] \\ 
$\collVH$ & [$30.08$, $84.33$] & [$0.24$, $0.60$] & [$31$, $3107$] & [$0$, $4$] & [$5$, $2252$] & [$0$, $0$] & [$17$, $2328$] \\ 
$\collF$ & [$30.06$, $84.33$] & [$0.24$, $0.60$] & [$31$, $3102$] & [$0$, $2$] & [$5$, $2233$] & [$0$, $11$] & [$17$, $2328$] \\ 
[.5em]
\multicolumn{8}{c}{\dstUbiAR{\iRel} (31)} \\ 
\midrule
$\collS$ & [$26.05$, $64.94$] & [$0.12$, $0.45$] & [$9$, $2567$] & [$0$, $0$] & [$0$, $0$] & [$0$, $0$] & [$158$, $8500$] \\ 
$\collV$ & [$26.05$, $64.94$] & [$0.12$, $0.45$] & [$9$, $2567$] & [$0$, $2$] & [$0$, $0$] & [$0$, $0$] & [$158$, $8500$] \\ 
$\collH$ & [$25.91$, $63.48$] & [$0.12$, $0.41$] & [$9$, $2083$] & [$0$, $0$] & [$0$, $339$] & [$0$, $0$] & [$158$, $35300$] \\ 
$\collVH$ & [$25.91$, $63.48$] & [$0.12$, $0.41$] & [$9$, $2083$] & [$0$, $2$] & [$0$, $334$] & [$0$, $0$] & [$158$, $35300$] \\ 
$\collF$ & [$25.91$, $63.48$] & [$0.12$, $0.41$] & [$9$, $2083$] & [$0$, $2$] & [$0$, $334$] & [$0$, $1$] & [$158$, $35300$] \\
\bottomrule
\end{tabular}
\end{table}
}

We see that the algorithm is able to find sets of patterns that compress the input event sequences. The compression ratio varies widely depending on the considered sequence, from a modest $84\%$ for some sequences from \dstUbiAR{\iAbs} to a reduction of more than two thirds, for instance for \dstSamba{}.
To an extent, the achieved compression can be interpreted as an indicator of how much periodic structure is present in the sequence (at least of the type that can be exploited by our proposed encoding and detected by our algorithm).
In some cases, as with \dstSamba{}, the compression is achieved almost exclusively with simple cycles, but in many cases the final selection contains a large fraction of horizontal patterns (sometimes even about two thirds), which bring a noticeable improvement in the compression ratio (as can be seen in Fig.~\ref{fig:xps-prcCl-ubi-abs}, for instance).    
Vertical patterns, on the other hand, are much more rare, and proper two-dimensional patterns are almost completely absent. The \dstBugzX{} sequence\VLongOnly{s} feature\VShortOnly{s} such patterns, and even more so the \dstTZapX{} sequence\VLongOnly{s}. This agrees with the intuition that recursive periodic structure is more likely to be found in execution logs tracing multiple recurrent automated processes. 

\VLongOnly{
\begin{figure}[tbp]
  \centering
  \begin{tabular}{@{}ccc@{}}
    \includegraphics[trim=5 0 20 0,width=.32\textwidth,clip]{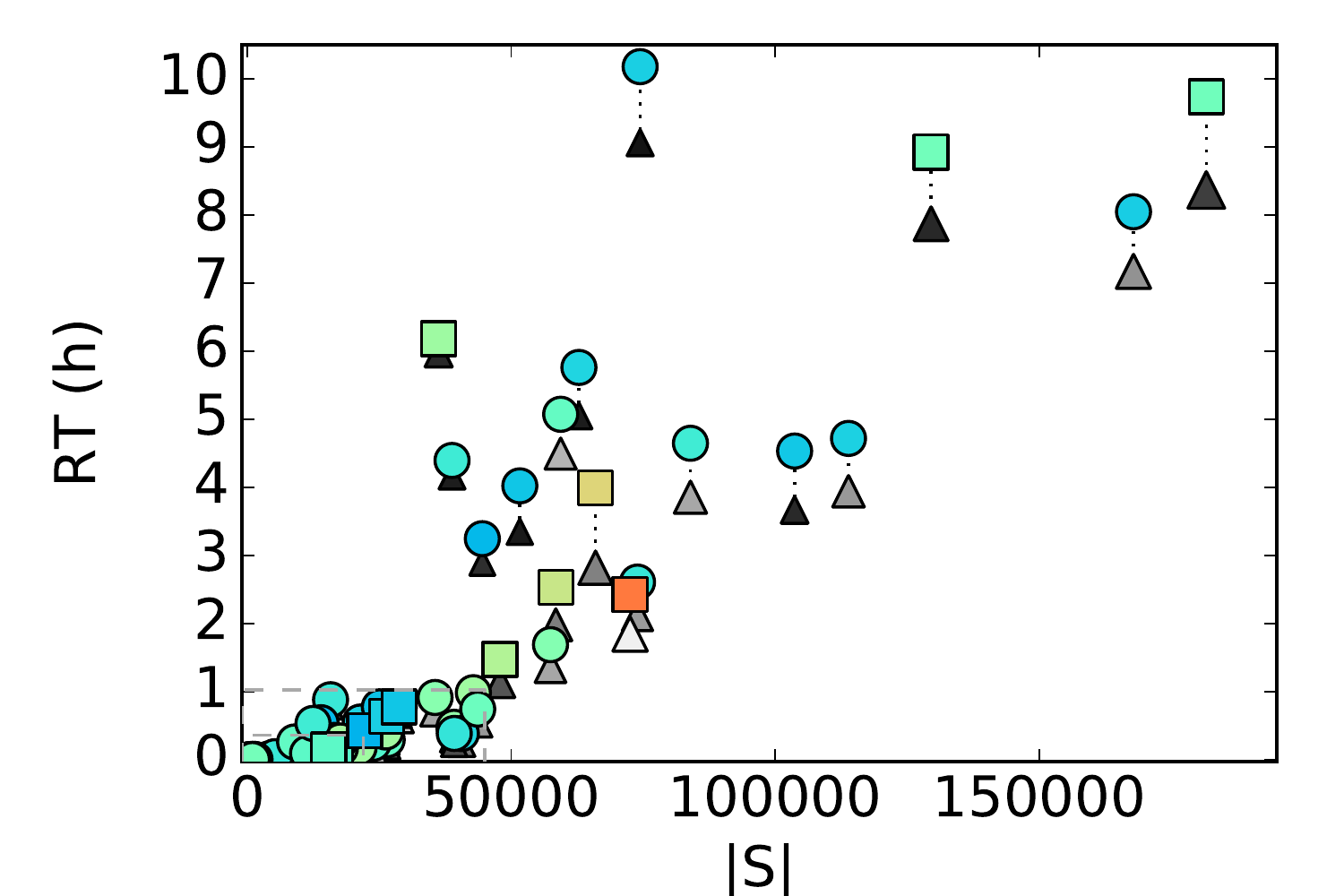} &
    \includegraphics[trim=5 0 20 0,width=.32\textwidth,clip]{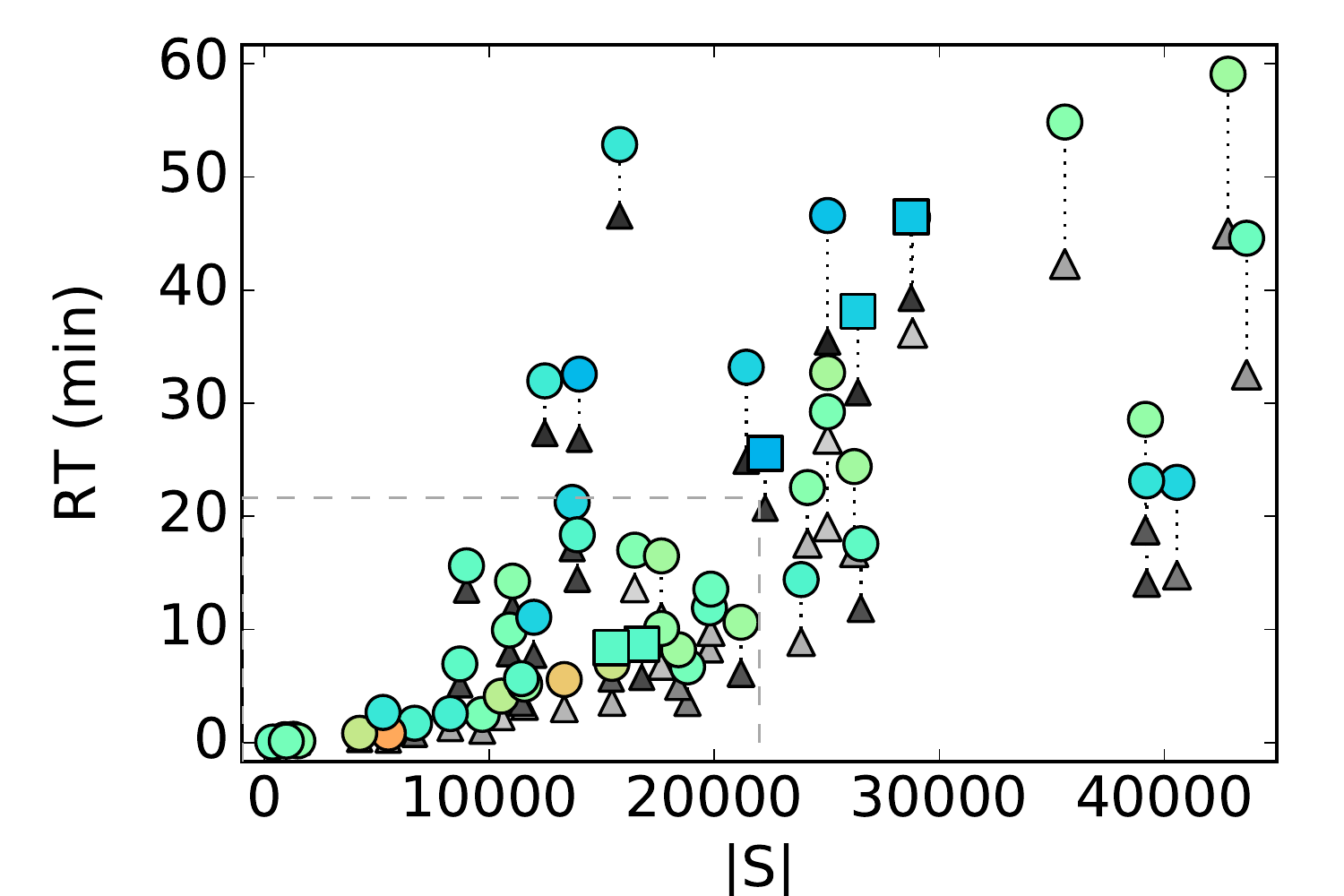} &
    \includegraphics[trim=5 0 20 0,width=.32\textwidth,clip]{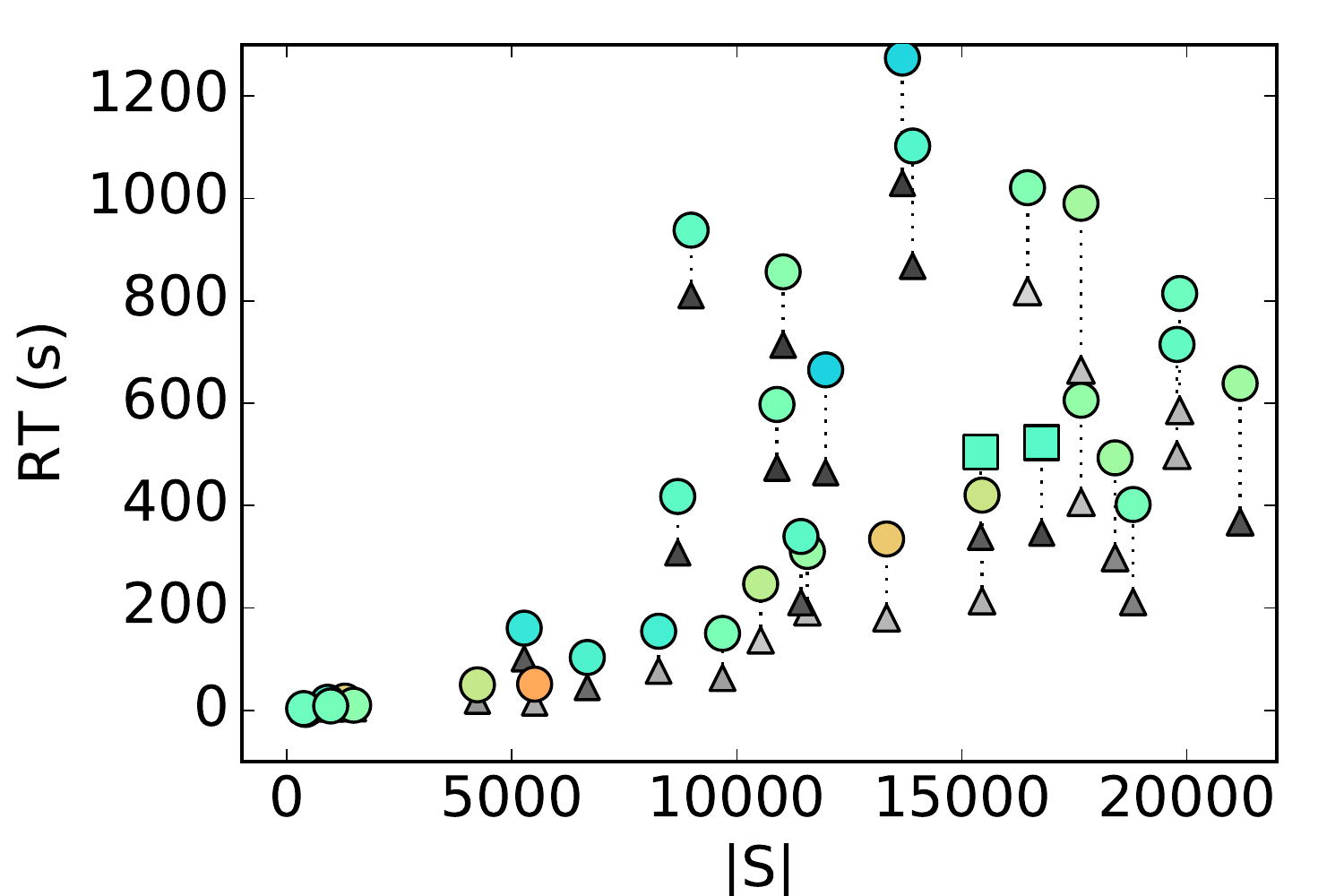} \\
    \multicolumn{3}{c}{\hfill \includegraphics[trim=100 133 10 222,width=.45\textwidth,clip]{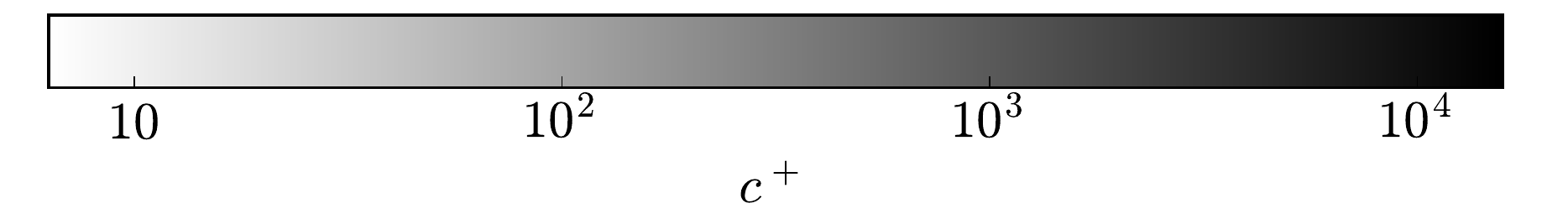} \hfill
    \includegraphics[trim=100 45 10 310,width=.45\textwidth,clip]{fig_times_l_details} \hfill} \\
  \end{tabular}
\caption{Running times for sequences from the different datasets, in hours (left) and zoomed-in in minutes (middle) and seconds (right).}
\label{fig:xps-times}
\end{figure}
}

\VShortOnly{
\begin{figure}[tbp]
  \centering
  \begin{tabular}{@{}cc@{}}
    \includegraphics[trim=5 0 20 0,height=3.5cm,clip]{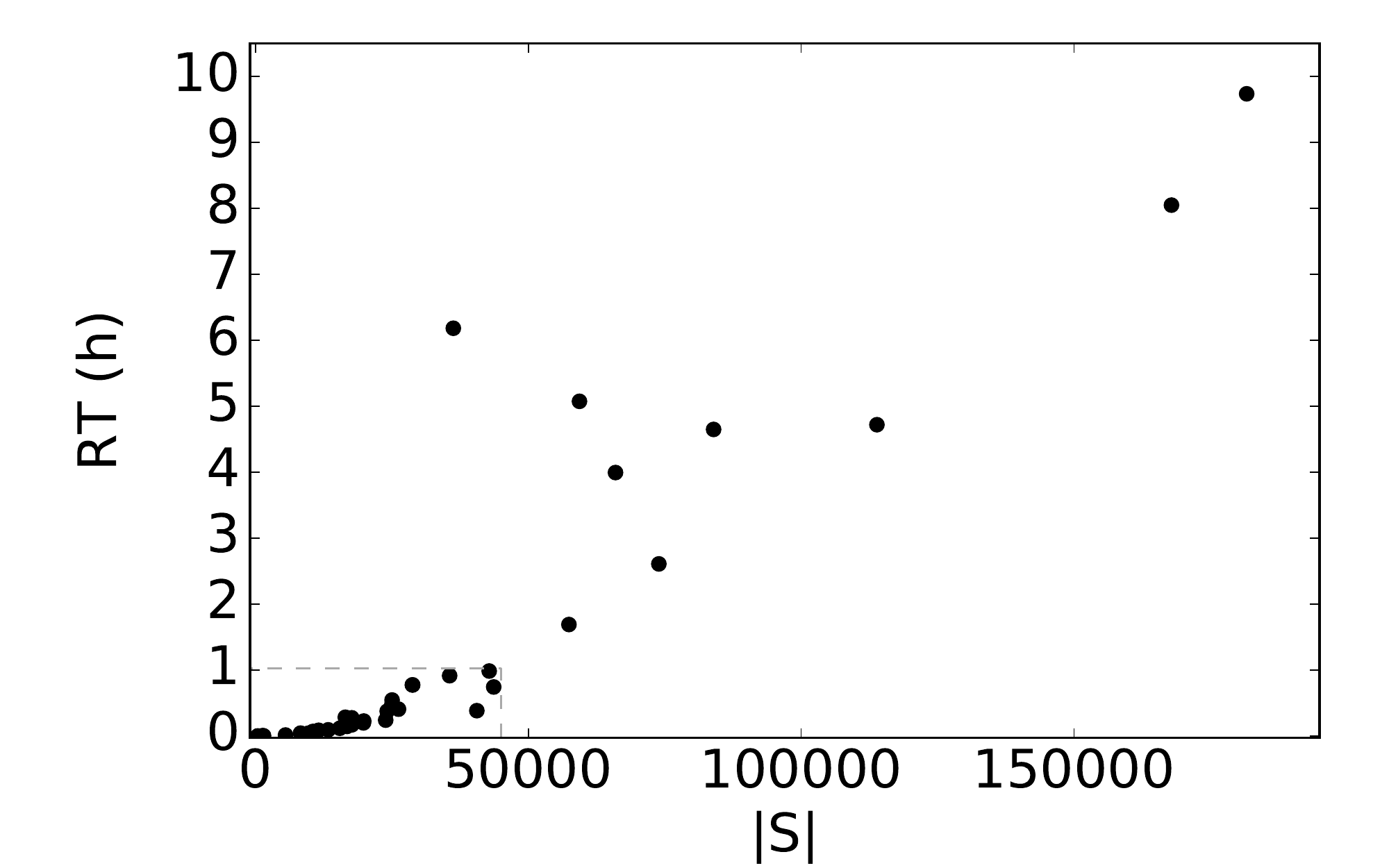} &
    \includegraphics[trim=5 0 20 0,height=3.5cm,clip]{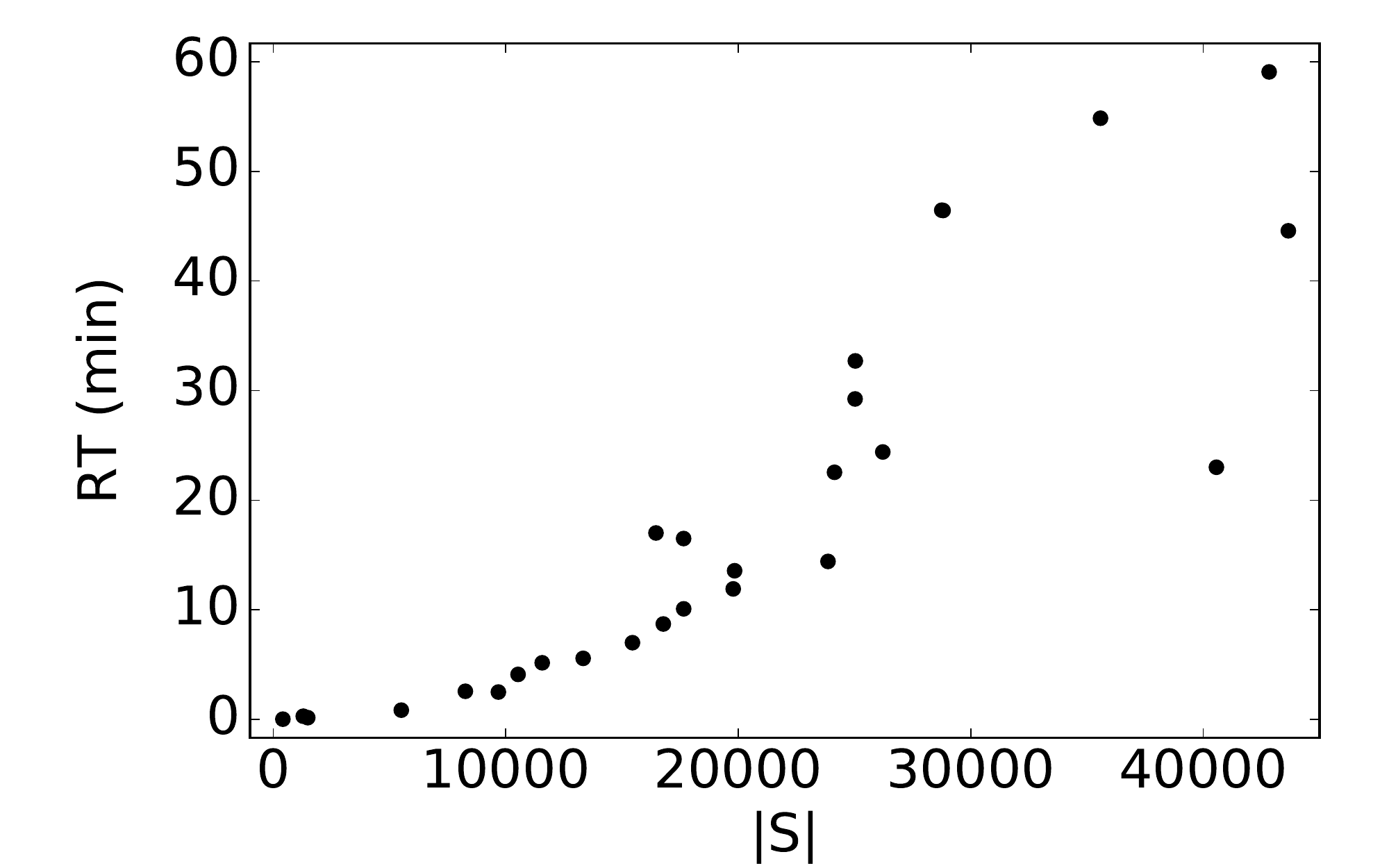} \\
  \end{tabular}
\caption{Running times for mining the different sequences (in hours, left) and zooming in on shorter sequences (in minutes, right).}
\label{fig:xps-times}
\end{figure}
}

\VLongOnly{
In most cases, a large proportion of the selected patterns cover more than the minimum three timestamp--event pairs. Some of the largest patterns cover several hundreds or a few thousand occurrences, depending on the length of the input sequence, obviously, as well as the strength of its periodic structure).   
Obviously, the more occurrences a pattern covers, the more efficient it is, assuming it can be represented concisely.

From Table~\ref{tab:res-long-sacha_1} we can see that the chosen time granularity has a strong impact on the extracted patterns. With the finest time granularity, i.e.\ $1$ minute time step (\dstSachaG{1}), few patterns are found because the activities need to reoccur with minute regularity and any deviation must be accounted in the shift corrections. Therefore periodic patterns are not very efficient and only little compression is achieved. When increasing the time granularity to $15$ minutes, $30$ minutes and to $1$ hour (respectively \dstSachaG{15}, \dstSachaG{30} and \dstSachaG{60}) allows to be more forgiving of small deviations the exact times when activities happen, resulting in more efficient patterns found. This is evidenced by a sharp decrease in the fraction of simple cycles ($\nbS/\abs{\ccycle}$) and increase in the fraction of patterns covering more than three occurrences ($\nbOTC$) and the maximum cover size ($\nbOmax$).
Further coarsening the time granularity, to a half day and a full day (\dstSachaG{720} and \dstSachaG{1440}) the fraction of simple cycles among the selected pattern increases again, but this time each one covers a large number of occurrences.
At such level of granularity, the time and order in which the activities are carried out during the day no longer matter, only which activities are performed on any given day.
Finally, with type of data considering the succession of activities rather than absolute timestamps (\dstSachaAR{\iRel}) might allow to identify fairly different patterns, since activities in a pattern are no longer separated by a time span but by the number of other activities performed in between. However, in this context, this can result in patterns that are difficult to understand, since they cannot be easily mapped back to time points and hence calendar dates and hours of the days cannot be used when interpreting the patterns.  
Hence, the choice of using succession or absolute timestamps, and, in the latter case, of choosing the granularity of the time step, has to be made by the analyst in consideration of the context and the time scale that is of interest.

In some cases (e.g.\ \dstBugz{0} in Table~\ref{tab:res-long-traces}, \dstSachaG{60} in Table~\ref{tab:res-long-sacha_1} and several \dstUbi{} sequences), the collection of patterns selected from the final set of candidates, $\collF$, achieves worse compression than collections selected from intermediate sets of candidates, despite the fact that the intermediate candidate sets are subsets of the final one. This is due to the fact that the pattern selection, which is in essence a weighted set cover problem is solved greedily (see Section~\ref{sec:algo}), and a local decision of choosing a more efficient pattern produced in later combination rounds, might eventually result in degraded compression. However, the degradation is fairly limited and one might simply decide to replace the final solution by an intermediate one, when the candidates produced later on do not appear to contribute to shortening the code length.}

Fig.~\ref{fig:xps-times} shows the running times for sequences from the different datasets.
\VLongOnly{Circles and squares, coloured according to achieved compression ratio, indicate the running time of the algorithm for sequences from the \dstUbi{} dataset and from other datasets, respectively.
Each such marker is connected to a triangle indicating the running time for the combination rounds. Larger triangles correspond to sequences for which more simple cycles are extracted during the initialisation phase. Darker triangles correspond to sequences for which the maximum cover size among these simple cycles is larger.}
The running times vary greatly, from only a few seconds to several hours. Naturally, mining longer sequences tends to require longer running times. \VLongOnly{However, directly observable characteristics of the sequence, such as its size, the size of its alphabet, relative frequencies of the events, etc.\ are not the only factors impacting the running time. The number and length of the cycles extracted in the first stage have a major effect on the time required by the combination rounds, i.e.\ the second stage, which take the bulk of the overall running time. Indeed, if the initial candidates contain many long cycles, many more tests will be needed when trying to combine them into more complex patterns.}

\setlength{\ndhlen}{2.2em}
\setlength{\ndwlen}{.8cm}

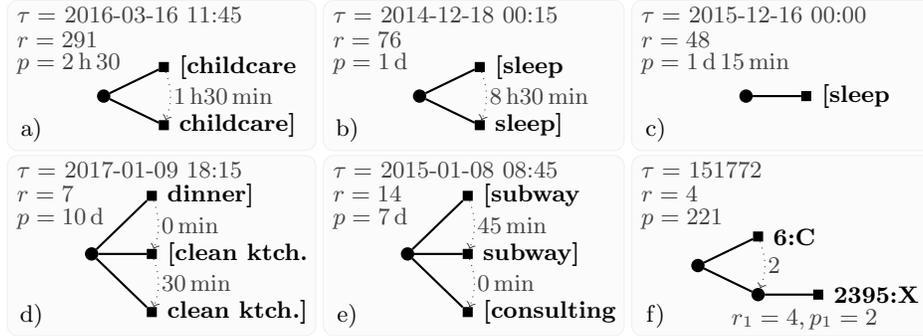
\begin{figure}[tbp] %
\centering
\begin{tikzpicture}[-,auto,node distance=\ndhlen, thick]
\draw[info rec, xshift=-1.4\ndwlen, yshift=1.4\ndhlen+.7em] (0, 0) rectangle +(4.1, -2.4) {};
\threeEvEx{0,0}{7}{\SI{10}{\day}}{2017-01-09 18:15}{dinner]}{\SI{0}{\minute}}{[clean ktch.}{\SI{30}{\minute}}{clean ktch.]}
\node[xshift=-1.\ndwlen, yshift=1.4\ndhlen+1.2em] at (0,-2.3) {d)};
\draw[info rec, xshift=-1.4\ndwlen, yshift=1.4\ndhlen+.7em] (4.2, 0) rectangle +(4., -2.4) {};
\threeEvEx{4.2,0}{14}{\SI{7}{\day}}{2015-01-08 08:45}{[subway}{\SI{45}{\minute}}{subway]}{\SI{0}{\minute}}{[consulting}
\node[xshift=-1.\ndwlen, yshift=1.4\ndhlen+1.2em] at (4.2,-2.3) {e)};
\draw[info rec, xshift=-1.4\ndwlen, yshift=1.4\ndhlen+.7em] (8.3, 0) rectangle +(4., -2.4) {};
\otherEvEx{8.3,0}{4}{221}{151772}{6:C}{2}{2395:X}{4}{2}
\node[xshift=-1.\ndwlen, yshift=1.4\ndhlen+1.2em] at (8.3,-2.3) {f)};

\draw[info rec, xshift=-1.4\ndwlen, yshift=1.4\ndhlen+.9em] (0, 2.) rectangle +(4.1, -2.) {};
\twoEvEx{0,2.1}{291}{\SI{2}{\hour}\,30}{2016-03-16 11:45}{[childcare}{\SI{1}{\hour} \SI{30}{\minute}}{childcare]}
\node[xshift=-1.\ndwlen, yshift=1.4\ndhlen+1.2em] at (0,.18) {a)};
\draw[info rec, xshift=-1.4\ndwlen, yshift=1.4\ndhlen+.9em] (4.2, 2.) rectangle +(4., -2.) {};
\twoEvEx{4.2,2.1}{76}{\SI{1}{\day}}{2014-12-18 00:15}{[sleep}{\SI{8}{\hour} \SI{30}{\minute}}{sleep]}
\node[xshift=-1.\ndwlen, yshift=1.4\ndhlen+1.2em] at (4.2,.18) {b)};
\draw[info rec, xshift=-1.4\ndwlen, yshift=1.4\ndhlen+.9em] (8.3, 2.) rectangle +(4., -2.) {};
\oneEvEx{8.3,2.1}{48}{\SI{1}{\day}\,\SI{15}{\minute}}{2015-12-16 00:00}{[sleep}
\node[xshift=-1.\ndwlen, yshift=1.4\ndhlen+1.2em] at (8.3,.18) {c)};
\end{tikzpicture}
\caption{Example patterns from \dstSachaG{15} (a--e) and \dstTZap{0} (f).}
\label{fig:res-ex}
\end{figure}

\mpara{Example patterns.}
Finally, we present some examples of patterns obtained from the \dstSachaG{15} and \dstTZap{0} sequences, in Fig.~\ref{fig:res-ex}. The start and end of an activity A are denoted as ``[A'' and ``A]'' respectively.
The patterns from the \dstSachaG{15} sequence are simple and rather obvious, but they make sense when considering everyday activities. The fact that we are able to find them is a clear sign that the method is working. The \dstTZap{0} pattern is a typical system case: the repetition of a context switch (6:C) followed by several activations of a process (2395:X). 
\VLongOnly{Further examples can be found in Tables~\ref{tab:res-ex-sacha} and~\ref{tab:res-ex-zap}. In \dstTZap{0} patterns, event names consist of a numerical part, indicating the process id, and one or two letter indicating the action. Upper and lower case letters represent the start and end of an action, respectively. The most common actions are interruption (I), context switch (C), system call (X), user function call (U).}

\medskip

Most of the discovered patterns are fairly simple. We suspect
that this is due to the nature of the data: there are no significantly complex patterns
in these event log sequences.
In any case, the expressivity of our proposed pattern language comes at no detriment to the simpler, more common patterns, but brings the potential benefit of identifying sequences containing exceptionally regular structure.

\section{Conclusion}
\label{sec:conclusion}

In this paper, we propose a novel approach for mining periodic patterns with a MDL criterion, and an algorithm to put it into practise.
Through our experimental evaluation, we show that we are able to extract sets of patterns that compress the input event sequences and to identify meaningful patterns. 

\VLongOnly{An analyst parsing a log might have some intuition about what periods are more meaningful, as well as relations and dependencies between events, depending on the generating process. For instance, we expect days and weeks to strongly structure life tracking logs, while patterns with periods of, say, 21 hours or 17 days would be considered less intuitive. }
How to take\VLongOnly{ such} prior knowledge into account is an interesting question to explore.\VLongOnly{

} Making the algorithm more robust to noise and making it more scalable using for instance parallelisation, are some pragmatic directions for future work, as is adding a visualisation tool to support the analysis and interpretation of the extracted patterns in the context of the event log sequence.

\mpara{Acknowledgements.}
The authors thank Hiroki Arimura and Jilles Vreeken for valuable discussions.
This work has been supported by Grenoble Alpes Metropole through the Nano2017 Itrami project, by the QCM-BioChem project (CNRS Mastodons) and by the Academy of Finland projects ``Nestor'' (286211) and ``Agra'' (313927).

\bibliographystyle{abbrv}
\bibliography{lib}

\appendix
\FloatBarrier
\newgeometry{margin=1.in,asymmetric}
\renewcommand\thefigure{A.\arabic{figure}}
\renewcommand\thetable{A.\arabic{table}}

\begin{figure}[p]
\centering
\begin{tabular}{@{\hspace{.5cm}}rcc@{\hspace{.5cm}}rcc@{\hspace{.5cm}}}
\toprule
$\Ptree_1$ & \qquad \BinfoRPT{4}{2} \Bstart{}\activity{a}\Bend{} & & $\Ptree_2$ & \qquad \BinfoRPT{3}{13} \Bstart{}\activity{a}\Bend{} & \\
&
\begin{tikzpicture}[-,auto]
\TreeOneEvtOneCyclVar{t1}{at (2,0)}{$4,2$}{a}
\end{tikzpicture} &
\begin{tikzpicture}[-,auto]
\LowestRep{1}{0}{0}{4}
\end{tikzpicture} & &
\begin{tikzpicture}[-,auto]
\TreeOneEvtOneCyclVar{t1}{at (2,0)}{$3,13$}{a}
\end{tikzpicture} &
\begin{tikzpicture}[-,auto]
\LowestRep{1}{0}{0}{3}
\end{tikzpicture}
\\ [.5em]
$\mapOids{\Ptree_1}=$ & \multicolumn{2}{p{.33\textwidth}}{$\lls (\Pblock_1, \LL{1}),$ $(\Pblock_1, \LL{2}),$ $(\Pblock_1, \LL{3}),$ $(\Pblock_1, \LL{4})\lle$} &
$\mapOids{\Ptree_2}=$ & \multicolumn{2}{p{.33\textwidth}}{$\lls (\Pblock_1, \LL{1}),$ $(\Pblock_1, \LL{2}),$ $(\Pblock_1, \LL{3})\lle$} \\
$\occsStar{\Ptree_1}=$ & \multicolumn{2}{p{.33\textwidth}}{$\lls (0, a),$ $(2, a),$ $(4, a),$ $(6, a) \lle$} &
$\occsStar{\Ptree_2}=$ & \multicolumn{2}{p{.33\textwidth}}{$\lls (0, a),$ $(13, a),$ $(26, a) \lle$} \\
\end{tabular}

\begin{tabular}{rcc}
\midrule
$\Ptree_3$ & \BinfoRPT{3}{13} \Bstart{}\BinfoRPT{4}{2} \Bstart{}\activity{a}\Bend{}\Bend{} & \\
& \begin{tikzpicture}[-,auto]
\TreeOneEvtTwoCyclVar{t2}{at (4,0)}{$3,13$}{$4,2$}{a}
\end{tikzpicture} &
\renewcommand{\Xslant}{2.5}
\renewcommand{\Yslant}{-.3}
\begin{tikzpicture}[-,auto]
  \path[rep edge] (0,0) edge ({(3-\rootf)*\Xslant},{(3-\rootf)*\Yslant});
   \node[main node] (R00) at (0,0) {};
\foreach \j in {1,...,3}{
 \LowestRep{\j}{(\j-\rootf)*\Xslant}{(\j-\rootf)*\Yslant}{4}
}
\end{tikzpicture} \\ [.5em]
$\mapOids{\Ptree_3}=$ & \multicolumn{2}{p{.8\textwidth}}{$\lls (\Pblock_{11}, \LL{1,1}),$ $(\Pblock_{11}, \LL{1,2}),$ $(\Pblock_{11}, \LL{1,3}),$ $(\Pblock_{11}, \LL{2,1}),$ $(\Pblock_{11}, \LL{2,2}),$ $(\Pblock_{11}, \LL{2,3}),$ $(\Pblock_{11}, \LL{3,1}),$ $(\Pblock_{11}, \LL{3,2}),$ $(\Pblock_{11}, \LL{3,3})\lle$} \\
$\occsStar{\Ptree_3}=$ & \multicolumn{2}{p{.8\textwidth}}{$\lls (0, a),$ $(2, a),$ $(4, a),$ $(6, a),$ $(13, a),$ $(15, a),$ $(17, a),$ $(19, a),$ $(26, a),$ $(28, a),$ $(30, a),$ $(32, a) \lle$} \\
\midrule
$\Ptree_4$ & \qquad \BinfoRPT{4}{2} \Bstart{}\BinfoRPT{3}{13} \Bstart{}\activity{a}\Bend{}\Bend{} & \\
& \begin{tikzpicture}[-,auto]
\TreeOneEvtTwoCyclVar{t3}{at (4,0)}{$4,2$}{$3,13$}{a}
\end{tikzpicture} &
\renewcommand{\Xslant}{2}
\renewcommand{\Yslant}{-.25}
\begin{tikzpicture}[-,auto]
  \path[rep edge] (0,0) edge ({(4-\rootf)*\Xslant},{(4-\rootf)*\Yslant});
  \node[main node] (R00) at (0,0) {};
  \foreach \j in {1,...,4}{
  \LowestRep{\j}{(\j-\rootf)*\Xslant}{(\j-\rootf)*\Yslant}{3}
}
\end{tikzpicture} \\ [.5em]
$\mapOids{\Ptree_4}=$ & \multicolumn{2}{p{.8\textwidth}}{$\lls (\Pblock_{11}, \LL{1,1}),$ $(\Pblock_{11}, \LL{1,2}),$ $(\Pblock_{11}, \LL{1,3}),$ $(\Pblock_{11}, \LL{2,1}),$ $(\Pblock_{11}, \LL{2,2}),$ $(\Pblock_{11}, \LL{2,3}),$ $(\Pblock_{11}, \LL{3,1}),$ $(\Pblock_{11}, \LL{3,2}),$ $(\Pblock_{11}, \LL{3,3})\lle$} \\
$\occsStar{\Ptree_4}=$ & \multicolumn{2}{p{.8\textwidth}}{$\lls (0, a),$ $(13, a),$ $(26, a),$ $(2, a),$ $(15, a),$ $(28, a),$ $(4, a),$ $(17, a),$ $(30, a),$ $(6, a),$ $(19, a),$ $(32, a) \lle$} \\
\bottomrule
\end{tabular}
\caption{Pattern trees $\Ptree_1$--$\Ptree_4$: Pattern and expansion trees, lists of leaf nodes and of perfect occurrences. }
\label{fig:ex_tree1a}
\end{figure}

\begin{figure}[p]
\centering
\begin{tabular}{rcc}
\toprule
$\Ptree_5$ & \qquad \BinfoRPT{3}{13} \Bstart{}\activity{b} \BinfoDT{3} \activity{a} \BinfoDT{1} \activity{c}\Bend{} & \\
& \begin{tikzpicture}[-,auto]
\TreeThreeEvtOneCyclVar{t3}{at (9,0)}{$3,13$}{b}{$3$}{a}{$1$}{c}
\end{tikzpicture} &
\renewcommand{\LowestBranch}[4]{\LowestBranchB{#1}{#2}{#3}{#4}}
\renewcommand{\xslant}{1.8}
\renewcommand{\yslant}{-.3}
\begin{tikzpicture}[-,auto]
\LowestRep{1}{0}{0}{3}
\end{tikzpicture} \\ [.5em]
$\mapOids{\Ptree_5}=$ & \multicolumn{2}{p{.8\textwidth}}{$\lls (\Pblock_{1}, \LL{1}),$ $(\Pblock_{2}, \LL{1}),$ $(\Pblock_{3}, \LL{1}),$ $(\Pblock_{1}, \LL{2}),$ $(\Pblock_{2}, \LL{2}),$ $(\Pblock_{3}, \LL{2}),$ $(\Pblock_{1}, \LL{3}),$ $(\Pblock_{2}, \LL{3}),$ $(\Pblock_{3}, \LL{3}) \lle$} \\
$\occsStar{\Ptree_5}=$ & \multicolumn{2}{p{.8\textwidth}}{$\lls (0, b),$ $(3, a),$ $(4, c), (13, b),$ $(16, a),$ $(17, c), (26, b),$ $(29, a),$ $(30, c) \lle$} \\
\midrule
$\Ptree_6$ & \qquad \BinfoRPT{5}{4} \Bstart{}\activity{b} \BinfoDT{3} \activity{a} \BinfoDT{1} \activity{c}\Bend{} & \\
& \begin{tikzpicture}[-,auto]
\TreeThreeEvtOneCyclVar{t3}{at (12,0)}{$5,4$}{b}{$3$}{a}{$1$}{c} 
\end{tikzpicture} &
\renewcommand{\LowestBranch}[4]{\LowestBranchB{#1}{#2}{#3}{#4}}
\renewcommand{\xslant}{1.8}
\renewcommand{\yslant}{-.3}
\begin{tikzpicture}[-,auto]
\LowestRep{1}{0}{0}{5}
\end{tikzpicture} \\ [.5em]
$\mapOids{\Ptree_6}=$ & \multicolumn{2}{p{.8\textwidth}}{$\lls (\Pblock_{1}, \LL{1}),$ $(\Pblock_{2}, \LL{1}),$ $(\Pblock_{3}, \LL{1}),$ $(\Pblock_{1}, \LL{2}),$ $(\Pblock_{2}, \LL{2}),$ $(\Pblock_{3}, \LL{2}),$ $(\Pblock_{1}, \LL{3}),$ $(\Pblock_{2}, \LL{3}),$ $(\Pblock_{3}, \LL{3}),$ $(\Pblock_{1}, \LL{4}),$ $(\Pblock_{2}, \LL{4}),$ $(\Pblock_{3}, \LL{4}),$ $(\Pblock_{1}, \LL{5}),$ $(\Pblock_{2}, \LL{5}),$ $(\Pblock_{3}, \LL{5}) \lle$} \\
$\occsStar{\Ptree_6}=$ & \multicolumn{2}{p{.8\textwidth}}{$\lls (0, b),$ $(3, a),$ $(4, c), (4, b),$ $(7, a),$ $(8, c), (8, b),$ $(11, a),$ $(12, c),$ $(12, b),$ $(15, a),$ $(16, c), (16, b),$ $(19, a),$ $(20, c) \lle$} \\
\bottomrule
\end{tabular}
\caption{Pattern trees $\Ptree_5$--$\Ptree_6$: Pattern and expansion trees, lists of leaf nodes and of perfect occurrences. }
\label{fig:ex_tree1b}
\end{figure}

\begin{figure}[p]
\centering
\begin{tabular}{rcc}
\toprule
$\Ptree_7$ & \multicolumn{2}{l}{\qquad \BinfoRPT{3}{10} \Bstart{}\activity{b} \BinfoDT{3} \BinfoRPT{4}{1} \Bstart{}\activity{a}\Bend{} \BinfoDT{1} \activity{c}\Bend{}} \\ 
& \begin{tikzpicture}[-,auto]
\TreeThreeEvtTwoCyclVar{t3}{at (15,0)}{$3,10$}{b}{$3$}{a}{$1$}{c}{$4,1$}
\end{tikzpicture} &
\renewcommand{\LowestBranch}[4]{\LowestBranchC{#1}{#2}{#3}{#4}}
\renewcommand{\xslant}{1.5}
\renewcommand{\yslant}{-.5}
\begin{tikzpicture}[-,auto]
\LowestRep{1}{0}{0}{3}
\end{tikzpicture} \\ [.5em]
$\mapOids{\Ptree_7}=$ & \multicolumn{2}{p{.8\textwidth}}{$\lls (\Pblock_1, \LL{1}),$ $(\Pblock_{21}, \LL{1, 1}),$ $(\Pblock_{21}, \LL{1, 2}),$ $(\Pblock_{21}, \LL{1, 3}),$ $(\Pblock_{21}, \LL{1, 4}),$ $(\Pblock_{3}, \LL{1}),$ $(\Pblock_1, \LL{2}),$ $(\Pblock_{21}, \LL{2, 1}),$ $(\Pblock_{21}, \LL{2, 2}),$ $(\Pblock_{21}, \LL{2, 3}),$ $(\Pblock_{21}, \LL{2,4}),$ $(\Pblock_{3}, \LL{2}),$ $(\Pblock_1, \LL{3}),$ $(\Pblock_{21}, \LL{3, 1}),$ $(\Pblock_{21}, \LL{3, 2}),$ $(\Pblock_{21}, \LL{3, 3}),$ $(\Pblock_{21}, \LL{3,4}),$ $(\Pblock_{3}, \LL{3}) \lle$} \\
$\occsStar{\Ptree_7}=$ & \multicolumn{2}{p{.8\textwidth}}{$\lls (0, b),$ $(3, a),$ $(4, a),$ $(5, a),$ $(6, a),$ $(4, c), (10, b),$ $(13, a),$  $(14, a),$ $(15, a),$ $(16, a),$ $(14, c), (20, b),$ $(23, a),$ $(24, a),$ $(25, a),$ $(26, a),$ $(24, c) \lle$} \\
\midrule
$\Ptree_8$ & \multicolumn{2}{l}{\qquad \BinfoRPT{2}{33} \Bstart{}\BinfoRPT{3}{10} \Bstart{}\activity{b} \BinfoDT{3} \BinfoRPT{4}{1} \Bstart{}\activity{a}\Bend{} \BinfoDT{5} \activity{c}\Bend{}\Bend{}} \\ [.8em]
& \begin{tikzpicture}[-,auto]
 \TreeThreeEvtThreeCyclVar{t3}{at (18.5,0)}{$3,10$}{b}{$3$}{a}{$5$}{c}{$4,1$}
\end{tikzpicture} &
\renewcommand{\LowestBranch}[4]{\LowestBranchC{#1}{#2}{#3}{#4}}
\renewcommand{\xslant}{1.5}
\renewcommand{\yslant}{-.5}
\renewcommand{\Xslant}{4.5}
\renewcommand{\Yslant}{-.35}
\begin{tikzpicture}[-,auto]
  \path[rep edge] ({(.85-\rootf)*\Xslant},{(.85-\rootf)*\Yslant}) edge ({(2-\rootf)*\Xslant},{(2-\rootf)*\Yslant});
  \node[main node] (R00) at ({(.85-\rootf)*\Xslant},{(.85-\rootf)*\Yslant}) {};
  \foreach \j in {1,...,2}{
  \LowestRep{\j}{(\j-\rootf)*\Xslant}{(\j-\rootf)*\Yslant}{3}
}
\end{tikzpicture} \\ [.5em]
$\mapOids{\Ptree_8}=$ & \multicolumn{2}{p{.8\textwidth}}{$\lls (\Pblock_{11}, \LL{1, 1}),$ $(\Pblock_{121}, \LL{1, 1, 1}),$ $(\Pblock_{121}, \LL{1, 1, 2}),$ $(\Pblock_{121}, \LL{1, 1, 3}),$ $(\Pblock_{121}, \LL{1, 1, 4}),$ $(\Pblock_{13}, \LL{1, 1}),$ $(\Pblock_{11}, \LL{1, 2}),$ $(\Pblock_{121}, \LL{1, 2, 1}),$ $(\Pblock_{121}, \LL{1, 2, 2}),$ $(\Pblock_{121}, \LL{1, 2, 3}),$ $(\Pblock_{121}, \LL{1, 2,4}),$ $(\Pblock_{13}, \LL{1, 2}),$ $(\Pblock_{11}, \LL{1, 3}),$ $(\Pblock_{121}, \LL{1, 3, 1}),$ $(\Pblock_{121}, \LL{1, 3, 2}),$ $(\Pblock_{121}, \LL{1, 3, 3}),$ $(\Pblock_{121}, \LL{1, 3,4}),$ $(\Pblock_{13}, \LL{1, 3}),$ $(\Pblock_{11}, \LL{2, 1}),$ $(\Pblock_{121}, \LL{2, 1, 1}),$ $(\Pblock_{121}, \LL{2, 1, 2}),$ $(\Pblock_{121}, \LL{2, 1, 3}),$ $(\Pblock_{121}, \LL{2, 1, 4}),$ $(\Pblock_{13}, \LL{2, 1}),$ $(\Pblock_{11}, \LL{2, 2}),$ $(\Pblock_{121}, \LL{2, 2, 1}),$ $(\Pblock_{121}, \LL{2, 2, 2}),$ $(\Pblock_{121}, \LL{2, 2, 3}),$ $(\Pblock_{121}, \LL{2, 2,4}),$ $(\Pblock_{13}, \LL{2, 2}),$ $(\Pblock_{11}, \LL{2, 3}),$ $(\Pblock_{121}, \LL{2, 3, 1}),$ $(\Pblock_{121}, \LL{2, 3, 2}),$ $(\Pblock_{121}, \LL{2, 3, 3}),$ $(\Pblock_{121}, \LL{2, 3,4}),$ $(\Pblock_{13}, \LL{2, 3}) \lle$} \\
$\occsStar{\Ptree_8}=$ & \multicolumn{2}{p{.8\textwidth}}{$\lls (0, b),$ $(3, a),$ $(4, a),$ $(5, a),$ $(6, a),$ $(8, c),$ $(10, b),$ $(13, a),$ $(14, a),$ $(15, a),$ $(16, a),$ $(18, c),$ $(20, b),$ $(23, a),$ $(24, a),$ $(25, a),$ $(26, a),$ $(28, c),$ $(33, b),$ $(36, a),$ $(37, a),$ $(38, a),$ $(39, a),$ $(41, c),$ $(43, b),$ $(46, a),$ $(47, a),$ $(48, a),$ $(49, a),$ $(51, c),$ $(53, b),$ $(56, a),$ $(57, a),$ $(58, a),$ $(59, a),$ $(61, c) \lle$} \\
\bottomrule
\end{tabular}
\caption{Pattern trees $\Ptree_7$--$\Ptree_8$: Pattern and expansion trees, lists of leaf nodes and of perfect occurrences. }
\label{fig:ex_tree2}
\end{figure}

\FloatBarrier
\renewcommand{\xsctm}{.47}
\begin{figure}[p]
\centering
\begin{tabular}{@{}c@{}}
\toprule
\timelineQThreeOne{$\patt_{3,1}$} \\ \midrule
\timelinePThree{$(\Ptree_3, 0, \nullC)$} \\ \midrule
\timelinePFour{$(\Ptree_4, 0, \nullC)$} \\
\bottomrule
\end{tabular}
\caption{Patterns $\patt_{3,1}$, $(\Ptree_3, 0, \nullC)$ and $(\Ptree_4, 0, \nullC)$ shown on timelines.}
\label{fig:ex_timelinesP3-P4}
\end{figure}

\begin{figure}[p]
\centering
\begin{tabular}{@{}c@{}}
\toprule
\timelineQSixOne{$\patt_{6,1}$} \\ \midrule
\timelinePFive{$(\Ptree_5, 0, \nullC)$} \\ 
\bottomrule
\end{tabular}
\caption{Pattern $\patt_{6,1}$ and $(\Ptree_5, 0, \nullC)$ shown on timeline.}
\label{fig:ex_timelinesP5}
\end{figure}

\FloatBarrier

\newcommand{\ccwidth}{\hspace{.2cm}}
\newcommand{\cswidth}{\hspace{.05cm}}

\begin{table}[tbp]
\centering
\caption{Code lengths for the example pattern collection $\ccycle_{1}$.}
\label{tab:ex_clC1}
\begin{tabular}{cr@{\ccwidth{}}c@{\cswidth{}}c@{}r@{\ccwidth{}}c@{\cswidth{}}c@{}r@{\ccwidth{}}c@{\cswidth{}}c@{}r}
\toprule
& \multicolumn{9}{l}{$\ccycle_{1}$} & $76.681$ \\ [.8em]
\cmidrule{2-11} %
& & \multicolumn{2}{c}{$\patt_{1,1}$} & $24.657$ & \multicolumn{2}{c}{$\patt_{1,2}$} & $26.417$ & \multicolumn{2}{c}{$\patt_{1,3}$} & $25.607$ \\
\cmidrule{3-11} %
$\evtseq$ && (a) & & $4.755$ & (a) & & $4.755$ & (a) & & $4.755$ \\
$\Csc$ && $\LL{1, 0, -1}$ & & $8.000$ & $\LL{0, 3, -1}$ & & $10.000$ & $\LL{1, 1, -1}$ & & $9.000$ \\
$\Clen_0$ && $4$ & $\log(12)=$ & $3.585$ & $4$ & $\log(12)=$ & $3.585$ & $4$ & $\log(12)=$ & $3.585$ \\
$\Cprd_0$ && $2$ & $\log(11)=$ & $3.459$ & $2$ & $\log(10)=$ & $3.322$ & $2$ & $\log(11)=$ & $3.459$ \\
$\Cto$ && $2$ & $\log(29)=$ & $4.858$ & $13$ & $\log(27)=$ & $4.755$ & $26$ & $\log(28)=$ & $4.807$ \\
\bottomrule
\end{tabular}
\end{table}

\begin{table}[tbp]
\centering
\caption{Code lengths for the example pattern collection $\ccycle_{2}$.}
\label{tab:ex_clC2}
\begin{tabular}{cr@{\ccwidth{}}c@{\cswidth{}}c@{}r@{\ccwidth{}}c@{\cswidth{}}c@{}r@{\ccwidth{}}c@{\cswidth{}}c@{}r@{\ccwidth{}}c@{\cswidth{}}c@{}r}
\toprule
& \multicolumn{12}{l}{$\ccycle_{2}$} & $87.437$ \\ [.8em]
\cmidrule{2-14} %
&& \multicolumn{2}{c}{$\patt_{2,1}$} & $21.969$ & \multicolumn{2}{c}{$\patt_{2,2}$} & $23.969$ & \multicolumn{2}{c}{$\patt_{2,3}$} & $20.749$ & \multicolumn{2}{c}{$\patt_{2,4}$} & $20.749$ \\
\cmidrule{3-14} %
$\evtseq$ && (a) & & $4.755$ & (a) & & $4.755$ & (a) & & $4.755$ & (a) & & $4.755$ \\
$\Csc$ && $\LL{-2, 0}$ & & $6.000$ & $\LL{-3, 1}$ & & $8.000$ & $\LL{0, -1}$ & & $5.000$ & $\LL{0, -1}$ & & $5.000$ \\
$\Clen_0$ && $3$ & $\log(12)=$ & $3.585$ & $3$ & $\log(12)=$ & $3.585$ & $3$ & $\log(12)=$ & $3.585$ & $3$ & $\log(12)=$ & $3.585$ \\
$\Cprd_0$ && $13$ & $\log(18)=$ & $4.170$ & $13$ & $\log(18)=$ & $4.170$ & $13$ & $\log(17)=$ & $4.087$ & $13$ & $\log(17)=$ & $4.087$ \\
$\Cto$ && $2$ & $\log(11)=$ & $3.459$ & $5$ & $\log(11)=$ & $3.459$ & $7$ & $\log(10)=$ & $3.322$ & $8$ & $\log(10)=$ & $3.322$ \\
\bottomrule
\end{tabular}
\end{table}

\begin{table}[tbp]
\centering
\caption{Code lengths for the example pattern collections $\ccycle_{3}$ and $\ccycle_{4}$.}
\label{tab:ex_clC3-4}
\begin{tabular}{cr@{\ccwidth{}}c@{\cswidth{}}c@{}rr@{\ccwidth{}}r@{\ccwidth{}}c@{\cswidth{}}c@{}r}
\toprule
&\multicolumn{3}{l}{$\ccycle_{3}$} & $59.724$ & & \multicolumn{3}{l}{$\ccycle_{4}$} & $63.920$ \\ [.8em]
\cmidrule{2-5} \cmidrule{7-10} %
&& \multicolumn{2}{c}{$\patt_{3,1}$} & $59.724$ & & & \multicolumn{2}{c}{$\patt_{4,1}$} & $63.920$ \\
\cmidrule{3-5} \cmidrule{8-10} %
$\evtseq$ && ((a)) & & $7.925$ & & & ((a)) & & $7.925$ \\
$\Csc$ && $\LL{1, 0, \dots}$ & & $33.000$ &&& $\LL{-2, 0, \dots}$ & & $32.000$ \\
$\Clen_0$ && $3$ & $\log(12)=$ & $3.585$ &&& $4$ & $\log(12)=$ & $3.585$ \\
$\Clen_1$ && $4$ & $\log(12)=$ & $3.585$ &&& $3$ & $\log(12)=$ & $3.585$ \\
$\Cprd_0$ && $13$ & $\log(18)=$ & $4.170$ &&& $2$ & $\log(11)=$ & $3.459$ \\
$\Cto$ && $2$ & $\log(11)=$ & $3.459$ &&& $2$ & $\log(29)=$ & $4.858$ \\
$\optspanRep^{*}$ && $6$ & $\log(8)=$ & $3.000$ &&& $26$ & $\log(28)=$ & $4.807$ \\
$\Cprd_{1}$ && $2$ & $\log(2)=$ & $1.000$ &&& $13$ & $\log(13)=$ & $3.700$ \\
\bottomrule
\end{tabular}
\end{table}

\FloatBarrier

\begin{table}[tbp]
\centering
\caption{Code lengths for the example pattern collection $\ccycle_{5}$.}
\label{tab:ex_clC5}
\begin{tabular}{cr@{\ccwidth{}}c@{\cswidth{}}c@{}r@{\ccwidth{}}c@{\cswidth{}}c@{}r@{\ccwidth{}}c@{\cswidth{}}c@{}r}
\toprule
& \multicolumn{9}{l}{$\ccycle_{5}$} & $65.443$ \\ [.8em]
\cmidrule{2-11} %
& & \multicolumn{2}{c}{$\patt_{5,1}$} & $21.554$ & \multicolumn{2}{c}{$\patt_{5,2}$} & $20.334$ & \multicolumn{2}{c}{$\patt_{5,3}$} & $23.554$ \\
\cmidrule{3-11} %
$\evtseq$ && (b) & & $6.340$ & (a) & & $6.340$ & (c) & & $6.340$ \\
$\Csc$ && $\LL{-2, 0}$ & & $6.000$ & $\LL{0, -1}$ & & $5.000$ & $\LL{1, -3}$ & & $8.000$ \\
$\Clen_0$ && $3$ & $\log(3)=$ & $1.585$ & $3$ & $\log(3)=$ & $1.585$ & $3$ & $\log(3)=$ & $1.585$ \\
$\Cprd_0$ && $13$ & $\log(18)=$ & $4.170$ & $13$ & $\log(17)=$ & $4.087$ & $13$ & $\log(18)=$ & $4.170$ \\
$\Cto$ && $2$ & $\log(11)=$ & $3.459$ & $5$ & $\log(10)=$ & $3.322$ & $7$ & $\log(11)=$ & $3.459$ \\
\bottomrule
\end{tabular}
\end{table}

\begin{table}[tbp]
\centering
\caption{Code lengths for the example pattern collection $\ccycle_{6}$.}
\label{tab:ex_clC6}
\begin{tabular}{cr@{\ccwidth{}}c@{\cswidth{}}c@{}r}
\toprule
&\multicolumn{3}{l}{$\ccycle_{6}$} & $53.538$ \\ [.8em]
\cmidrule{2-5}
&&\multicolumn{2}{c}{$\patt_{6,1}$} & $53.538$ \\
\cmidrule{3-5}
$\evtseq$ && (b a c) & & $12.680$ \\
$\Csc$ && $\LL{0, 1, \dots}$ & & $24.000$ \\
$\Clen_0$ && $3$ & $\log(3)=$ & $1.585$ \\
$\Cprd_0$ && $13$ & $\log(18)=$ & $4.170$ \\
$\Cto$ && $2$ & $\log(11)=$ & $3.459$ \\
$d_{12}$ && $3$ & $\log(4)=$ & $2.000$ \\
$d_{13}$ && $1$ & $\log(4)=$ & $2.000$ \\
$\optspanRep^{*}$ && $4$ & $\log(8)=$ & $3.000$ \\
\bottomrule
\end{tabular}
\end{table}

\FloatBarrier

\renewcommand{\PlenI}{r_I}
\renewcommand{\PlenJ}{r_J}
\renewcommand{\PnbOI}{\ell}
\renewcommand{\PnbOJ}{\ell'}

\begin{figure}[p]
\centering
\begin{tabular}{c} 
\begin{tikzpicture}[-,auto, xscale=\xsc, yscale=-\ysc]
\blocksPCombineVOccs{}
\end{tikzpicture} \\ [.2cm]
\begin{tikzpicture}[-,auto, xscale=\xsc, yscale=-\ysc]
\blocksPCombineVPR{}
\end{tikzpicture} \\ [.2cm]
\begin{tikzpicture}[-,auto, xscale=\xsc, yscale=-\ysc]
\blocksPCombineVPC{}
\end{tikzpicture} \\ [.2cm]
\begin{tikzpicture}[-,auto, xscale=\xsc, yscale=-\ysc]
\blocksPCombineV{}
\end{tikzpicture} 
\end{tabular}
\caption{Vertical combination: Combining patterns $\patt_{I,1}, \dots,\patt_{I,\PlenJ}$ into nested pattern $\patt_{N}$. Rounded rectangles represent event occurrences. Each colored rectangle represents a pattern and encloses the occurrence covered by the pattern. Arrows link occurrences to the preceding occurrences relative to which their timestamp is computed.}
\label{fig:ex_combineV}
\end{figure}

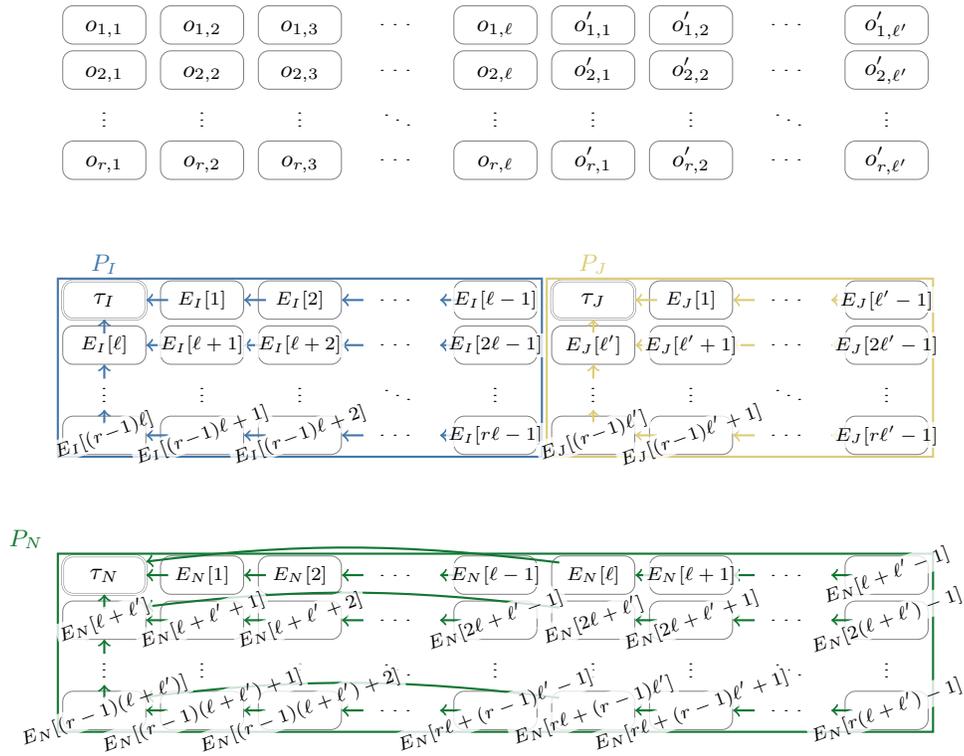
\begin{figure}[p]
\centering
\begin{tabular}{c}
\begin{tikzpicture}[-,auto, xscale=\xsc, yscale=-\ysc]
\blocksPCombineHOccs{}
\end{tikzpicture}  \\ [.2cm]
\begin{tikzpicture}[-,auto, xscale=\xsc, yscale=-\ysc]
\blocksPCombineHRC{}
\end{tikzpicture} \\ [.2cm]
\begin{tikzpicture}[-,auto, xscale=\xsc, yscale=-\ysc]
\blocksPCombineH{}
\end{tikzpicture} 
\end{tabular}
\caption{Horizontal combination: Concatenating patterns $\patt_{I}$ and $\patt_{J}$ into new pattern $\patt_{N}$. Rounded rectangles represent event occurrences. Each colored rectangle represents a pattern and encloses the occurrence covered by the pattern. Arrows link occurrences to the preceding occurrences relative to which their timestamp is computed.}
\label{fig:ex_combineH}
\end{figure}

\FloatBarrier

\begin{figure}[tbp] \centering
\caption{Compression ratios for planted and extracted pattern collections ($\prcCl_H$ and $\prcCl_F$, respectively) on synthetic sequences perturbed only by shift noise.}
\label{fig:synthe_1_cr}
\includegraphics[trim=0 0 0 0,clip,width=.9\textwidth]{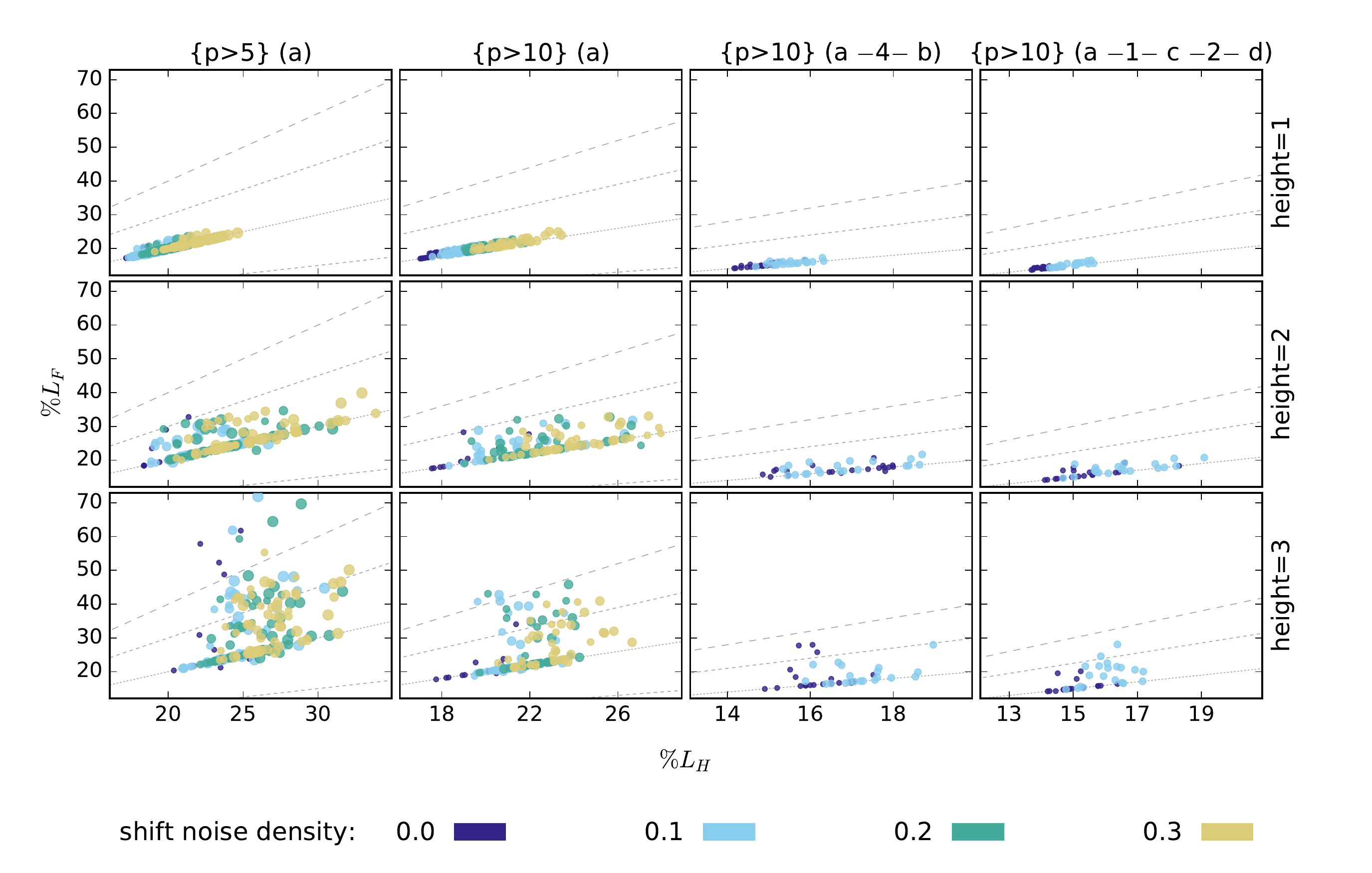}
\end{figure}

\begin{figure}[tbp] \centering
\caption{Compression ratios for planted and extracted pattern collections ($\prcCl_H$ and $\prcCl_F$, respectively) on synthetic sequences perturbed by additive noise $(a, 0.1)$.}
\label{fig:synthe_2_cr}
\includegraphics[trim=0 0 0 0,clip,width=.9\textwidth]{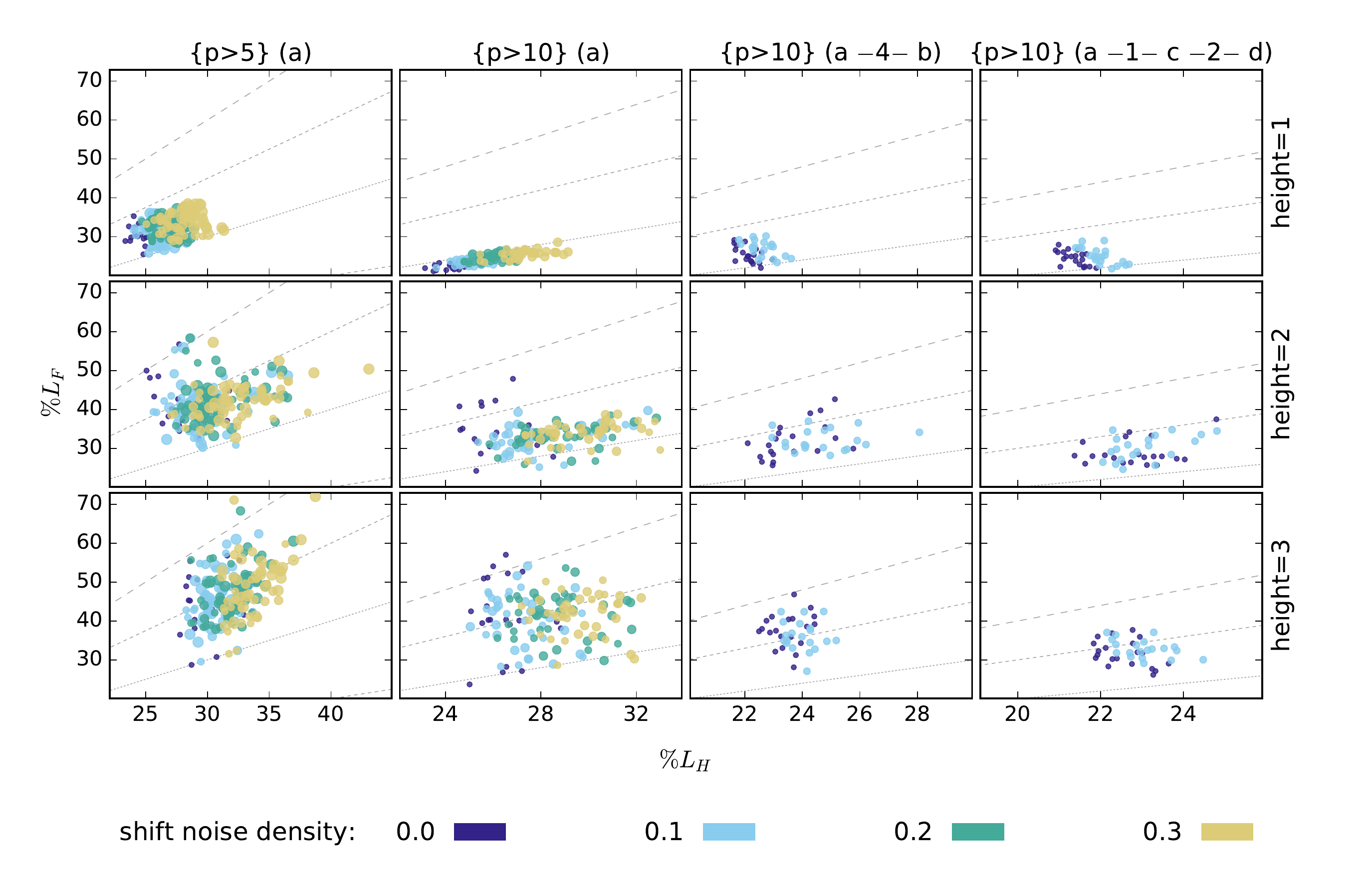}
\end{figure}

\begin{figure}[tbp] \centering
\caption{Compression ratios for planted and extracted pattern collections ($\prcCl_H$ and $\prcCl_F$, respectively) on synthetic sequences perturbed by additive noise $(a, 0.5)$.}
\label{fig:synthe_3_cr}
\includegraphics[trim=0 0 0 0,clip,width=.9\textwidth]{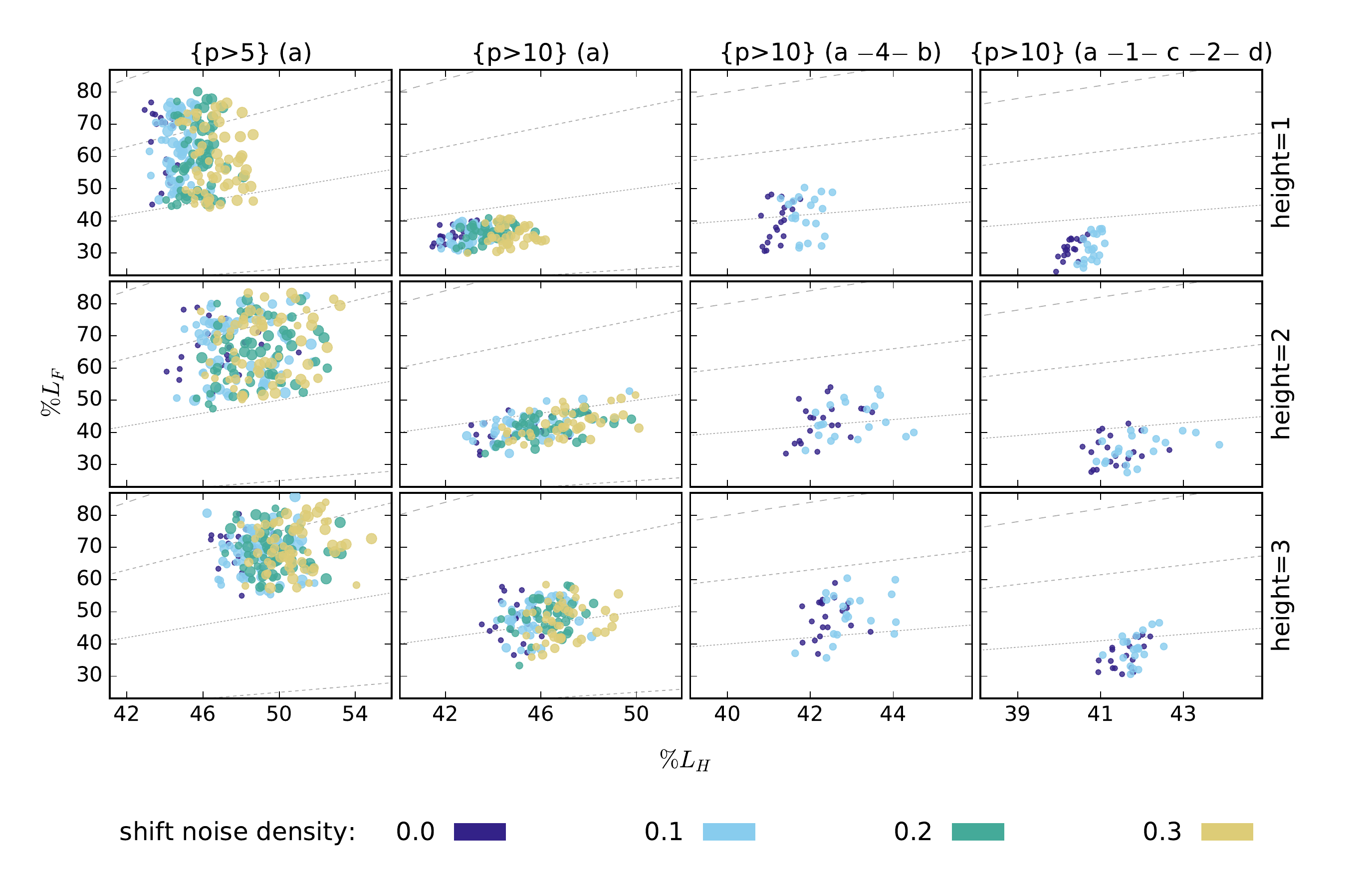}
\end{figure}

\begin{figure}[tbp] \centering
\caption{Compression ratios for planted and extracted pattern collections ($\prcCl_H$ and $\prcCl_F$, respectively) on synthetic sequences containing interleaving.}
\label{fig:synthe_4_cr}
\includegraphics[trim=0 0 0 0,clip,width=.9\textwidth]{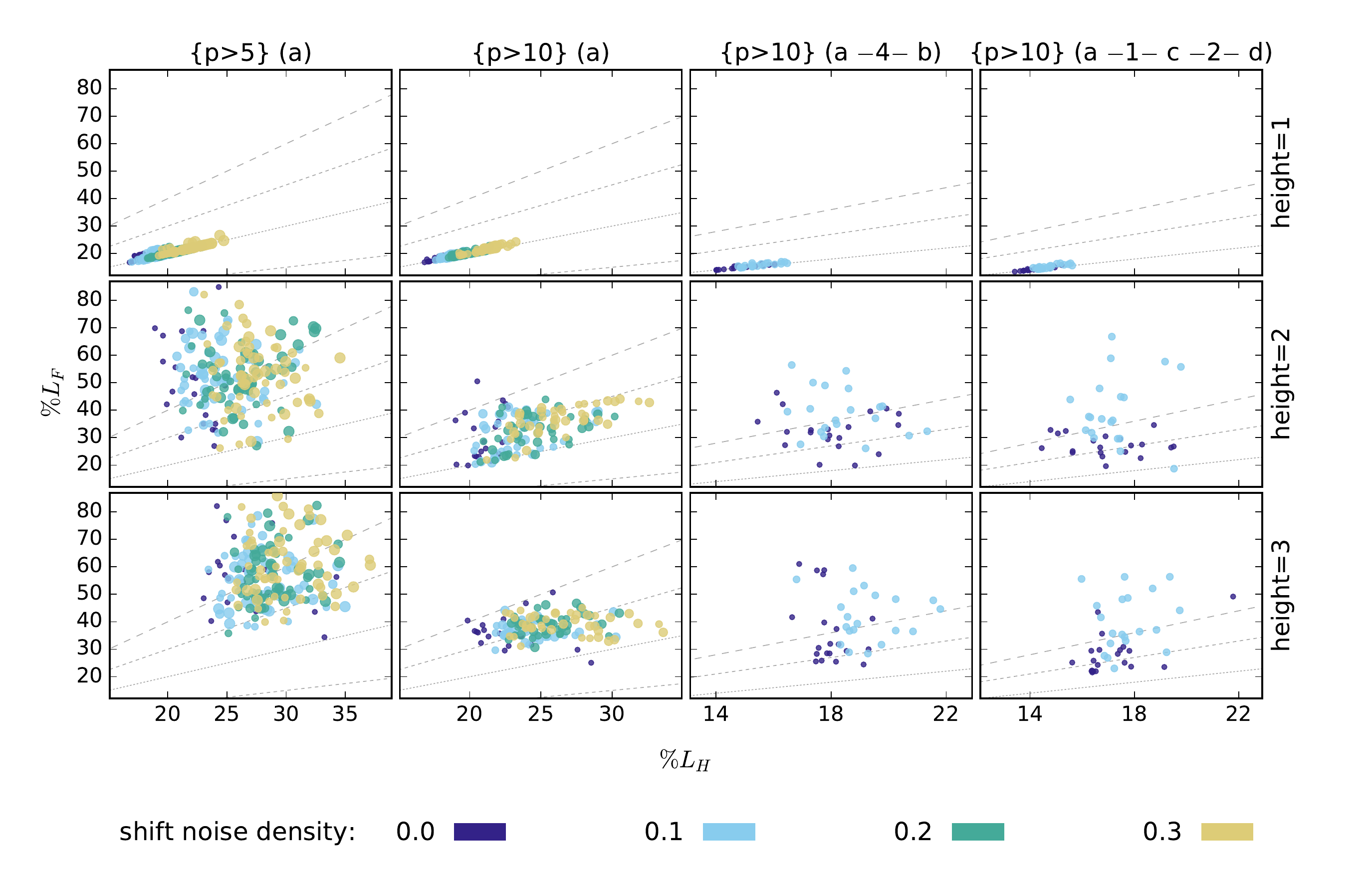}
\end{figure}

\begin{figure}[tbp] \centering
\caption{Differences in compression ratios for planted and extracted pattern collections ($\prcCl_H$ and $\prcCl_F$, respectively) on synthetic sequences perturbed only by shift noise.}
\label{fig:synthe_1}
\includegraphics[trim=30 20 100 20,clip,width=.9\textwidth]{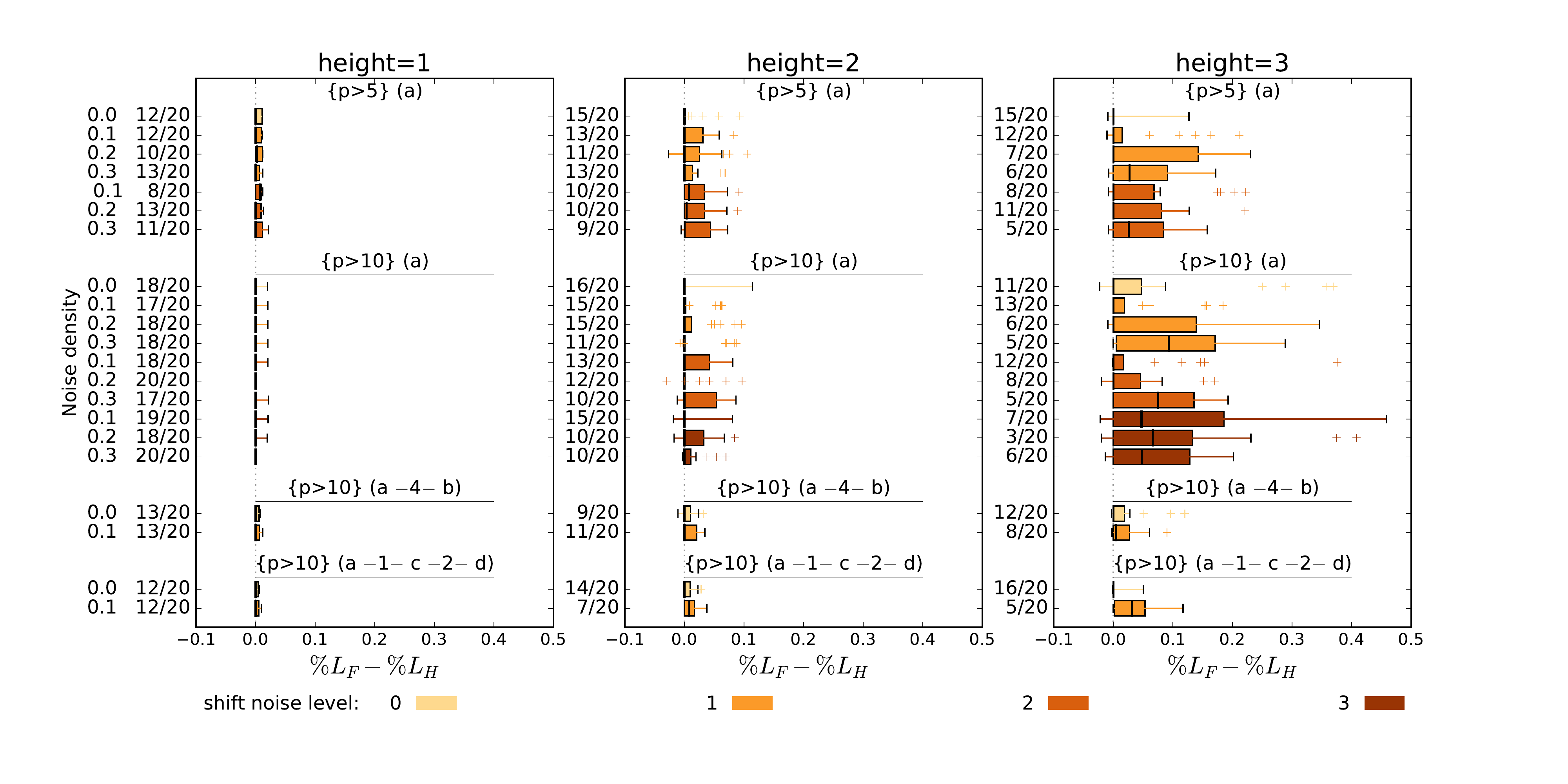}
\end{figure}

\begin{figure}[tbp] \centering
\caption{Differences in compression ratios for planted and extracted pattern collections ($\prcCl_H$ and $\prcCl_F$, respectively) on synthetic sequences perturbed by additive noise $(a, 0.1)$.}
\label{fig:synthe_2}
\includegraphics[trim=30 20 100 20,clip,width=.9\textwidth]{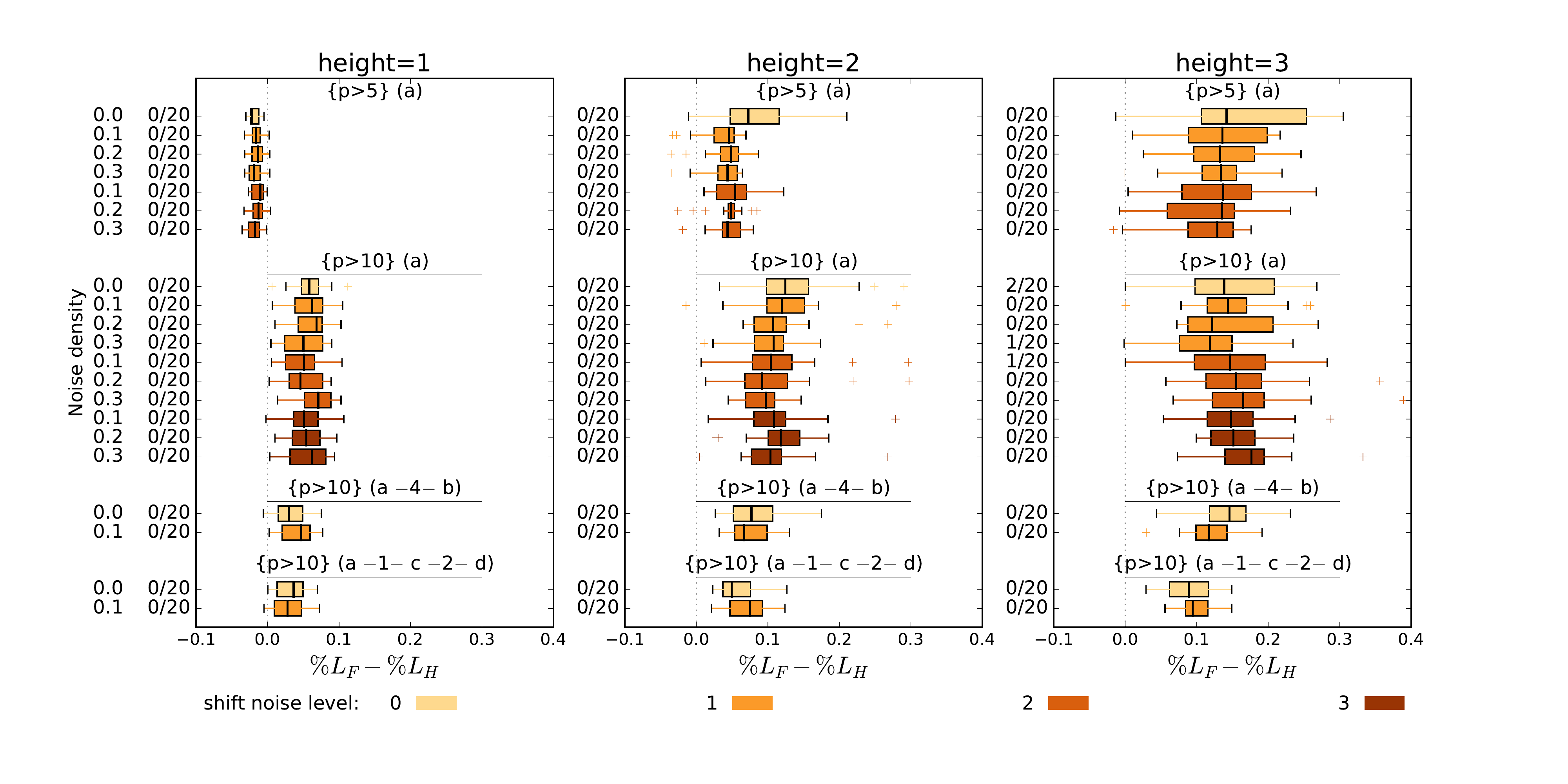}
\end{figure}

\begin{figure}[tbp] \centering
\caption{Differences in compression ratios for planted and extracted pattern collections ($\prcCl_H$ and $\prcCl_F$, respectively) on synthetic sequences perturbed by additive noise $(a, 0.5)$.}
\label{fig:synthe_3}
\includegraphics[trim=30 20 100 20,clip,width=.9\textwidth]{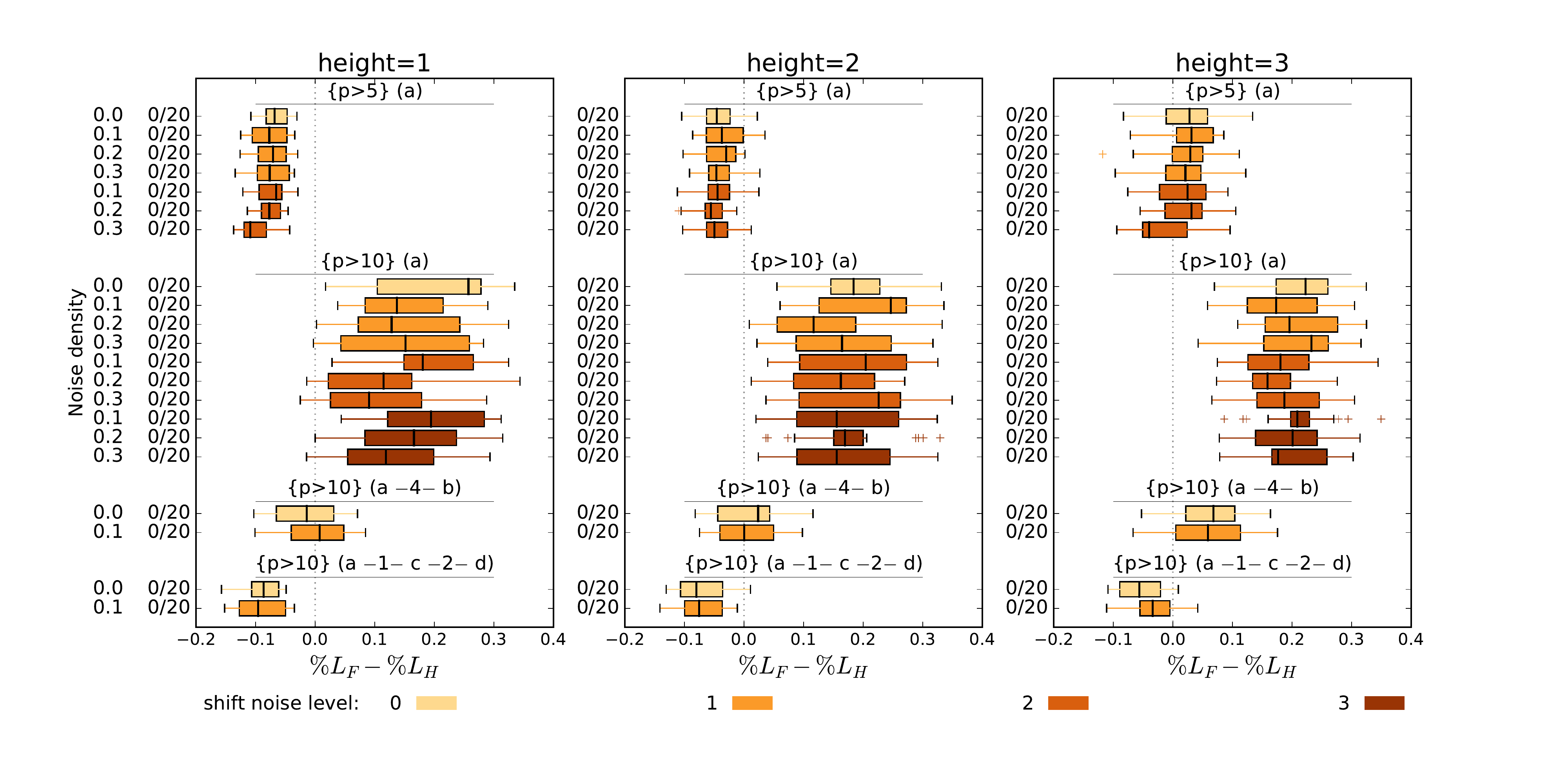}
\end{figure}

\begin{figure}[tbp] \centering
\caption{Differences in compression ratios for planted and extracted pattern collections ($\prcCl_H$ and $\prcCl_F$, respectively) on synthetic sequences containing interleaving.}
\label{fig:synthe_4}
\includegraphics[trim=30 20 100 20,clip,width=.9\textwidth]{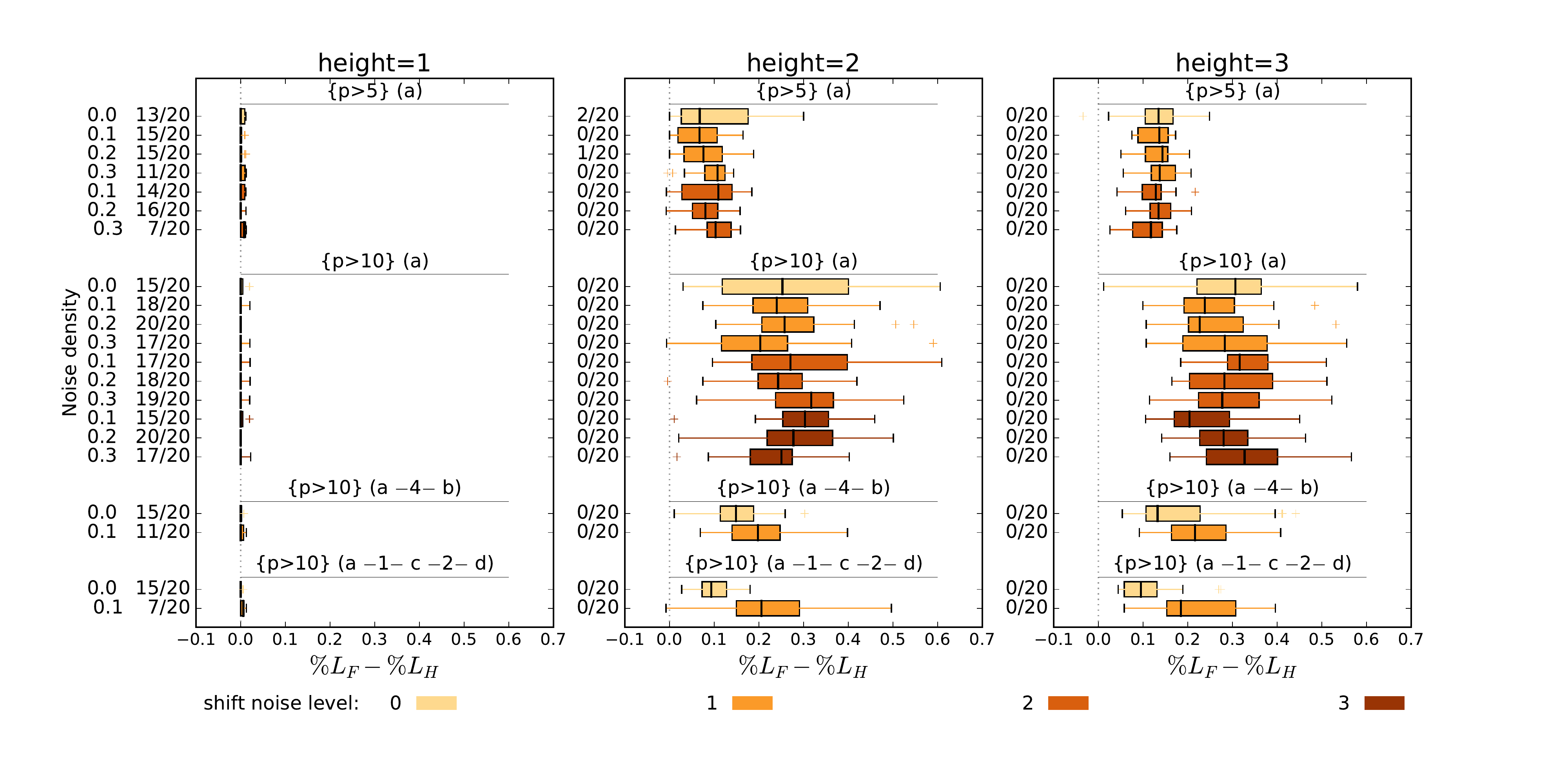}
\end{figure}

\begin{figure}[tbp] \centering
\caption{Compression ratios for planted and extracted pattern collections ($\prcCl_H$ and $\prcCl_F$, respectively) on synthetic sequences with multiple planted patterns.}
\label{fig:synthe_comb_cr}
\includegraphics[trim=0 0 0 0,clip,width=.8\textwidth]{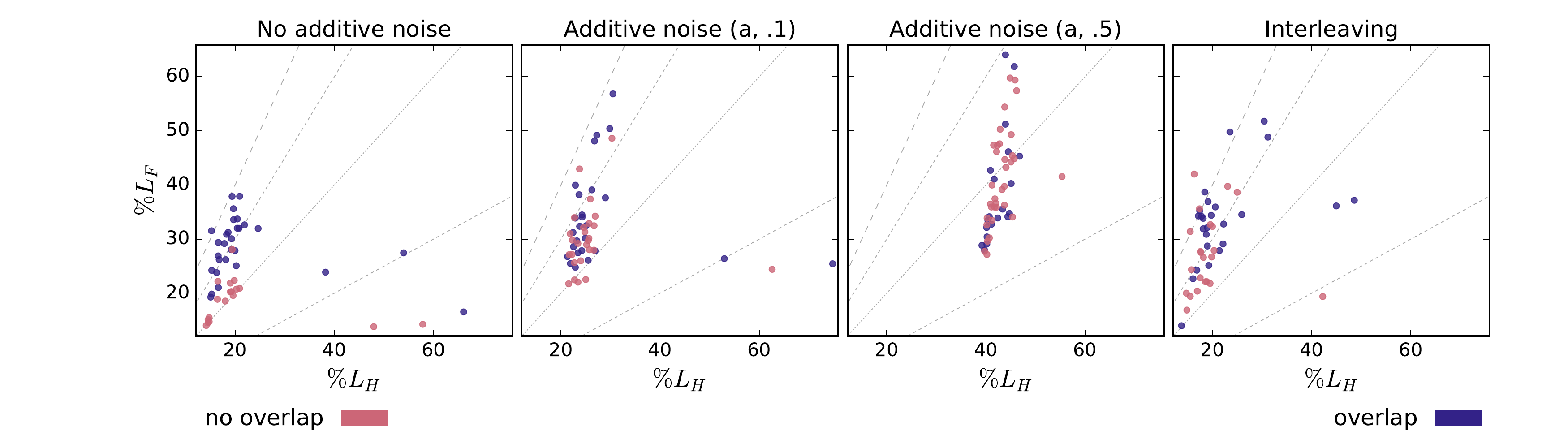}
\end{figure}

\begin{figure}[tbp] \centering
\caption{Differences in compression ratios for planted and extracted pattern collections ($\prcCl_H$ and $\prcCl_F$, respectively) on synthetic sequences with multiple planted patterns.}
\label{fig:synthe_comb}
\includegraphics[trim=0 0 80 20,clip,width=.8\textwidth]{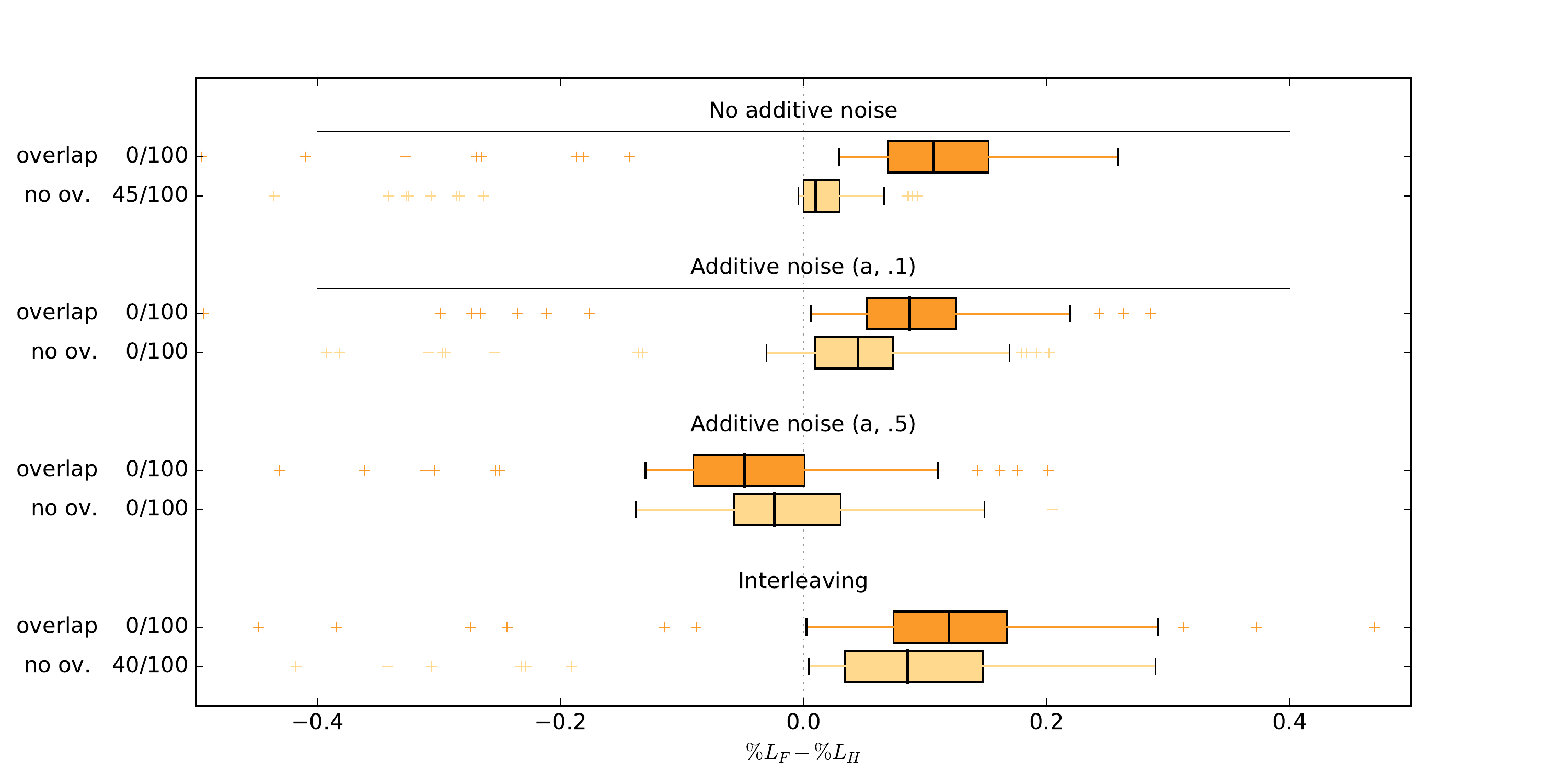}
\end{figure}

\FloatBarrier

\begin{table}
\caption{Statistics for application log trace sequences.}
\label{tab:data-stats-traces}
\centering
\begin{tabular}{@{\hspace{1ex}}l@{\hspace{4ex}}r@{\hspace{2ex}}r@{\hspace{2ex}}r@{\hspace{2.5ex}}r@{\hspace{1.5ex}}r@{\hspace{2ex}}r@{\hspace{2ex}}r@{\hspace{2.5ex}}r@{\hspace{1.5ex}}r@{\hspace{1ex}}}
\toprule
 & $\len{\seq}$ & $\tspan{\seq}$ & $\abs{\ABC}$ & \multicolumn{2}{c}{$\len{\seq[\alpha]}$} &  & $\clEmpty$ & \multicolumn{2}{c}{RT (s)} \\ 
 &  &  &  & $\omed$ & $\max$ &  &  & cycles & overall \\ 
\midrule
\dstTZap{0} & $181644$ & $181643$ & $443$ & $22$ & $36697$ &  & $4154277$ & $2094$ & $35048$ \\ 
\dstTZap{1} & $129532$ & $129531$ & $214$ & $57$ & $29376$ &  & $2849285$ & $1697$ & $32125$ \\ 
\dstBugz{0} & $16775$ & $16774$ & $91$ & $6$ & $3332$ &  & $303352$ & $112$ & $522$ \\ 
\dstBugz{1} & $15418$ & $15417$ & $61$ & $24$ & $3551$ &  & $276298$ & $116$ & $504$ \\ 
\dstSamba{} & $28751$ & $7461$ & $119$ & $44$ & $2905$ &  & $520443$ & $214$ & $2787$ \\ 
\bottomrule
\end{tabular}
\end{table}
\begin{table}
\caption{Statistics for \dstSacha{} sequences.}
\label{tab:data-stats-sacha}
\centering
\begin{tabular}{@{\hspace{1ex}}l@{\hspace{4ex}}r@{\hspace{2ex}}r@{\hspace{2ex}}r@{\hspace{2.5ex}}r@{\hspace{1.5ex}}r@{\hspace{2ex}}r@{\hspace{2ex}}r@{\hspace{2.5ex}}r@{\hspace{1.5ex}}r@{\hspace{1ex}}}
\toprule
 & $\len{\seq}$ & $\tspan{\seq}$ & $\abs{\ABC}$ & \multicolumn{2}{c}{$\len{\seq[\alpha]}$} &  & $\clEmpty$ & \multicolumn{2}{c}{RT (s)} \\ 
 &  &  &  & $\omed$ & $\max$ &  &  & cycles & overall \\ 
\midrule
\dstSachaG{1} & $72516$ & $3321680$ & $94$ & $523$ & $5531$ &  & $1987678$ & $700$ & $8734$ \\ 
\dstSachaG{15} & $65977$ & $221445$ & $141$ & $231$ & $4389$ &  & $1573140$ & $2963$ & $14377$ \\ 
\dstSachaG{30} & $58447$ & $110722$ & $141$ & $254$ & $3033$ &  & $1343757$ & $980$ & $9125$ \\ 
\dstSachaG{60} & $47880$ & $55361$ & $141$ & $154$ & $4270$ &  & $1045284$ & $598$ & $5310$ \\ 
\dstSachaG{720} & $26380$ & $4613$ & $69$ & $174$ & $3547$ &  & $450453$ & $212$ & $2287$ \\ 
\dstSachaG{1440} & $22261$ & $2306$ & $55$ & $306$ & $2287$ &  & $359005$ & $153$ & $1533$ \\ 
\dstSachaAR{\iRel} & $36258$ & $36257$ & $47$ & $523$ & $5531$ &  & $721270$ & $373$ & $22252$ \\ 
\bottomrule
\end{tabular}
\end{table}
\begin{table}
\caption{Statistics for \dstUbiAR{\iAbs} sequences.}
\label{tab:data-stats-UbiqLog-ISE-abs}
\centering
\begin{tabular}{@{\hspace{1ex}}l@{\hspace{4ex}}r@{\hspace{2ex}}r@{\hspace{2ex}}r@{\hspace{2.5ex}}r@{\hspace{1.5ex}}r@{\hspace{2ex}}r@{\hspace{2ex}}r@{\hspace{2.5ex}}r@{\hspace{1.5ex}}r@{\hspace{1ex}}}
\toprule
 & $\len{\seq}$ & $\tspan{\seq}$ & $\abs{\ABC}$ & \multicolumn{2}{c}{$\len{\seq[\alpha]}$} &  & $\clEmpty$ & \multicolumn{2}{c}{RT (s)} \\ 
 &  &  &  & $\omed$ & $\max$ &  &  & cycles & overall \\ 
\midrule
\dstUbiSAbs{25-F} & $413$ & $11391$ & $10$ & $23$ & $211$ &  & $6599$ & $1$ & $1$ \\ 
\dstUbiSAbs{10-M} & $1290$ & $21116$ & $17$ & $28$ & $194$ &  & $23110$ & $14$ & $18$ \\ 
\dstUbiSAbs{9-M} & $1483$ & $29499$ & $25$ & $26$ & $365$ &  & $27362$ & $6$ & $10$ \\ 
\dstUbiSAbs{21-F} & $5506$ & $92897$ & $85$ & $38$ & $479$ &  & $121869$ & $33$ & $51$ \\ 
\dstUbiSAbs{23-F} & $8262$ & $7389217$ & $50$ & $63$ & $699$ &  & $227374$ & $68$ & $154$ \\ 
\dstUbiSAbs{14-F} & $9682$ & $100684$ & $49$ & $47$ & $1818$ &  & $199781$ & $76$ & $150$ \\
\bottomrule
\end{tabular}
\end{table}
\begin{table}
\caption{Statistics for \dstUbiAR{\iRel} sequences.}
\label{tab:data-stats-UbiqLog-IS-rel}
\centering
\begin{tabular}{@{\hspace{1ex}}l@{\hspace{4ex}}r@{\hspace{2ex}}r@{\hspace{2ex}}r@{\hspace{2.5ex}}r@{\hspace{1.5ex}}r@{\hspace{2ex}}r@{\hspace{2ex}}r@{\hspace{2.5ex}}r@{\hspace{1.5ex}}r@{\hspace{1ex}}}
\toprule
 & $\len{\seq}$ & $\tspan{\seq}$ & $\abs{\ABC}$ & \multicolumn{2}{c}{$\len{\seq[\alpha]}$} &  & $\clEmpty$ & \multicolumn{2}{c}{RT (s)} \\ 
 &  &  &  & $\omed$ & $\max$ &  &  & cycles & overall \\ 
\midrule
\dstUbiSRel{25-F} & $372$ & $371$ & $6$ & $33$ & $211$ &  & $3896$ & $1$ & $3$ \\ 
\dstUbiSRel{10-M} & $905$ & $904$ & $7$ & $101$ & $334$ &  & $10899$ & $7$ & $15$ \\ 
\dstUbiSRel{9-M} & $973$ & $972$ & $18$ & $25$ & $365$ &  & $12693$ & $5$ & $8$ \\ 
\dstUbiSRel{21-F} & $4234$ & $4233$ & $51$ & $51$ & $806$ &  & $71351$ & $29$ & $49$ \\ 
\dstUbiSRel{23-F} & $5274$ & $5273$ & $23$ & $83$ & $991$ &  & $83287$ & $36$ & $160$ \\ 
\dstUbiSRel{14-F} & $6676$ & $6675$ & $29$ & $29$ & $2670$ &  & $104860$ & $48$ & $103$ \\
\bottomrule
\end{tabular}
\end{table}

\begin{table}
\caption{Detailed results for application log trace sequences.}
\label{tab:res-long-traces}
\centering
\begin{tabular}{@{}c@{\hspace{3ex}}r@{\hspace{1ex}}r@{\hspace{1ex}}r@{\hspace{3ex}}r@{\hspace{1ex}}r@{\hspace{1ex}}c@{\,/\,}c@{\,/\,}c@{\,/\,}c@{\hspace{3ex}}r@{\hspace{1ex}}r@{\hspace{1ex}}r@{}}
\toprule
 & $\prcCl{}$ & $\clCC$ & $\ratioClR{}$ & $\abs{\resSet}$ & $\abs{\ccycle}$ & \parbox[b][1em][b]{.5cm}{\centering $\nbS$} & \parbox[b][1em][b]{.5cm}{\centering $\nbV$} & \parbox[b][1em][b]{.5cm}{\centering $\nbH$} & \parbox[b][1em][b]{.5cm}{\centering $\nbTD$} & $\nbOTC$ & $\nbOmed$ & $\nbOmax$ \\ 
\cmidrule{2-13} \\ [-.5em]
\multicolumn{13}{c}{\dstTZap{0}} \\ 
\midrule
$\collS$ & $56.32$ & $2339741$ & 0.41 & $37048$ & $11852$ & $11852$ & $0$ & $0$ & $0$ & 0.94 & $5$ & $2325$ \\ 
$\collV$ & $55.14$ & $2290523$ & 0.40 & $35270$ & $11162$ & $10581$ & $581$ & $0$ & $0$ & 0.93 & $5$ & $2325$ \\ 
$\collH$ & $47.84$ & $1987311$ & 0.35 & $26773$ & $8371$ & $3459$ & $0$ & $4912$ & $0$ & 0.97 & $8$ & $2325$ \\ 
$\collVH$ & $47.40$ & $1969139$ & 0.34 & $26261$ & $8220$ & $3499$ & $419$ & $4302$ & $0$ & 0.97 & $8$ & $2325$ \\ 
$\collF$ & $\textbf{46.99}$ & $1952299$ & 0.34 & $25982$ & $8012$ & $3499$ & $91$ & $4154$ & $268$ & 0.96 & $8$ & $2325$ \\
\bottomrule
\end{tabular}
\end{table}
\begin{table}
\caption{Detailed results for \dstSacha{}.}
\label{tab:res-long-sacha_1}
\centering
\begin{tabular}{@{}c@{\hspace{3ex}}r@{\hspace{1ex}}r@{\hspace{1ex}}r@{\hspace{3ex}}r@{\hspace{1ex}}r@{\hspace{1ex}}c@{\,/\,}c@{\,/\,}c@{\,/\,}c@{\hspace{3ex}}r@{\hspace{1ex}}r@{\hspace{1ex}}r@{}}
\toprule
 & $\prcCl{}$ & $\clCC$ & $\ratioClR{}$ & $\abs{\resSet}$ & $\abs{\ccycle}$ & \parbox[b][1em][b]{.5cm}{\centering $\nbS$} & \parbox[b][1em][b]{.5cm}{\centering $\nbV$} & \parbox[b][1em][b]{.5cm}{\centering $\nbH$} & \parbox[b][1em][b]{.5cm}{\centering $\nbTD$} & $\nbOTC$ & $\nbOmed$ & $\nbOmax$ \\ 
\cmidrule{2-13} \\ [-.5em]
\multicolumn{13}{c}{\dstSachaG{1}} \\ 
\midrule
$\collS$ & $86.24$ & $1714241$ & 0.44 & $27295$ & $13134$ & $13134$ & $0$ & $0$ & $0$ & 0.37 & $3$ & $10$ \\ 
$\collV$ & $86.24$ & $1714241$ & 0.44 & $27295$ & $13134$ & $13134$ & $0$ & $0$ & $0$ & 0.37 & $3$ & $10$ \\ 
$\collH$ & $\textbf{84.19}$ & $1673502$ & 0.44 & $26754$ & $11327$ & $9321$ & $0$ & $2006$ & $0$ & 0.47 & $3$ & $28$ \\ 
$\collVH$ & $\textbf{84.19}$ & $1673502$ & 0.44 & $26754$ & $11327$ & $9321$ & $0$ & $2006$ & $0$ & 0.47 & $3$ & $28$ \\ 
$\collF$ & $\textbf{84.19}$ & $1673502$ & 0.44 & $26754$ & $11327$ & $9321$ & $0$ & $2006$ & $0$ & 0.47 & $3$ & $28$ \\
\bottomrule
\end{tabular}
\end{table}
\begin{table}
\caption{Detailed results for \dstUbiAR{\iAbs} sequences (1/5).}
\label{tab:res-long-UbiqLog-ISE-abs_1}
\centering
\begin{tabular}{@{}c@{\hspace{3ex}}r@{\hspace{1ex}}r@{\hspace{1ex}}r@{\hspace{3ex}}r@{\hspace{1ex}}r@{\hspace{1ex}}c@{\,/\,}c@{\,/\,}c@{\,/\,}c@{\hspace{3ex}}r@{\hspace{1ex}}r@{\hspace{1ex}}r@{}}
\toprule
 & $\prcCl{}$ & $\clCC$ & $\ratioClR{}$ & $\abs{\resSet}$ & $\abs{\ccycle}$ & \parbox[b][1em][b]{.5cm}{\centering $\nbS$} & \parbox[b][1em][b]{.5cm}{\centering $\nbV$} & \parbox[b][1em][b]{.5cm}{\centering $\nbH$} & \parbox[b][1em][b]{.5cm}{\centering $\nbTD$} & $\nbOTC$ & $\nbOmed$ & $\nbOmax$ \\ 
\cmidrule{2-13} \\ [-.5em]
\multicolumn{13}{c}{\dstUbiSAbs{25-F}} \\ 
\midrule
$\collS$ & $85.52$ & $5643$ & 0.57 & $193$ & $41$ & $41$ & $0$ & $0$ & $0$ & 0.88 & $5$ & $17$ \\ 
$\collV$ & $85.52$ & $5643$ & 0.57 & $193$ & $41$ & $41$ & $0$ & $0$ & $0$ & 0.88 & $5$ & $17$ \\ 
$\collH$ & $\textbf{84.33}$ & $5564$ & 0.58 & $194$ & $36$ & $31$ & $0$ & $5$ & $0$ & 0.86 & $5$ & $17$ \\ 
$\collVH$ & $\textbf{84.33}$ & $5564$ & 0.58 & $194$ & $36$ & $31$ & $0$ & $5$ & $0$ & 0.86 & $5$ & $17$ \\ 
$\collF$ & $\textbf{84.33}$ & $5564$ & 0.58 & $194$ & $36$ & $31$ & $0$ & $5$ & $0$ & 0.86 & $5$ & $17$ \\
\bottomrule
\end{tabular}
\end{table}
\begin{table}
\caption{Detailed results for \dstUbiAR{\iAbs} sequences (2/5).}
\label{tab:res-long-UbiqLog-ISE-abs_2}
\centering
\begin{tabular}{@{}c@{\hspace{3ex}}r@{\hspace{1ex}}r@{\hspace{1ex}}r@{\hspace{3ex}}r@{\hspace{1ex}}r@{\hspace{1ex}}c@{\,/\,}c@{\,/\,}c@{\,/\,}c@{\hspace{3ex}}r@{\hspace{1ex}}r@{\hspace{1ex}}r@{}}
\toprule
 & $\prcCl{}$ & $\clCC$ & $\ratioClR{}$ & $\abs{\resSet}$ & $\abs{\ccycle}$ & \parbox[b][1em][b]{.5cm}{\centering $\nbS$} & \parbox[b][1em][b]{.5cm}{\centering $\nbV$} & \parbox[b][1em][b]{.5cm}{\centering $\nbH$} & \parbox[b][1em][b]{.5cm}{\centering $\nbTD$} & $\nbOTC$ & $\nbOmed$ & $\nbOmax$ \\ 
\cmidrule{2-13} \\ [-.5em]
\multicolumn{13}{c}{\dstUbiSAbs{30-F}} \\ 
\midrule
$\collS$ & $73.52$ & $172191$ & 0.40 & $3305$ & $1497$ & $1497$ & $0$ & $0$ & $0$ & 0.75 & $4$ & $54$ \\ 
$\collV$ & $73.52$ & $172185$ & 0.40 & $3307$ & $1492$ & $1487$ & $5$ & $0$ & $0$ & 0.75 & $4$ & $54$ \\ 
$\collH$ & $\textbf{54.91}$ & $128612$ & 0.36 & $2101$ & $796$ & $261$ & $0$ & $535$ & $0$ & 0.79 & $12$ & $138$ \\ 
$\collVH$ & $\textbf{54.91}$ & $128612$ & 0.36 & $2101$ & $796$ & $261$ & $0$ & $535$ & $0$ & 0.79 & $12$ & $138$ \\ 
$\collF$ & $\textbf{54.91}$ & $128612$ & 0.36 & $2101$ & $796$ & $261$ & $0$ & $535$ & $0$ & 0.79 & $12$ & $138$ \\
\bottomrule
\end{tabular}
\end{table}
\begin{table}
\caption{Detailed results for \dstUbiAR{\iAbs} sequences (3/5).}
\label{tab:res-long-UbiqLog-ISE-abs_3}
\centering
\begin{tabular}{@{}c@{\hspace{3ex}}r@{\hspace{1ex}}r@{\hspace{1ex}}r@{\hspace{3ex}}r@{\hspace{1ex}}r@{\hspace{1ex}}c@{\,/\,}c@{\,/\,}c@{\,/\,}c@{\hspace{3ex}}r@{\hspace{1ex}}r@{\hspace{1ex}}r@{}}
\toprule
 & $\prcCl{}$ & $\clCC$ & $\ratioClR{}$ & $\abs{\resSet}$ & $\abs{\ccycle}$ & \parbox[b][1em][b]{.5cm}{\centering $\nbS$} & \parbox[b][1em][b]{.5cm}{\centering $\nbV$} & \parbox[b][1em][b]{.5cm}{\centering $\nbH$} & \parbox[b][1em][b]{.5cm}{\centering $\nbTD$} & $\nbOTC$ & $\nbOmed$ & $\nbOmax$ \\ 
\cmidrule{2-13} \\ [-.5em]
\multicolumn{13}{c}{\dstUbiSAbs{2-F}} \\ 
\midrule
$\collS$ & $66.44$ & $322518$ & 0.32 & $4061$ & $2325$ & $2325$ & $0$ & $0$ & $0$ & 0.80 & $5$ & $60$ \\ 
$\collV$ & $66.40$ & $322299$ & 0.32 & $4054$ & $2327$ & $2321$ & $6$ & $0$ & $0$ & 0.80 & $5$ & $60$ \\ 
$\collH$ & $\textbf{46.42}$ & $225305$ & 0.34 & $2928$ & $1135$ & $458$ & $0$ & $677$ & $0$ & 0.81 & $8$ & $315$ \\ 
$\collVH$ & $\textbf{46.42}$ & $225305$ & 0.34 & $2928$ & $1135$ & $458$ & $0$ & $677$ & $0$ & 0.81 & $8$ & $315$ \\ 
$\collF$ & $46.42$ & $225345$ & 0.34 & $2931$ & $1134$ & $457$ & $0$ & $676$ & $1$ & 0.81 & $8$ & $315$ \\
\bottomrule
\end{tabular}
\end{table}
\begin{table}
\caption{Detailed results for \dstUbiAR{\iAbs} sequences (4/5).}
\label{tab:res-long-UbiqLog-ISE-abs_4}
\centering
\begin{tabular}{@{}c@{\hspace{3ex}}r@{\hspace{1ex}}r@{\hspace{1ex}}r@{\hspace{3ex}}r@{\hspace{1ex}}r@{\hspace{1ex}}c@{\,/\,}c@{\,/\,}c@{\,/\,}c@{\hspace{3ex}}r@{\hspace{1ex}}r@{\hspace{1ex}}r@{}}
\toprule
 & $\prcCl{}$ & $\clCC$ & $\ratioClR{}$ & $\abs{\resSet}$ & $\abs{\ccycle}$ & \parbox[b][1em][b]{.5cm}{\centering $\nbS$} & \parbox[b][1em][b]{.5cm}{\centering $\nbV$} & \parbox[b][1em][b]{.5cm}{\centering $\nbH$} & \parbox[b][1em][b]{.5cm}{\centering $\nbTD$} & $\nbOTC$ & $\nbOmed$ & $\nbOmax$ \\ 
\cmidrule{2-13} \\ [-.5em]
\multicolumn{13}{c}{\dstUbiSAbs{11-F}} \\ 
\midrule
$\collS$ & $66.37$ & $508461$ & 0.31 & $7047$ & $4382$ & $4382$ & $0$ & $0$ & $0$ & 0.82 & $4$ & $96$ \\ 
$\collV$ & $66.41$ & $508739$ & 0.31 & $7095$ & $4379$ & $4373$ & $6$ & $0$ & $0$ & 0.82 & $4$ & $96$ \\ 
$\collH$ & $\textbf{51.53}$ & $394797$ & 0.29 & $5016$ & $2615$ & $1279$ & $0$ & $1336$ & $0$ & 0.83 & $6$ & $294$ \\ 
$\collVH$ & $51.56$ & $395014$ & 0.29 & $5057$ & $2610$ & $1270$ & $1$ & $1339$ & $0$ & 0.83 & $6$ & $294$ \\ 
$\collF$ & $51.58$ & $395154$ & 0.29 & $5062$ & $2613$ & $1282$ & $0$ & $1329$ & $2$ & 0.82 & $6$ & $294$ \\
\bottomrule
\end{tabular}
\end{table}
\begin{table}
\caption{Detailed results for \dstUbiAR{\iAbs} sequences (5/5).}
\label{tab:res-long-UbiqLog-ISE-abs_5}
\centering
\begin{tabular}{@{}c@{\hspace{3ex}}r@{\hspace{1ex}}r@{\hspace{1ex}}r@{\hspace{3ex}}r@{\hspace{1ex}}r@{\hspace{1ex}}c@{\,/\,}c@{\,/\,}c@{\,/\,}c@{\hspace{3ex}}r@{\hspace{1ex}}r@{\hspace{1ex}}r@{}}
\toprule
 & $\prcCl{}$ & $\clCC$ & $\ratioClR{}$ & $\abs{\resSet}$ & $\abs{\ccycle}$ & \parbox[b][1em][b]{.5cm}{\centering $\nbS$} & \parbox[b][1em][b]{.5cm}{\centering $\nbV$} & \parbox[b][1em][b]{.5cm}{\centering $\nbH$} & \parbox[b][1em][b]{.5cm}{\centering $\nbTD$} & $\nbOTC$ & $\nbOmed$ & $\nbOmax$ \\ 
\cmidrule{2-13} \\ [-.5em]
\multicolumn{13}{c}{\dstUbiSAbs{33-F}} \\ 
\midrule
$\collS$ & $59.41$ & $1096442$ & 0.27 & $12832$ & $7861$ & $7861$ & $0$ & $0$ & $0$ & 0.85 & $5$ & $94$ \\ 
$\collV$ & $59.41$ & $1096403$ & 0.27 & $12832$ & $7862$ & $7861$ & $1$ & $0$ & $0$ & 0.85 & $5$ & $94$ \\ 
$\collH$ & $\textbf{37.54}$ & $692736$ & 0.32 & $9270$ & $3213$ & $1683$ & $0$ & $1530$ & $0$ & 0.81 & $6$ & $581$ \\ 
$\collVH$ & $\textbf{37.54}$ & $692736$ & 0.32 & $9270$ & $3213$ & $1683$ & $0$ & $1530$ & $0$ & 0.81 & $6$ & $581$ \\ 
$\collF$ & $\textbf{37.54}$ & $692736$ & 0.32 & $9270$ & $3213$ & $1683$ & $0$ & $1530$ & $0$ & 0.81 & $6$ & $581$ \\
\bottomrule
\end{tabular}
\end{table}
\begin{table}
\caption{Detailed results for \dstUbiAR{\iRel} sequences (1/5).}
\label{tab:res-long-UbiqLog-IS-rel_1}
\centering
\begin{tabular}{@{}c@{\hspace{3ex}}r@{\hspace{1ex}}r@{\hspace{1ex}}r@{\hspace{3ex}}r@{\hspace{1ex}}r@{\hspace{1ex}}c@{\,/\,}c@{\,/\,}c@{\,/\,}c@{\hspace{3ex}}r@{\hspace{1ex}}r@{\hspace{1ex}}r@{}}
\toprule
 & $\prcCl{}$ & $\clCC$ & $\ratioClR{}$ & $\abs{\resSet}$ & $\abs{\ccycle}$ & \parbox[b][1em][b]{.5cm}{\centering $\nbS$} & \parbox[b][1em][b]{.5cm}{\centering $\nbV$} & \parbox[b][1em][b]{.5cm}{\centering $\nbH$} & \parbox[b][1em][b]{.5cm}{\centering $\nbTD$} & $\nbOTC$ & $\nbOmed$ & $\nbOmax$ \\ 
\cmidrule{2-13} \\ [-.5em]
\multicolumn{13}{c}{\dstUbiSRel{25-F}} \\ 
\midrule
$\collS$ & $\textbf{46.22}$ & $1800$ & 0.24 & $34$ & $9$ & $9$ & $0$ & $0$ & $0$ & 1.00 & $12$ & $211$ \\ 
$\collV$ & $\textbf{46.22}$ & $1800$ & 0.24 & $34$ & $9$ & $9$ & $0$ & $0$ & $0$ & 1.00 & $12$ & $211$ \\ 
$\collH$ & $\textbf{46.22}$ & $1800$ & 0.24 & $34$ & $9$ & $9$ & $0$ & $0$ & $0$ & 1.00 & $12$ & $211$ \\ 
$\collVH$ & $\textbf{46.22}$ & $1800$ & 0.24 & $34$ & $9$ & $9$ & $0$ & $0$ & $0$ & 1.00 & $12$ & $211$ \\ 
$\collF$ & $\textbf{46.22}$ & $1800$ & 0.24 & $34$ & $9$ & $9$ & $0$ & $0$ & $0$ & 1.00 & $12$ & $211$ \\
\bottomrule
\end{tabular}
\end{table}
\begin{table}
\caption{Detailed results for \dstUbiAR{\iRel} sequences (2/5).}
\label{tab:res-long-UbiqLog-IS-rel_2}
\centering
\begin{tabular}{@{}c@{\hspace{3ex}}r@{\hspace{1ex}}r@{\hspace{1ex}}r@{\hspace{3ex}}r@{\hspace{1ex}}r@{\hspace{1ex}}c@{\,/\,}c@{\,/\,}c@{\,/\,}c@{\hspace{3ex}}r@{\hspace{1ex}}r@{\hspace{1ex}}r@{}}
\toprule
 & $\prcCl{}$ & $\clCC$ & $\ratioClR{}$ & $\abs{\resSet}$ & $\abs{\ccycle}$ & \parbox[b][1em][b]{.5cm}{\centering $\nbS$} & \parbox[b][1em][b]{.5cm}{\centering $\nbV$} & \parbox[b][1em][b]{.5cm}{\centering $\nbH$} & \parbox[b][1em][b]{.5cm}{\centering $\nbTD$} & $\nbOTC$ & $\nbOmed$ & $\nbOmax$ \\ 
\cmidrule{2-13} \\ [-.5em]
\multicolumn{13}{c}{\dstUbiSRel{30-F}} \\ 
\midrule
$\collS$ & $45.39$ & $68909$ & 0.37 & $1274$ & $332$ & $332$ & $0$ & $0$ & $0$ & 0.92 & $5$ & $1759$ \\ 
$\collV$ & $45.31$ & $68792$ & 0.36 & $1264$ & $332$ & $331$ & $1$ & $0$ & $0$ & 0.92 & $5$ & $1759$ \\ 
$\collH$ & $44.48$ & $67520$ & 0.35 & $1189$ & $305$ & $238$ & $0$ & $67$ & $0$ & 0.93 & $6$ & $1759$ \\ 
$\collVH$ & $\textbf{44.45}$ & $67483$ & 0.35 & $1194$ & $302$ & $237$ & $1$ & $64$ & $0$ & 0.93 & $6$ & $1759$ \\ 
$\collF$ & $\textbf{44.45}$ & $67483$ & 0.35 & $1194$ & $302$ & $237$ & $1$ & $64$ & $0$ & 0.93 & $6$ & $1759$ \\
\bottomrule
\end{tabular}
\end{table}
\begin{table}
\caption{Detailed results for \dstUbiAR{\iRel} sequences (3/5).}
\label{tab:res-long-UbiqLog-IS-rel_3}
\centering
\begin{tabular}{@{}c@{\hspace{3ex}}r@{\hspace{1ex}}r@{\hspace{1ex}}r@{\hspace{3ex}}r@{\hspace{1ex}}r@{\hspace{1ex}}c@{\,/\,}c@{\,/\,}c@{\,/\,}c@{\hspace{3ex}}r@{\hspace{1ex}}r@{\hspace{1ex}}r@{}}
\toprule
 & $\prcCl{}$ & $\clCC$ & $\ratioClR{}$ & $\abs{\resSet}$ & $\abs{\ccycle}$ & \parbox[b][1em][b]{.5cm}{\centering $\nbS$} & \parbox[b][1em][b]{.5cm}{\centering $\nbV$} & \parbox[b][1em][b]{.5cm}{\centering $\nbH$} & \parbox[b][1em][b]{.5cm}{\centering $\nbTD$} & $\nbOTC$ & $\nbOmed$ & $\nbOmax$ \\ 
\cmidrule{2-13} \\ [-.5em]
\multicolumn{13}{c}{\dstUbiSRel{2-F}} \\ 
\midrule
$\collS$ & $44.31$ & $110476$ & 0.33 & $1749$ & $584$ & $584$ & $0$ & $0$ & $0$ & 0.76 & $5$ & $1891$ \\ 
$\collV$ & $44.31$ & $110476$ & 0.33 & $1749$ & $584$ & $584$ & $0$ & $0$ & $0$ & 0.76 & $5$ & $1891$ \\ 
$\collH$ & $\textbf{41.61}$ & $103753$ & 0.32 & $1616$ & $552$ & $469$ & $0$ & $83$ & $0$ & 0.82 & $5$ & $5673$ \\ 
$\collVH$ & $\textbf{41.61}$ & $103753$ & 0.32 & $1616$ & $552$ & $469$ & $0$ & $83$ & $0$ & 0.82 & $5$ & $5673$ \\ 
$\collF$ & $\textbf{41.61}$ & $103753$ & 0.32 & $1616$ & $552$ & $469$ & $0$ & $83$ & $0$ & 0.82 & $5$ & $5673$ \\
\bottomrule
\end{tabular}
\end{table}
\begin{table}
\caption{Detailed results for \dstUbiAR{\iRel} sequences (4/5).}
\label{tab:res-long-UbiqLog-IS-rel_4}
\centering
\begin{tabular}{@{}c@{\hspace{3ex}}r@{\hspace{1ex}}r@{\hspace{1ex}}r@{\hspace{3ex}}r@{\hspace{1ex}}r@{\hspace{1ex}}c@{\,/\,}c@{\,/\,}c@{\,/\,}c@{\hspace{3ex}}r@{\hspace{1ex}}r@{\hspace{1ex}}r@{}}
\toprule
 & $\prcCl{}$ & $\clCC$ & $\ratioClR{}$ & $\abs{\resSet}$ & $\abs{\ccycle}$ & \parbox[b][1em][b]{.5cm}{\centering $\nbS$} & \parbox[b][1em][b]{.5cm}{\centering $\nbV$} & \parbox[b][1em][b]{.5cm}{\centering $\nbH$} & \parbox[b][1em][b]{.5cm}{\centering $\nbTD$} & $\nbOTC$ & $\nbOmed$ & $\nbOmax$ \\ 
\cmidrule{2-13} \\ [-.5em]
\multicolumn{13}{c}{\dstUbiSRel{31-F}} \\ 
\midrule
$\collS$ & $28.01$ & $128149$ & 0.29 & $1661$ & $439$ & $439$ & $0$ & $0$ & $0$ & 0.98 & $7$ & $5196$ \\ 
$\collV$ & $28.01$ & $128149$ & 0.29 & $1661$ & $439$ & $439$ & $0$ & $0$ & $0$ & 0.98 & $7$ & $5196$ \\ 
$\collH$ & $\textbf{27.52}$ & $125919$ & 0.29 & $1597$ & $433$ & $396$ & $0$ & $37$ & $0$ & 0.98 & $8$ & $10392$ \\ 
$\collVH$ & $\textbf{27.52}$ & $125919$ & 0.29 & $1597$ & $433$ & $396$ & $0$ & $37$ & $0$ & 0.98 & $8$ & $10392$ \\ 
$\collF$ & $\textbf{27.52}$ & $125919$ & 0.29 & $1597$ & $433$ & $396$ & $0$ & $37$ & $0$ & 0.98 & $8$ & $10392$ \\
\bottomrule
\end{tabular}
\end{table}
\begin{table}
\caption{Detailed results for \dstUbiAR{\iRel} sequences (5/5).}
\label{tab:res-long-UbiqLog-IS-rel_5}
\centering
\begin{tabular}{@{}c@{\hspace{3ex}}r@{\hspace{1ex}}r@{\hspace{1ex}}r@{\hspace{3ex}}r@{\hspace{1ex}}r@{\hspace{1ex}}c@{\,/\,}c@{\,/\,}c@{\,/\,}c@{\hspace{3ex}}r@{\hspace{1ex}}r@{\hspace{1ex}}r@{}}
\toprule
 & $\prcCl{}$ & $\clCC$ & $\ratioClR{}$ & $\abs{\resSet}$ & $\abs{\ccycle}$ & \parbox[b][1em][b]{.5cm}{\centering $\nbS$} & \parbox[b][1em][b]{.5cm}{\centering $\nbV$} & \parbox[b][1em][b]{.5cm}{\centering $\nbH$} & \parbox[b][1em][b]{.5cm}{\centering $\nbTD$} & $\nbOTC$ & $\nbOmed$ & $\nbOmax$ \\ 
\cmidrule{2-13} \\ [-.5em]
\multicolumn{13}{c}{\dstUbiSRel{33-F}} \\ 
\midrule
$\collS$ & $33.75$ & $427496$ & 0.28 & $4911$ & $1801$ & $1801$ & $0$ & $0$ & $0$ & 0.72 & $4$ & $6620$ \\ 
$\collV$ & $33.74$ & $427439$ & 0.28 & $4912$ & $1799$ & $1798$ & $1$ & $0$ & $0$ & 0.72 & $4$ & $6620$ \\ 
$\collH$ & $31.40$ & $397794$ & 0.29 & $4720$ & $1685$ & $1504$ & $0$ & $181$ & $0$ & 0.78 & $5$ & $26476$ \\ 
$\collVH$ & $\textbf{31.39}$ & $397687$ & 0.29 & $4715$ & $1683$ & $1499$ & $1$ & $183$ & $0$ & 0.78 & $5$ & $26476$ \\ 
$\collF$ & $\textbf{31.39}$ & $397687$ & 0.29 & $4715$ & $1683$ & $1499$ & $1$ & $183$ & $0$ & 0.78 & $5$ & $26476$ \\
\bottomrule
\end{tabular}
\end{table}

\FloatBarrier

\begin{figure}
\centering
\caption{Compression ratios for \dstTZapX{}, \dstBugzX{} and \dstSamba{} sequences.}
\label{fig:xps-prcCl-other}
\includegraphics[trim=0 0 40 30,clip,width=.7\textwidth]{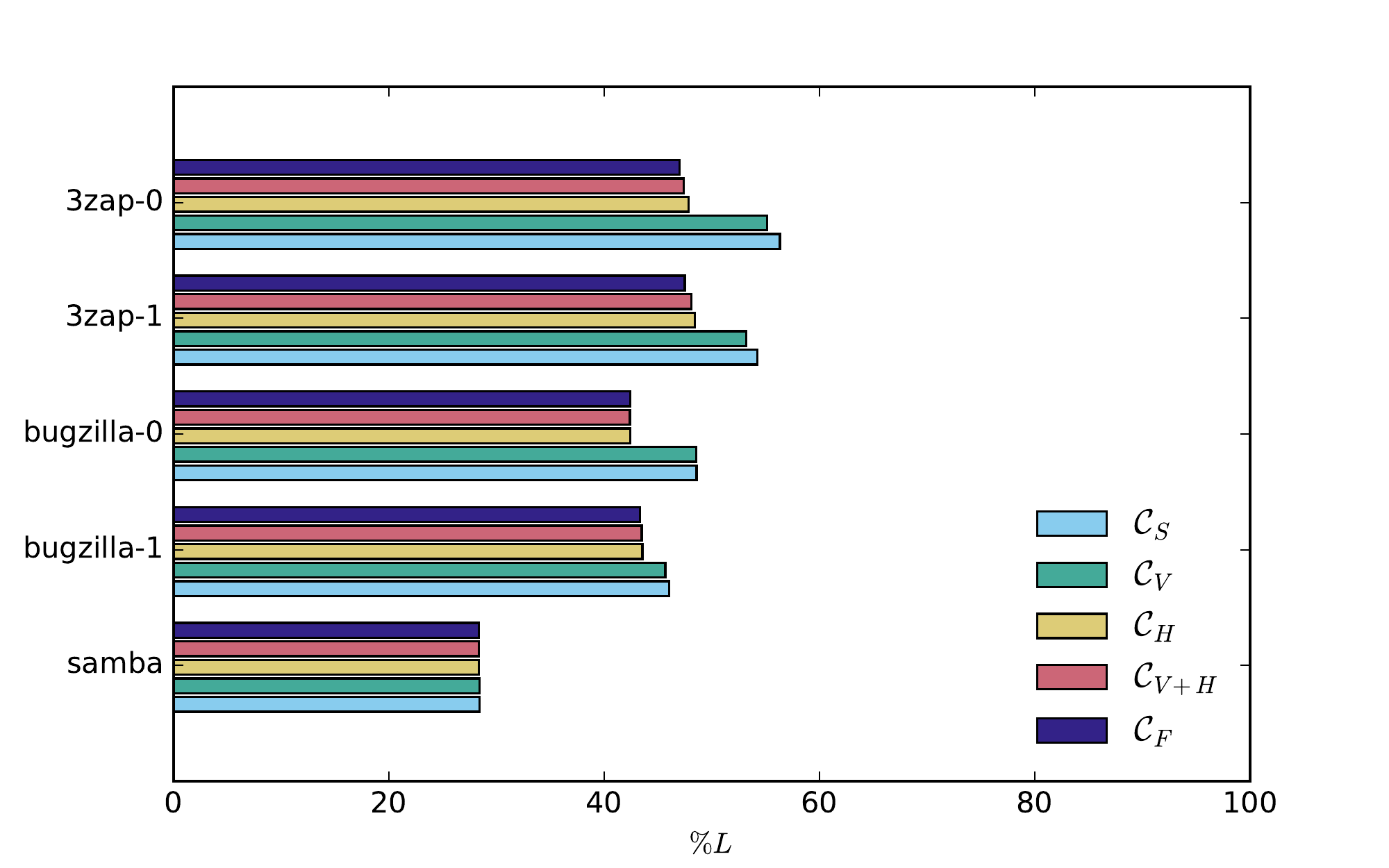}
\end{figure}

\begin{figure}
\centering
\caption{Compression ratios for \dstSacha{} sequences with various time granularities.}
\label{fig:xps-prcCl-sacha}
\includegraphics[trim=20 0 40 0,clip,width=.8\textwidth]{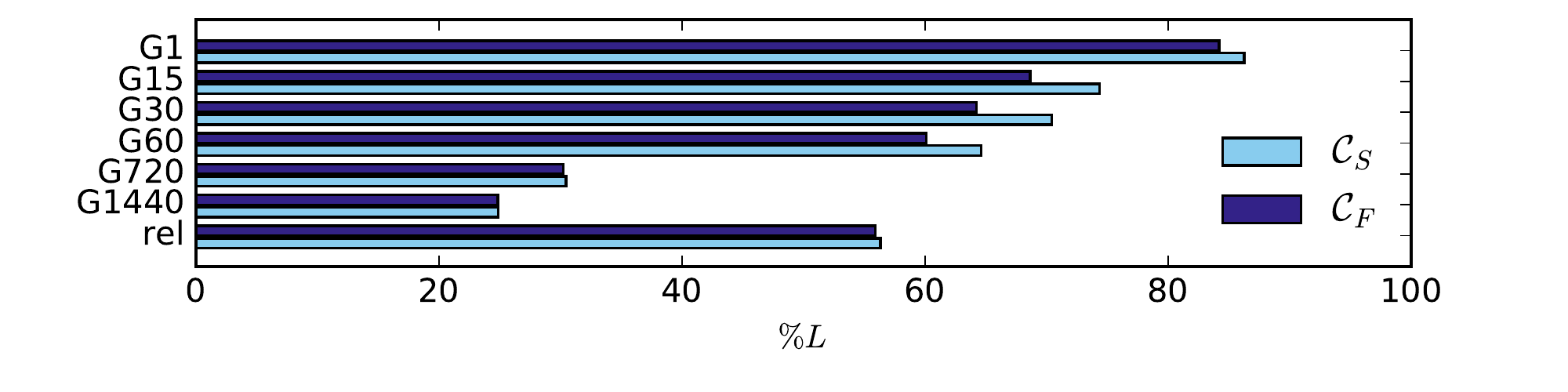}
\end{figure}

\begin{figure}
\centering
\caption{Compression ratios for the sequences from the \dstUbiAR{\iAbs} dataset.}
\label{fig:xps-prcCl-ubi-abs}
\includegraphics[trim=30 0 40 30,clip,width=.8\textwidth]{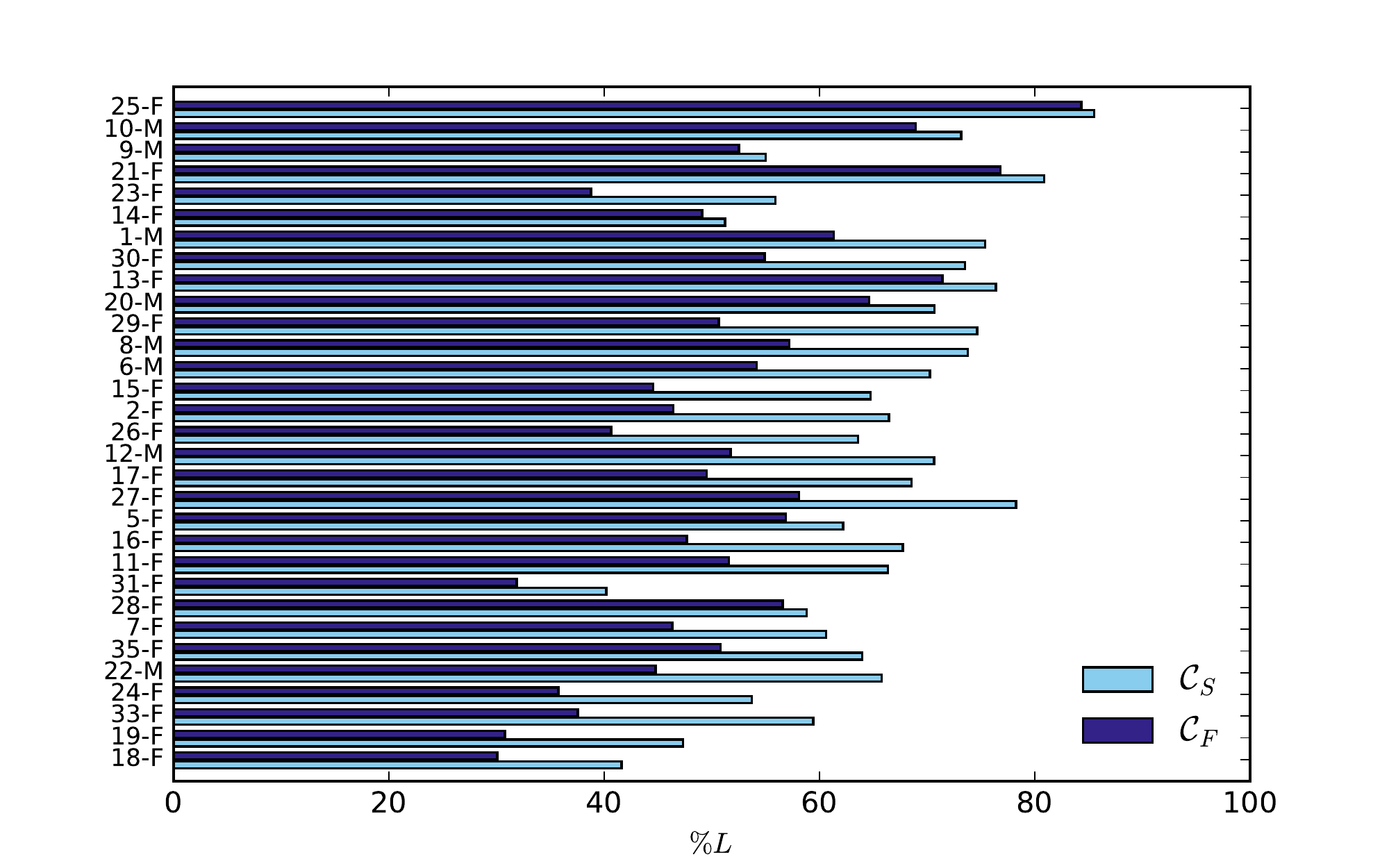}
\end{figure}

\begin{figure}
\centering
\caption{Compression ratios for the sequences from the \dstUbiAR{\iRel} dataset.}
\label{fig:xps-prcCl-ubi-rel}
\includegraphics[trim=30 0 40 30,clip,width=.8\textwidth]{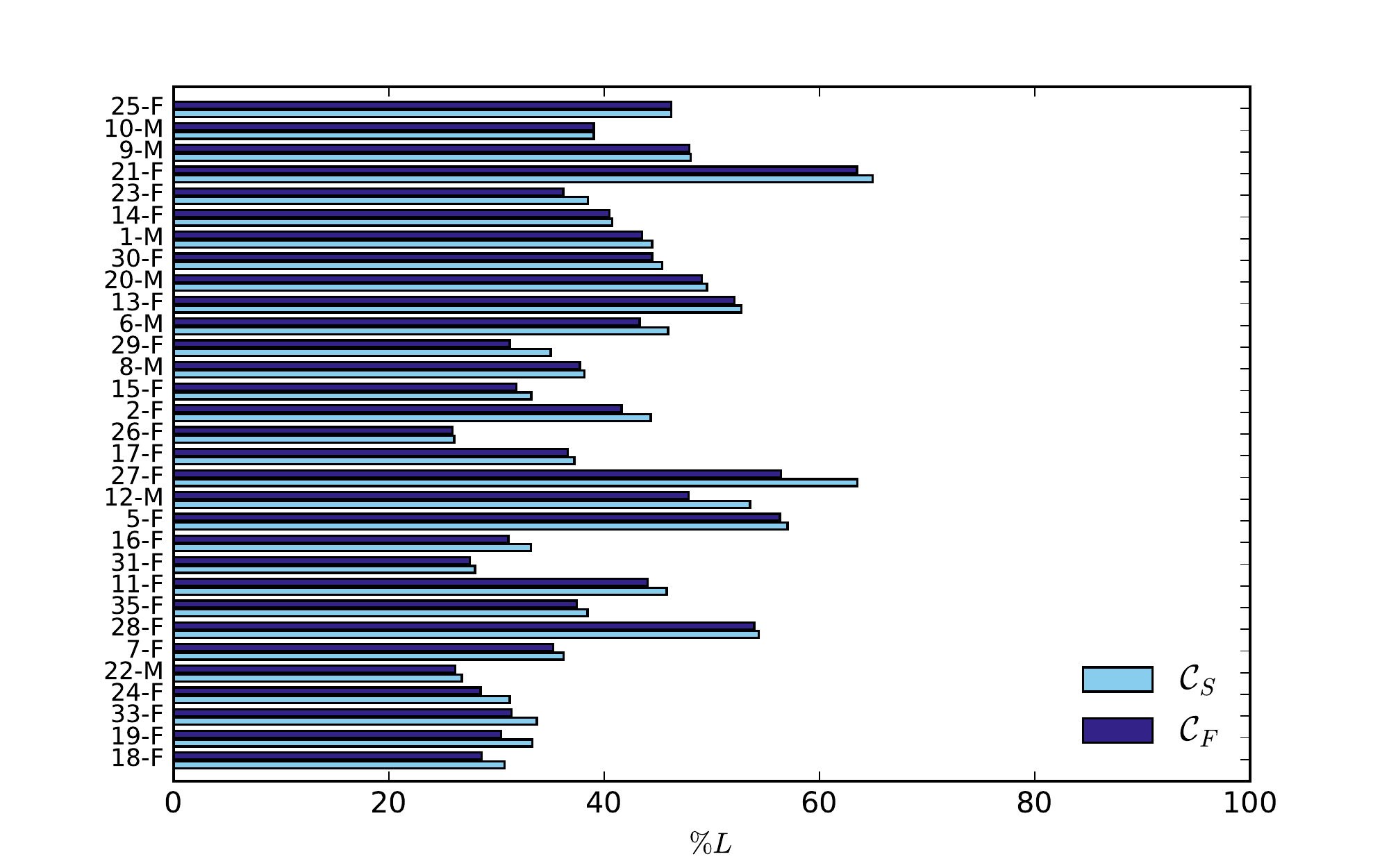}
\end{figure}

\begin{table}
\caption{Example patterns from \dstSacha{} sequences with different time granularities.}
\label{tab:res-ex-sacha}
\centering
{
\begin{tabular}{@{\hspace{1ex}}c@{\hspace{2ex}}l@{\hspace{2ex}}p{.68\textwidth}@{\hspace{3ex}}r@{\hspace{1ex}}r@{\hspace{1ex}}}
\toprule
 & $\Cto$ & \multicolumn{2}{l}{$\Ptree$ \hfill $\sum\abs{\Csc}$} & $\abs{\opoccs}$ \\
\cmidrule{2-5} \\ [-.5em]
\multicolumn{5}{c}{ \dstSachaG{1}} \\ 
\midrule
a) & 2017-09-10 12:09 & \BinfoRPT{7}{\SI{1}{\day}} \Bstart{}\activityStart{Sleep} \BinfoDT{\SI{6}{\hour}\,36} \activityEnd{Sleep} & $28$ & $28$ \\
 & & \hspace*{10ex} \BinfoDT{\SI{0}{\minute}} \activityStart{Childcare} \BinfoDT{\SI{17}{\hour}\,24} \activityEnd{Childcare}\Bend{} & & \\ [.2em]
b) & 2011-12-12 17:07 & \BinfoRPT{3}{\SI{1}{\day}\,\SI{2}{\minute}} \Bstart{}\activityEnd{Work} \BinfoDT{\SI{0}{\minute}} \activityStart{Walk} \BinfoDT{\SI{6}{\minute}} \activityEnd{Walk} & $13$ & $21$ \\
 & & \hspace*{10ex}  \BinfoDT{\SI{0}{\minute}} \activityStart{Subway} \BinfoDT{\SI{17}{\minute}} \activityEnd{Subway} & & \\
 & & \hspace*{10ex} \BinfoDT{\SI{0}{\minute}} \activityStart{Walk} \BinfoDT{\SI{11}{\minute}} \activityEnd{Walk}\Bend{} & & \\ [.2em]
c) & 2012-03-06 07:40 & \BinfoRPT{4}{\SI{1}{\day}} \Bstart{}\activityStart{Subway} \BinfoDT{\SI{0}{\minute}} \activityEnd{Routines} \BinfoDT{\SI{50}{\minute}} \activityStart{Consulting-E}\Bend{} & $7$ & $12$ \\
d) & 2011-11-29 08:51 & \BinfoRPT{3}{\SI{23}{\hour}\,51} \Bstart{}\activityStart{Walk} \BinfoDT{\SI{9}{\minute}} \activityStart{Subway} \BinfoDT{\SI{17}{\minute}} \activityEnd{Subway} & $22$ & $18$ \\
 & & \hspace*{10ex} \BinfoDT{\SI{0}{\minute}} \activityStart{Walk} \BinfoDT{\SI{5}{\minute}}  \activityEnd{Walk} \BinfoDT{\SI{0}{\minute}} \activityStart{Work}\Bend{} & & \\ [.2em]
e) & 2012-05-28 16:10 & \BinfoRPT{3}{\SI{1}{\day}\,\SI{3}{\minute}} \Bstart{}\activityEnd{Consulting-E} \BinfoDT{\SI{0}{\minute}} \activityStart{Bike} \BinfoDT{\SI{15}{\minute}} \activityEnd{Bike} & $14$ & $12$ \\
 & & \hspace*{10ex} \BinfoDT{\SI{0}{\minute}} \activityStart{Consulting}\Bend{} & & \\ [.8em]
\multicolumn{5}{c}{ \dstSachaG{15}} \\ 
\midrule
f) & 2015-01-08 08:45 & \BinfoRPT{14}{\SI{7}{\day}} \Bstart{}\activityStart{Subway} \BinfoDT{\SI{45}{\minute}} \activityEnd{Subway} \BinfoDT{\SI{0}{\minute}} \activityStart{Consulting-E}\Bend{} & $26$ & $42$ \\
g) & 2016-01-18 17:45 & \BinfoRPT{17}{\SI{1}{\day}} \Bstart{}\activityStart{Dinner} \BinfoDT{\SI{30}{\minute}} \activityEnd{Dinner}\Bend{} & $54$ & $34$ \\ 
h) & 2014-12-18 00:15 & \BinfoRPT{76}{\SI{1}{\day}} \Bstart{}\activityStart{Sleep} \BinfoDT{\SI{8}{\hour}\,30} \activityEnd{Sleep}\Bend{} & $517$ & $152$ \\
i) & 2012-03-29 16:45 & \BinfoRPT{7}{\SI{217}{\day}} \Bstart{}\activityEnd{Consulting-E} \BinfoDT{\SI{0}{\minute}} \activityStart{Subway}\Bend{} & $12$ & $14$ \\
 [.8em]
\multicolumn{5}{c}{ \dstSachaG{60}} \\ 
\midrule
j) & 2011-11-27 21:30 & \BinfoRPT{968}{\SI{1}{\day}} \Bstart{}\activityStart{Sleep}\Bend{} & $2157$ & $968$ \\
k) & 2011-11-28 08:30 & \BinfoRPT{4}{\SI{11}{\hour}} \Bstart{}\BinfoRPT{4}{\SI{7}{\day}} \Bstart{}\activityIns{Walk}\Bend{}\Bend{} & $8$ & $16$ \\
l) & 2015-10-24 23:30 & \BinfoRPT{22}{\SI{1}{\day}} \Bstart{}\activityStart{VideoGame-B2} \BinfoDT{\SI{1}{\hour}} \activityEnd{VideoGame-B2}\Bend{} & $137$ & $44$ \\
 [.8em]
\multicolumn{5}{c}{ \dstSachaAR{\iRel}} \\ 
\midrule
m) & $23200$ & \BinfoRPT{3460}{3} \Bstart{}\activity{Childcare}\Bend{} & $5879$ & $3460$ \\
n) & $862$ & \BinfoRPT{237}{12} \Bstart{}\activity{Sleep}\Bend{} & $775$ & $237$ \\ 
o) & $33140$ & \BinfoRPT{3}{155} \Bstart{}\BinfoRPT{4}{1} \Bstart{}\activity{Consulting-E}\Bend{}\Bend{} & $1$ & $12$ \\ 
p) & $7091$ & \BinfoRPT{3}{14207} \Bstart{}\activity{Emacs} \BinfoDT{445} \BinfoRPT{5}{2} \Bstart{}\activity{Coding}\Bend{}\Bend{} & $29$ & $18$ \\
\bottomrule
\end{tabular}}
\end{table}

\begin{table}
\caption{Example patterns from the \dstTZap{0} sequences.}
\label{tab:res-ex-zap}
\centering
{
\begin{tabular}{@{\hspace{1ex}}c@{\hspace{2ex}}l@{\hspace{2ex}}p{.76\textwidth}@{\hspace{3ex}}r@{\hspace{1ex}}r@{\hspace{1ex}}}
\toprule
 & $\Cto$ & \multicolumn{2}{l}{$\Ptree$ \hfill $\sum\abs{\Csc}$} & $\abs{\opoccs}$ \\
\midrule
a) & $36060$ & \BinfoRPT{110}{2} \Bstart{}\activity{1561:X} \BinfoDT{1} \activity{1561:E}\Bend{} & $80$ & $220$ \\
b) & $33415$ & \BinfoRPT{20}{8} \Bstart{}\activity{1561:I} \BinfoDT{1} \activity{1561:i} \BinfoDT{1} \activity{1561:Ix} \BinfoDT{1} \activity{1561:C} \BinfoDT{1} \activity{53:C}\Bend{} & $63$ & $100$ \\
c) & $11680$ & \BinfoRPT{3}{5116} \Bstart{}\BinfoRPT{8}{1} \Bstart{}\activity{2429:U} \BinfoDT{3} \activity{2429:u}\Bend{}\Bend{} & $40$ & $48$ \\
d) & $7908$ & \BinfoRPT{3}{17729} \Bstart{}\BinfoRPT{5}{2} \Bstart{}\activity{2400:E} \BinfoDT{1} \activity{2400:X}\Bend{} & $74$ & $60$ \\
 & & \hspace*{10ex} \BinfoDT{91} \BinfoRPT{5}{2} \Bstart{}\activity{2400:E} \BinfoDT{1} \activity{2400:X}\Bend{}\Bend{} & & \\ [.2em]
e) & $84347$ & \BinfoRPT{3}{10563} \Bstart{}\activity{2399:U} \BinfoDT{1} \BinfoRPT{4}{2} \Bstart{}\activity{2399:C} \BinfoDT{1} \activity{2427:C}\Bend{}\Bend{} & $3$ & $27$ \\
f) & $85889$ & \BinfoRPT{7}{248} \Bstart{}\BinfoRPT{4}{2} \Bstart{}\activity{2400:X}\Bend{} \BinfoDT{7} \activity{2400:C}\Bend{} & $48$ & $35$ \\
g) & $104793$ & \BinfoRPT{3}{17790} \Bstart{}\BinfoRPT{5}{6} \Bstart{}\activity{2445:C}\Bend{} \BinfoDT{3} \BinfoRPT{4}{8} \Bstart{}\activity{2447:C}\Bend{}\Bend{} & $15$ & $27$ \\
h) & $126101$ & \BinfoRPT{5}{253} \Bstart{}\activity{2426:C} \BinfoDT{3} \activity{18:C} \BinfoDT{3} \activity{2445:U} \BinfoDT{1} \activity{2445:u} \BinfoDT{1} \activity{2445:C} & $15$ & $35$ \\
 & & \hspace*{10ex} \BinfoDT{3} \activity{2447:C} \BinfoDT{21} \activity{2447:C}\Bend{} & & \\ [.2em]
i) & $151772$ & \BinfoRPT{4}{221} \Bstart{}\activity{6:C} \BinfoDT{2} \BinfoRPT{4}{2} \Bstart{}\activity{2395:X}\Bend{}\Bend{} & $15$ & $20$ \\
j) & $12071$ & \BinfoRPT{3}{2235} \Bstart{}\BinfoRPT{4}{2} \Bstart{}\activity{2395:X} \BinfoDT{1} \activity{2395:E}\Bend{} \BinfoDT{7} \BinfoRPT{4}{6} \Bstart{}\activity{2395:C}\Bend{}\Bend{} & $76$ & $36$ \\
\bottomrule
\end{tabular}}
\end{table}

\FloatBarrier
\newpage
\section*{List of Symbols}
\label{los}
\vfill
\begin{tabular}{@{\hspace{2ex}}l@{\hspace{4ex}}p{.72\textwidth}@{\hspace{4ex} p.}c@{\hspace{2ex}}}
$\ABC$ & event alphabet & \pageref{sym:ABC} \\
$\alpha$ & an event &\pageref{sym:alpha} \\
$\seq$ & an event sequence & \pageref{sym:seq} \\
$\seq[\alpha]$ & projection of sequence $\seq$ on event $\alpha$ & \pageref{sym:seqalpha} \\
$\len{\seq}$ & \emph{length} of sequence $\seq$, number of timestamp--event pairs in $\seq$ & \pageref{sym:lenS} \\
$\tSstart(\seq)$ & smallest timestamp in $\seq$ & \pageref{sym:tSstart} \\
$\tSend(\seq)$ & largest timestamp in $\seq$ & \pageref{sym:tSend} \\
$\tspan{\seq}$ & \emph{duration} of sequence $\seq$, time spanned by $\seq$ & \pageref{sym:durationS} \\
$\cycle$ & an event cycle & \pageref{sym:cycle} \\
$\Cev$ & \emph{cycle event} & \pageref{sym:Cev} \\
$\Clen$ & \emph{cycle length} & \pageref{sym:Clen} \\
$\Cprd$ & \emph{cycle period} & \pageref{sym:Cprd} \\
$\Cto$ & \emph{cycle starting point} & \pageref{sym:Cto} \\
$\Csc$ & \emph{cycle shift corrections} & \pageref{sym:Csc} \\
$\tspan{C}$ & \emph{duration} of cycle $C$, time spanned by $C$ & \pageref{sym:Cspan} \\
$\sumel{\Csc}$ & sum of the shift corrections in $\Csc$ & \pageref{sym:sumel} \\
$\cov{\cycle}$ & \emph{cover} of cycle $\cycle$, set of timestamp--event pairs reconstructed from $\cycle$ & \pageref{sym:cov} \\
$\ccycle$ & a collection of cycles & \pageref{sym:ccycle} \\
$\residual{\ccycle, \seq}$ & set of \emph{residuals}, timestamp--event pairs of sequence $\seq$ not covered by any cycle in the collection of cycles $\ccycle$ & \pageref{sym:residual} \\
$\cl$ & \emph{cost}, code length & \pageref{sym:costCL} \\
$\patt$ & a periodic pattern & \pageref{sym:patt} \\
$\Ptree$ & pattern tree & \pageref{sym:ptree} \\
$\Pblock_{\Bid}$ & a bock in a periodic pattern & \pageref{sym:Pblock} \\
$\child{\Pblock_{\Bid}}$ & ordered list of children of block $\Pblock_{\Bid}$ & \pageref{sym:child} \\
$d_{\Bid{}i}$ & \emph{inter-block distance}, time separating occurences of blocks $\Pblock_{\Bid{}(i-1)}$ and $\Pblock_{\Bid{}i}$ & \pageref{sym:interBd} \\
$\Lchild{X}$ & left-most leaf descendant of block/node $X$ & \pageref{sym:Lchild} \\
$\shift{S}{t_s}$ & function that shifts sequence $S$ forward by $t_s$ & \pageref{sym:shift} \\
$\occsStar{\patt}$ & list of timestamp--event pairs reconstructed from the pattern tree of $\patt$ prior to correction, a.k.a.\ perfect occurences & \pageref{sym:occStar} \\
$\occs{\patt}$ & list of timestamp--event pairs reconstructed from the pattern tree of $\patt$ after correction, a.k.a.\ corrected occurences & \pageref{sym:occs} \\
$\cume{o}$ & the cumulated time correction to be applied to timestamp--event pair $o$ & \pageref{sym:cume} \\
$\evtseq$ & the string representing the event sequence of a block/node & \pageref{sym:evtseq} \\
$\tspanStar{\Pblock_\Bid{}}$ & time spanned by the entire cycle of block $\Pblock_\Bid{}$ & \pageref{sym:tspanStar} \\
$\tspanRepStar{\Pblock_\Bid{}}$ & time spanned by a single repetition of block $\Pblock_\Bid{}$ & \pageref{sym:tspanRepStar} \\
$\maxtspanStar{\Pblock_{\Bid}}$ & maximum time span of the entire cycle of block $\Pblock_\Bid{}$ & \pageref{sym:maxtspanStar} \\
$\maxtspanRepStar{\Pblock_{\Bid}}$ & maximum time span of a repetition of block $\Pblock_\Bid{}$ & \pageref{sym:maxtspanRepStar} \\
$\inDP$ & collection of all the periods (except of the top block) and inter-block distances in the pattern tree, as well as $\optspanRep^{*}_{\max}$ of top block, if necessary & \pageref{sym:inDP} \\
$\prcCl{}$ & \emph{compression ratio}, ratio of the sequence code length using the considered collection of patterns vs.\ using an empty collection of patterns & \pageref{sym:prcCl} \\
$\resSet$ & set of residuals & \pageref{sym:resSet} \\
$\ratioClR$ & fraction of the code length spent on residuals & \pageref{sym:ratioClR} \\
\end{tabular}
\vfill

\end{document}